\documentclass [12pt]{amsart}
\usepackage{amsmath}
\usepackage{amssymb}
\usepackage{amsthm}
\usepackage{amscd}
\usepackage{amsthm}
\usepackage{amsopn}
\usepackage{xspace}
\usepackage{verbatim}
\usepackage{amsmath}
\usepackage{amsfonts}
\usepackage{epic}
\usepackage{eepic}
\usepackage[active]{srcltx}
\usepackage{euscript}
\usepackage{epsfig}
\usepackage{color}
\definecolor{red}{rgb}{1,0,0}
\definecolor{green}{rgb}{0,1,0}
\definecolor{SeaGreen}{RGB}{46,139,87}
\definecolor{Maroon}{RGB}{128,0,0}

\definecolor{olive}{named}{OliveGreen}
\definecolor{forest}{named}{ForestGreen}
\hoffset=- 2cm \voffset=- 0cm 
\def\qed{\hbox {\hskip 1pt \vrule width 4pt height 6pt depth 1.5pt
        \hskip 1pt}\\}

\newcommand{\N}{\mathbb{N}}

\newcommand{\C}{{\mathbb{C}}}
\newcommand{\R}{{\mathbb{R}}}

\newcommand{\A}{\mathcal A}

\newcommand{\D}{\mathcal D}
\def\Dg {{\mathcal D}}
\newcommand{\F}{\mathcal F}
\newcommand{\FF}{\EuScript F}

\def\Hg {{\mathcal H}}

\newcommand{\LL}{\mathcal L}
\def\Lg {{\mathfrak L}}
\newcommand{\OO}{\mathcal O}
\def\Og {\mathcal O}

\newcommand{\RR}{\mathcal R}

\def\Sg {{\mathcal S}}

\def\Wg {{\mathcal W}}
\def\Vg {{\mathcal V}}
\def\Tg {{\mathcal T}}

\def\curl{\text{\rm curl}}

\def\mag{\text{\rm mag}}
\def\curl{\text{\rm curl\,}}
\def\det{\text{\rm det\,}}
\def\Div{\text{\rm div\,}}
\def\grad{\text{\rm grad\,}}

\newcommand{\supp}{\rm Supp\,}

\renewcommand {\Re}{{\rm Re\,}}
\renewcommand{\Im}{{\rm Im\,}}
\newcommand {\pa}{\partial}

\newcommand {\ar}{\to}

\def\0{\mathbf  0}

\newcommand{\eq}{\begin{equation}}
\newcommand{\eeq}{\end{equation}}

\def\Xint#1{\mathchoice
{\XXint\displaystyle\textstyle{#1}}%
{\XXint\textstyle\scriptstyle{#1}}%
{\XXint\scriptstyle\scriptscriptstyle{#1}}%
{\XXint\scriptscriptstyle\scriptscriptstyle{#1}}%
\!\int}
\def\XXint#1#2#3{{\setbox0=\hbox{$#1{#2#3}{\int}$ }
\vcenter{\hbox{$#2#3$ }}\kern-.6\wd0}}

\def\dashint{\Xint-}

 

\errorcontextlines=0 \numberwithin{equation}{section}
\theoremstyle{plain}
\newtheorem{theorem}{Theorem}[section]
\newtheorem{lemma}[theorem]{Lemma}
\newtheorem{proposition}[theorem]{Proposition}

\newtheorem{remark}[theorem]{Remark}
\newtheorem{corollary}[theorem]{Corollary}

\begin{document}
\bibliographystyle{siam}

\title[Global stability]
{Global stability of the normal state of superconductors in the
  presence of a strong electric current}
\author[Y. Almog, B. Helffer]{Yaniv Almog and Bernard Helffer}

\address{ Department of Mathematics, Louisiana State University,
    Baton Rouge, LA 70803, USA.}
\email{almog@math.lsu.edu}

\address{Laboratoire de Math\'ematiques, Universit\'e Paris-Sud 11 et CNRS,
B\^at 425, 91405 Orsay Cedex, France.}
\email{Bernard.Helffer@math.u-psud.fr}

\keywords{superconductivity, critical current, critical magnetic
field, Ginzburg-Landau}

\subjclass{82D55, 35B25, 35B40, 35Q55}

\begin{abstract}
  We consider the time-dependent Ginzburg-Landau model of
  superconductivity in the presence of an electric current flowing
  through a two-dimensional wire. We show that when the current is
  sufficiently strong the solution converges in the long-time limit to
  the normal state. We provide two types of upper bounds for the
  critical current where such global stability is achieved: by using
  the principal eigenvalue of the magnetic Laplacian associated with
  the normal magnetic field, and through the norm of the resolvent of
  the linearized steady-state operator. In the latter case we estimate
  the resolvent norm in large domains by the norms of approximate
  operators defined on the plane and the half-plane. We also obtain a
  lower bound, in large domains, for the above critical current by
  obtaining the current for which the normal state looses its local stability.
\end{abstract}
\maketitle
\section{Introduction}
\label{sec:1}
Consider a superconductor placed at a temperature lower than the
critical one. It is well-understood from numerous experimental
observations, that a sufficiently strong current, applied through the
sample, will force the superconductor to arrive at the normal state.
To explain this phenomenon mathematically, we use  the
time-dependent Ginzburg-Landau model which is defined by the following
system of equations, and will be referred to as (TDGL1)
\begin{subequations}  
\label{eq:1}
\begin{alignat}{2}
& \frac{\partial \psi}{\partial t} +
  i\phi\psi = \left(\nabla - iA \right)^{2} \psi + \psi 
  \left( 1 - |\psi|^{2} \right)\,,& \quad \text{ in } \R_+\times\Omega\, ,\\
 & \kappa^2\curl^2A + \sigma \left(\frac{\partial A}{\partial t} +
 \nabla\phi\right)  =  \Im(\bar\psi\, \nabla_A\psi) \, , & \quad \text{ in }  \R_+\times\Omega\,,\\
  &\psi=0 \,,&\quad \text{ on }  \R_+\times\partial\Omega_c\,, \\
 &(i\nabla+A)\psi\cdot\nu=0 \,,& \quad \text{ on }  \R_+\times\partial\Omega_i\,,\\
 &\sigma \left(\frac{\partial A}{\partial t} +
 \nabla\phi\right)\cdot\nu = J \,,&\quad \text{ on } \R_+\times\partial\Omega_c\,,\\
&\sigma \left(\frac{\partial A}{\partial t} +
 \nabla\phi\right)\cdot\nu=0\,,
 &\quad \text{ on }  \R_+\times\partial\Omega_i \,, \\[1.2ex]
&\dashint_{\partial\Omega}\curl A(t,x)\,ds = h_{ex}\,, & \text{ on } \R_+\,, \\
&\psi(0,x)=\psi_0(x) \,, & \quad \text{ in } \Omega\,, \\ 
&A(0,x)=A_0(x)\,, & \quad \text{ in } \Omega \,.
\end{alignat}
\end{subequations}
In the above $\psi$ denotes the order parameter, $A$ is the magnetic
potential, $\phi$ is the electric potential, $\kappa$ denotes the
Ginzburg-Landau parameter, which is a material property, and the normal
conductivity of the sample is denoted by $\sigma$.  We use the notation
$\nabla_A=\nabla-iA$ and $ds$ for the induced measure on $\pa \Omega$. In
(\ref{eq:1}g) we use the standard notation
\begin{displaymath}
  \dashint_D  = \frac{1}{|D|}\int \,,
\end{displaymath}
for any domain $D\subset\R^2$.  The spatial coordinates have been scaled in
(\ref{eq:1}) by the coherence length, characterizing variations in
$\psi$. The domain $\Omega\subset\subset\R^2$, occupied by the superconducting sample,
has a smooth interface, denoted by $\partial\Omega_c$, with a conducting metal
which is at the normal state. Thus, we require that $\psi$ would vanish
on $\partial\Omega_c$ in (\ref{eq:1}c), and allow for a smooth current
satisfying
\begin{equation}
\label{eq:75}
(J1)\quad J\in C^2(\overline{\partial\Omega_c}),
\end{equation}
to enter the sample in (\ref{eq:1}e). We further require that
\begin{equation}
  \label{eq:76}
(J2) \quad \int_{\partial\Omega_c}J \,ds=0 \,,
\end{equation}
and, mainly in the last section, that
\begin{equation}\label{eq:75aa}
(J3)\quad  \mbox{the sign of } J \mbox{ is constant on each
connected component of } \partial\Omega_c\,.
\end{equation}
We allow for $J\neq0$ at the corners despite the fact that no current is
allowed to enter the sample through the insulator. \\

The rest of the boundary, denoted by $\partial\Omega_i$ is adjacent to an
insulator.
To simplify some of our regularity arguments we introduce the
 following geometrical assumption  (for further discussion we refer the
 reader to Appendix \ref{sec:a}) on $\partial\Omega$:
\begin{equation}\label{propertyR}
(R1)\,\left\{
\begin{array}{l}
(a) \; \pa \Omega_i \mbox{ and } \pa \Omega_c  \mbox{ are of class } C^3\,;\\
(b)\; \ \mbox{ Near each edge, } \pa
\Omega_i \mbox{  and  } \pa \Omega_c \mbox{ are flat and meet with  an angle of } \frac \pi 2\,.
\end{array}\right.
\end{equation}
  We also require in the last section  that:
\begin{equation}\label{hyptopolo}
(R2) \quad\quad \mbox{Both }  \partial\Omega_c \mbox{  and } \partial\Omega_i \mbox{  have two components}.
\end{equation}

The average field along the boundary is given by the constant
$h_{ex}$. Figure 1 presents a typical sample with properties (R1) and
(R2), where the current flows into the sample from one connected
component of $\partial\Omega_c$, and exits from another part, disconnected from
the first one. Most wires would fall into the above class of domains.
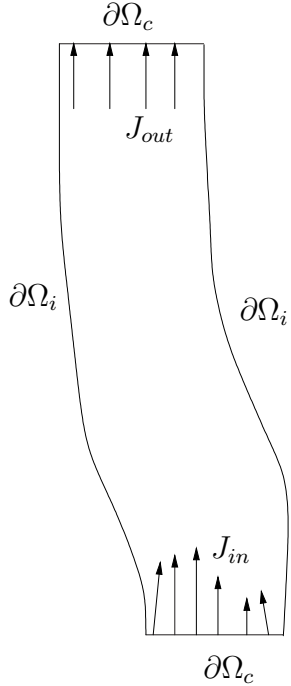
\begin{figure}
  \centering
\setlength{\unitlength}{0.0005in}
\begingroup\makeatletter\ifx\SetFigFont\undefined%
\gdef\SetFigFont#1#2#3#4#5{%
  \reset@font\fontsize{#1}{#2pt}%
  \fontfamily{#3}\fontseries{#4}\fontshape{#5}%
  \selectfont}%
\fi\endgroup%
{\renewcommand{\dashlinestretch}{30}\begin{picture}(2766,6795)(0,-10)
\path(1275,300)(2700,300)
\path(375,6450)(1875,6450)
\path(1575,300)(1575,1125)
\blacken\path(1605.000,1005.000)(1575.000,1125.000)(1545.000,1005.000)(1605.000,1005.000)
\path(1800,300)(1800,1200)
\blacken\path(1830.000,1080.000)(1800.000,1200.000)(1770.000,1080.000)(1830.000,1080.000)
\path(2025,300)(2025,900)
\blacken\path(2055.000,780.000)(2025.000,900.000)(1995.000,780.000)(2055.000,780.000)
\path(2325,300)(2325,675)
\blacken\path(2355.000,555.000)(2325.000,675.000)(2295.000,555.000)(2355.000,555.000)
\path(1350,300)(1425,1050)
\blacken\path(1442.911,927.610)(1425.000,1050.000)(1383.208,933.581)(1442.911,927.610)
\path(2550,300)(2475,750)
\blacken\path(2524.320,636.565)(2475.000,750.000)(2465.136,626.701)(2524.320,636.565)
\path(525,5775)(525,6450)
\blacken\path(555.000,6330.000)(525.000,6450.000)(495.000,6330.000)(555.000,6330.000)
\path(900,5775)(900,6450)
\blacken\path(930.000,6330.000)(900.000,6450.000)(870.000,6330.000)(930.000,6330.000)
\path(1275,5775)(1275,6450)
\blacken\path(1305.000,6330.000)(1275.000,6450.000)(1245.000,6330.000)(1305.000,6330.000)
\path(1575,5775)(1575,6450)
\blacken\path(1605.000,6330.000)(1575.000,6450.000)(1545.000,6330.000)(1605.000,6330.000)
\path(1275,300)(1275,301)(1275,304)
	(1275,312)(1275,326)(1275,346)
	(1274,371)(1274,400)(1273,431)
	(1273,461)(1272,491)(1271,518)
	(1271,544)(1270,567)(1269,589)
	(1267,609)(1266,627)(1264,645)
	(1263,662)(1260,680)(1258,697)
	(1256,714)(1253,732)(1249,751)
	(1246,770)(1242,789)(1237,810)
	(1232,830)(1227,851)(1221,872)
	(1216,894)(1209,915)(1203,937)
	(1196,959)(1190,980)(1182,1002)
	(1175,1025)(1169,1042)(1163,1060)
	(1157,1079)(1150,1099)(1143,1119)
	(1135,1140)(1127,1162)(1119,1185)
	(1110,1208)(1101,1232)(1092,1257)
	(1082,1282)(1073,1307)(1063,1333)
	(1053,1358)(1043,1383)(1032,1409)
	(1022,1433)(1012,1458)(1002,1482)
	(992,1506)(982,1529)(972,1552)
	(963,1575)(952,1598)(942,1621)
	(932,1644)(921,1668)(910,1692)
	(899,1717)(887,1742)(876,1768)
	(864,1793)(852,1819)(840,1845)
	(828,1872)(817,1897)(805,1923)
	(794,1948)(783,1973)(773,1998)
	(763,2021)(753,2044)(744,2066)
	(736,2088)(728,2109)(720,2130)
	(713,2150)(705,2172)(698,2194)
	(691,2215)(684,2237)(677,2260)
	(671,2283)(665,2306)(659,2329)
	(653,2353)(648,2377)(643,2402)
	(638,2426)(633,2451)(629,2475)
	(624,2499)(620,2523)(617,2546)
	(613,2570)(610,2593)(606,2616)
	(603,2639)(600,2662)(597,2684)
	(594,2706)(591,2729)(588,2753)
	(585,2777)(582,2802)(579,2828)
	(576,2855)(573,2882)(569,2910)
	(566,2939)(562,2968)(559,2997)
	(556,3026)(552,3055)(549,3084)
	(546,3112)(543,3140)(540,3168)
	(537,3195)(534,3222)(531,3248)
	(528,3274)(525,3300)(522,3324)
	(520,3348)(517,3372)(514,3397)
	(511,3422)(509,3447)(506,3473)
	(503,3500)(500,3527)(497,3554)
	(494,3582)(491,3610)(487,3638)
	(484,3665)(481,3693)(478,3721)
	(475,3748)(472,3775)(469,3802)
	(466,3828)(464,3853)(461,3878)
	(458,3903)(455,3927)(453,3951)
	(450,3975)(447,3999)(445,4023)
	(442,4047)(439,4072)(437,4097)
	(434,4123)(431,4149)(428,4176)
	(425,4204)(423,4232)(420,4260)
	(417,4289)(414,4318)(412,4347)
	(409,4376)(406,4405)(404,4435)
	(402,4464)(400,4492)(397,4521)
	(395,4549)(394,4577)(392,4605)
	(390,4632)(389,4660)(388,4687)
	(386,4713)(385,4740)(384,4767)
	(383,4794)(382,4822)(381,4851)
	(381,4880)(380,4910)(379,4941)
	(379,4972)(378,5004)(378,5037)
	(377,5069)(377,5102)(376,5134)
	(376,5167)(376,5200)(376,5232)
	(375,5264)(375,5295)(375,5326)
	(375,5356)(375,5385)(375,5414)
	(375,5443)(375,5470)(375,5498)
	(375,5525)(375,5556)(375,5588)
	(375,5619)(375,5651)(375,5683)
	(375,5715)(375,5747)(375,5780)
	(375,5812)(375,5844)(375,5876)
	(375,5908)(375,5938)(375,5968)
	(375,5997)(375,6025)(375,6051)
	(375,6076)(375,6100)(375,6123)
	(375,6144)(375,6164)(375,6182)
	(375,6200)(375,6228)(375,6254)
	(375,6278)(375,6300)(375,6323)
	(375,6345)(375,6367)(375,6389)
	(375,6409)(375,6426)(375,6439)
	(375,6446)(375,6450)
\path(2704,291)(2704,293)(2704,299)
	(2705,309)(2706,324)(2707,346)
	(2709,374)(2711,409)(2713,449)
	(2715,495)(2718,545)(2721,599)
	(2724,654)(2727,711)(2730,767)
	(2732,823)(2735,877)(2738,930)
	(2740,980)(2742,1028)(2744,1073)
	(2746,1115)(2748,1155)(2749,1193)
	(2750,1228)(2751,1261)(2752,1293)
	(2753,1322)(2753,1350)(2754,1377)
	(2754,1403)(2754,1428)(2754,1464)
	(2753,1498)(2753,1531)(2752,1563)
	(2750,1595)(2749,1626)(2747,1656)
	(2745,1685)(2742,1714)(2739,1742)
	(2736,1768)(2733,1794)(2729,1819)
	(2726,1842)(2722,1864)(2718,1885)
	(2714,1905)(2709,1924)(2705,1942)
	(2701,1959)(2696,1975)(2692,1991)
	(2686,2010)(2680,2029)(2673,2048)
	(2666,2067)(2659,2086)(2651,2106)
	(2643,2126)(2634,2147)(2626,2168)
	(2616,2189)(2607,2209)(2598,2230)
	(2588,2251)(2579,2272)(2570,2292)
	(2560,2312)(2551,2333)(2542,2353)
	(2534,2371)(2526,2389)(2517,2408)
	(2508,2427)(2500,2447)(2490,2468)
	(2481,2490)(2471,2513)(2461,2536)
	(2450,2560)(2440,2584)(2429,2608)
	(2419,2633)(2408,2657)(2398,2681)
	(2387,2705)(2377,2729)(2367,2752)
	(2357,2775)(2348,2797)(2338,2819)
	(2329,2841)(2320,2861)(2312,2881)
	(2303,2902)(2294,2923)(2285,2944)
	(2276,2966)(2266,2989)(2257,3011)
	(2247,3035)(2238,3059)(2228,3083)
	(2218,3107)(2209,3131)(2199,3156)
	(2190,3180)(2181,3205)(2172,3229)
	(2163,3253)(2155,3276)(2147,3299)
	(2139,3323)(2131,3345)(2124,3368)
	(2117,3391)(2109,3414)(2102,3438)
	(2095,3461)(2088,3486)(2081,3511)
	(2074,3537)(2068,3563)(2061,3590)
	(2054,3617)(2047,3645)(2041,3672)
	(2034,3700)(2028,3728)(2022,3755)
	(2016,3782)(2011,3808)(2006,3834)
	(2001,3858)(1997,3883)(1993,3906)
	(1989,3928)(1985,3950)(1982,3971)
	(1979,3991)(1976,4011)(1974,4030)
	(1971,4050)(1969,4070)(1967,4090)
	(1965,4110)(1963,4131)(1962,4152)
	(1960,4174)(1958,4196)(1957,4219)
	(1956,4242)(1954,4266)(1953,4291)
	(1952,4316)(1951,4341)(1950,4367)
	(1949,4394)(1948,4421)(1947,4448)
	(1945,4477)(1944,4505)(1943,4535)
	(1942,4566)(1940,4592)(1939,4618)
	(1938,4646)(1936,4675)(1935,4705)
	(1933,4736)(1932,4768)(1930,4801)
	(1928,4835)(1926,4870)(1924,4906)
	(1922,4942)(1920,4979)(1918,5016)
	(1916,5054)(1915,5091)(1913,5128)
	(1911,5165)(1909,5201)(1907,5237)
	(1905,5272)(1903,5306)(1901,5339)
	(1900,5371)(1898,5402)(1897,5432)
	(1895,5461)(1894,5489)(1893,5515)
	(1892,5541)(1890,5575)(1889,5607)
	(1887,5638)(1886,5668)(1885,5698)
	(1884,5726)(1883,5754)(1883,5781)
	(1882,5808)(1881,5833)(1881,5858)
	(1880,5881)(1880,5904)(1880,5925)
	(1879,5946)(1879,5965)(1879,5984)
	(1879,6001)(1879,6018)(1879,6034)
	(1879,6050)(1879,6066)(1879,6085)
	(1879,6105)(1879,6125)(1879,6146)
	(1879,6168)(1879,6193)(1879,6219)
	(1879,6247)(1879,6278)(1879,6309)
	(1879,6340)(1879,6369)(1879,6394)
	(1879,6415)(1879,6429)(1879,6437)
	(1879,6440)(1879,6441)
\put(1875,-200){\makebox(0,0)[lb]{{\SetFigFont{12}{14.4}{\rmdefault}{\mddefault}{\updefault}$\partial\Omega_c$}}}
\put(825,6600){\makebox(0,0)[lb]{{\SetFigFont{12}{14.4}{\rmdefault}{\mddefault}{\updefault}$\partial\Omega_c$}}}
\put(2250,3525){\makebox(0,0)[lb]{{\SetFigFont{12}{14.4}{\rmdefault}{\mddefault}{\updefault}$\partial\Omega_i$}}}
\put(-150,3675){\makebox(0,0)[lb]{{\SetFigFont{12}{14.4}{\rmdefault}{\mddefault}{\updefault}$\partial\Omega_i$}}}
\put(1950,1050){\makebox(0,0)[lb]{{\SetFigFont{12}{14.4}{\rmdefault}{\mddefault}{\updefault}$J_{in}$}}}
\put(1050,5425){\makebox(0,0)[lb]{{\SetFigFont{12}{14.4}{\rmdefault}{\mddefault}{\updefault}$J_{out}$}}}
\end{picture}
}
  \caption{Typical superconducting sample. The arrows denote the direction of the current flow ($J_{in}$  for the inlet, and $J_{out}$ for the outlet).}
  \label{fig:1}
\end{figure}
 We assume, for the
initial conditions (\ref{eq:1}h,i), that 
\begin{equation}\label{init}
\psi_0\in H^1(\Omega,\C) \mbox{  and }
A_0\in H^1(\Omega,\mathbb R^2)\,.
\end{equation}
 We further assume everywhere in the sequel that:
\begin{equation}\label{eq:2.2}
\|\psi_0\|_\infty \leq 1\,.
\end{equation}
In most of this work we consider Coulomb gauge solutions of
\eqref{eq:1} which satisfy in addition
\begin{equation}
\label{eq:178}
   \Div A =0\,,\, A\cdot  \nu_{/\pa \Omega} =0\,.
\end{equation}

To complete the presentation of the problem, we need to make two
further assumptions on the normal magnetic and electric potentials
which we respectively denote by $A_n$ and $\phi_n$. To this end, we write that
$(0, h A_n,h \phi_n)$ is a stationary  solution of (\ref{eq:1}), where $h$
is a positive parameter representing the intensity of the  applied field. More
explicitly, $(A_n,\phi_n)$ must satisfy
\begin{displaymath}
  \begin{cases}
     - c\, \curl^2A_n + \nabla\phi_n = 0 & \text{in } \Omega\,, \\
     -\sigma\frac{\partial\phi_n}{\partial\nu} = J_r  & \text{on } \partial\Omega\,, \\
     \dashint_{\partial\Omega}\curl A_n\,ds = h_{r}\,,
  \end{cases}
\end{displaymath}
in which $c=\kappa^2/\sigma$ and
$$J_r=J/h\,,  \mbox{ and } h_r=h_{ex}/h$$
respectively denote some reference electric current and magnetic
field. (Obviously, $J_r\equiv0$ on $\partial\Omega_i$.) We fix the Coulomb gauge for
$A_n$, i.e., we require that it satisfies \eqref{eq:178}.  In the next
section we discuss the existence, uniqueness, and regularity of
solutions to the above problem.

The next assumption will be
rephrased in the next section. Here we write
\begin{equation}
  (B) \quad B_n \neq 0 \text{ at the corners}\,,
\end{equation}
where $B_n=\curl A_n$.\\
In the last section we restate that
\begin{equation}\label{(C)}
  (C) \quad \nabla\phi_n\perp\partial\Omega \text{ on } B_n^{-1}(0)\cap\partial\Omega\,. 
\end{equation}
The  reasons for the above assumptions will be clarified in
the sections where they are restated.

One possible mechanism which contributes to the breakdown of
superconductivity by a strong current is the magnetic field induced by
the current. In the absence of electric current, it was proved by
Giorgi \& Phillips in \cite{giph99} that, when a sufficiently strong
magnetic field is applied on the sample's boundary (or when $h$
is sufficiently large), the normal state, for which $\psi\equiv0$, becomes
the unique solution for the steady-state version of (\ref{eq:1}) (cf.
also \cite{fohe09} and the references therein). For the time-dependent
Ginzburg-Landau equations it was proved in \cite{feta01} that every
solution must reach an equilibrium in the long-time limit.  When
combined with the results in \cite{giph99} it follows that when the
applied magnetic field is sufficiently large the normal state becomes
globally stable.

No such result is available in the presence of electric currents. The
results in \cite{feta01} are based on the fact that, in the absence of
currents, the Ginzburg-Landau energy functional serves as a Lyapunov
functional. In the presence of a current one has to take account of
the work it produces, which does not necessarily decrease the energy
(cf. \cite{ruetal10b} for instance). Moreover, the magnetic field is
not the only mechanism which forces the sample into the normal state
when the electric current is sufficiently large. For a reduced model,
which neglects the magnetic field (induced and applied) effect it has
been proved in \cite{ivko84,ruetal10,al08} that the normal state is at
least locally (linearly) stable when the current is sufficiently
strong. This reduced model can be easily obtained by setting $A\equiv 0$
in (\ref{eq:1}), and has received significant attention in the greater
context of PT symmetric Schr\"odinger operators (cf. \cite{ruetal07},
\cite{bejo12} for instance).

When the magnetic field's effect is accounted for, we can report here
the results of three recent contributions
\cite{aletal10,aletal13,aletal12}, we have obtained together with Pan.
In all of them we consider the linearization of (\ref{eq:1}a) near the
normal state $(0,A_n,\phi_n)$ both in the entire plane \cite{aletal10},
and the half-plane \cite{aletal13,aletal12}. Thus, we have analyzed
the spectrum of two operators associated with the differential
operator $-\partial_x^2 + (i\partial_y-\frac 12 x^2)^2 +icy\,,$ where, as above,
$c=\kappa^2/\sigma$. We define in \cite{aletal10} $\A$ as the maximal
accretive extension in $L^2(\mathbb R^2)$ of this differential
operator initially defined on $C_0^\infty(\mathbb R^2)$ and in
\cite{aletal12,aletal13} $\A_+$ as an unbounded operator in
$L^2(\mathbb R^2_+)$, where $ \R^2_+ = \{ (x,y)\in\R^2 \,: \, y>0 \}
\,, $ using this time Lax-Milgram's Theorem for the associated
sesquilinear form in $H^1_0(\mathbb R^2_+)$
(see below \eqref{eq:98}-\eqref{zzz4} for more details).\\

In \cite{aletal10} we show that the spectrum of $\A$ is empty
and find estimates on its resolvent norm. In contrast, in
\cite{aletal13,aletal12} we show that $\sigma(\A^+)$ is non-empty and
even manage to evaluate the leftmost eigenvalue, in the asymptotic
limits $c\to\infty$ \cite{aletal13} and $c\to0$ \cite{aletal12}.  One can
easily derive from this leftmost eigenvalue the critical current where
the normal state looses its local (linear) stability. In particular,
we show that it tends, as $c\to\infty$ (i.e. in  the small normal conductivity
limit), to the critical value for the reduced model \cite{ivko84}.
This result suggests, once again, that stability is being forced not
only by the magnetic field that the current induces, but also by the
potential term in (\ref{eq:1}a). We conclude the foregoing literature
survey by mentioning a few works considering the motion of vortices under
the action of an electric current \cite{ti10,seti11,peetal13}. \\

In the present contribution we prove global, long-time, stability of
the normal state, as a solution of (\ref{eq:1}), for sufficiently
large currents. In contrast with \cite{aletal10,aletal13,aletal12} we
consider the fully non-linear problem \eqref{eq:1} in a bounded domain
of the type presented in Fig.~1. While the linear analysis in
\cite{aletal10,aletal13,aletal12} provides us with some useful
insights and tools, employed throughout this work, it cannot be easily
applied to obtain long-time stability of the normal state for a wide
class of initial conditions, not necessarily close to the normal state
in any sense. In particular, it is necessary to bound the effect of
the non-linear terms, that are not necessarily small at $t=0$. The
effect of boundaries needs to be taken into account as well. \\

The rest of this contribution is organized as follows:\\

We begin by
dealing with a few preliminaries in Section~\ref{section2}. In
particular, we prove global existence and uniqueness of solutions for
\eqref{eq:1} and obtain their regularity. While these questions have 
previously been addressed (cf. \cite{chenetal93}, \cite{feetal98}, and
\cite{du94} to name just a few references) the fact that the boundary
is not smooth at the corners requires some additional attention. Some
of the results we state in the next section are proved in the
appendices.

In Section \ref{section3} we prove that if the current is strong
enough, the magnetic field it induces forces the semigroup associated
with (\ref{eq:1}) to become asymptotically a contraction. 
Let
\begin{displaymath}
    \mu(h) = \inf_{
    \begin{subarray}{c}
      u\in H^1(\Omega,\C) \\
      u|_{\partial\Omega_c}=0 \;;\;\|u\|_2 =1
    \end{subarray}} \| \nabla_{hA_n}u\|_2^2 \,.
\end{displaymath}
The main result of Section \ref{section3} is the following
 \begin{theorem}
Let $(\psi,A,\phi)$ denote a solution of \eqref{eq:1} and
\eqref{eq:178} satisfying \eqref{eq:2.2}. Then, there exists $\gamma>0$ for which
whenever
   \begin{equation}
\label{eq:179}
\mu(h)>1+ \frac{\gamma}{\kappa^2} + \frac{\gamma^2}{\kappa^4} \,,
  \end{equation}
there exist  $C= C(\Omega,\kappa,c,\|\psi_0\|_2,\|A_0\|_2,h)>0$  and $\lambda_m=
\lambda_m(c,\kappa,\mu(h),\Omega) >0 $, where $c=\kappa^2/\sigma$,\, 
such that, for all $t>0$, 
 we have:
\begin{equation}
\|\psi\|_2 + \|A-hA_n\|_2 +\|\phi-h\phi_n\|_2 \leq Ce^{-\lambda_mt} \,. 
  \end{equation}
Furthermore, there exists $t^{*}(\kappa,c,\|A_0\|_2,\Omega)$ such that
$[t^*+1,+\infty)\ni  t\mapsto \|\psi (t,\cdot) \|_2$ is monotone decreasing. 
  \end{theorem}
  Note that, as is explained at the beginning of Section \ref{section3}, \eqref{eq:179}
  means that the semigroup associated with the linearized version of
  \eqref{eq:1} is a contraction.  The reader is referred to Theorem
  \ref{thm:3.1} and to Proposition \ref{prop:3.4} for the precise
  values of $\gamma$, $\lambda_m$, and $t^*$ in the large $\kappa$ limit.

Let
\begin{displaymath}
  \LL_h = -\nabla_{hA_n}^2 + ih\phi_n \,,
\end{displaymath}
be defined over the domain
\begin{displaymath}
  D(\LL_h) = \{ u\in H^2(\Omega) \, | \; u|_{\partial\Omega_c}=0 \; ; \;
  \nabla u\cdot\nu|_{\partial\Omega_i}=0\,\} \,.
\end{displaymath}
In Section \ref{section4} we prove that a proper bound on the
resolvent of $\LL_h$, which is the elliptic operator in (\ref{eq:1}a)
linearized near $(0,hA_n,h\phi_n)$, obtained over a vertical line in the
complex plane, guarantees global stability of the normal state.  In
particular we show:
\begin{theorem}
\label{thm:1.2}
   Let $\nu\geq0$. There exists $\kappa_0>0$ and $C_1>0$ such that, if  for some $\kappa>\kappa_0$ we
 have
  \begin{equation}
\label{eq:180}
\sup_{\gamma\in\R}\|(\LL_h-i\gamma-\nu)^{-1}\| < 1- \frac{C_1}{\kappa^2} \,,
  \end{equation}
then, any solution of \eqref{eq:1} must satisfy
\begin{equation}
\label{eq:181}
  \int_0^\infty e^{2\nu t}\, \|\psi(t,\cdot)\|_2^2 \,dt <\infty \,.
\end{equation}
\end{theorem}
Unlike \eqref{eq:179}, \eqref{eq:180} does not guarantee that the
semigroup necessarily becomes a contraction in the long-time limit.
The above stability is proved in the large $\kappa$ limit both for
\eqref{eq:1} and, in Section \ref{section5}, for the same system,
scaled with respect to the penetration depth, which is obtained by
applying the transformation $x\to x/\kappa$ in \eqref{eq:1}, (cf.
Proposition~\ref{prop:5.6}).\\

As the resolvent of $\LL_h$ in an arbitrary domain is difficult to
control, we provide, in Section 6, an estimate of its norm for large
values of $h$, which can be applied for either large domains (with
respect to the coherence length), or large $\kappa$ values for penetration
depth scaling.  We show that its norm can be controlled using bounds
derived from two approximate problems, with constant current defined
either in $\R^2$ or in $\R^2_+$ with Dirichlet boundary conditions.
From the resolvent estimates, together with the results in
\cite{aletal10,aletal13,aletal12} we deduce that the critical current,
for which the normal state looses its local stability, can be
approximated by the same critical current obtained for the above
$\R^2_+$ problem.

For a more precise description of the results in Section \ref{section6},  we recall from
\cite{aletal10} and \cite{aletal13} the definitions of these model  operators
in $\R^2$ and $\R^2_+$. Let
\begin{equation}\label{eq:98}
  \A({\mathfrak j},c)=
-\Big(\nabla-i {\mathfrak j} \frac{x^2}{2}\hat{i}_y\Big)^2
+ ic{\mathfrak j} y \,,
\end{equation}
(where $\hat{i}_y$ is a unit vector in the $y$ direction) defined on
$D(\A)$ where
\begin{equation}\label{zzz2}
  D(\A) = \{ u\in L^2(\R^2) \,| \, \A u\in L^2(\R^2) \} \,.
\end{equation}
Let $\A_+({\mathfrak j},c)$ be defined by the same differential
operator defining $\A$ but on the domain 
\begin{equation}\label{eq:112}
D(\A_+) = \{ u\in \tilde V \, : \, \A_{+} u\in L^2(\R^2_+,\mathbb C) \},
\end{equation}
where
\begin{equation}\label{zzz4}
  \tilde V = H_0^{1,\mag}(\mathbb R^{2}_+,\mathbb C)\cap L^2(\mathbb R^{2}_+,\mathbb C;
t\, ds dt)\, .
\end{equation}
Set
\begin{equation}\label{deffrakj}
h|\nabla B_n(z_0)| = {\mathfrak j}(z_0)\,,
\end{equation}
and then define
\begin{equation}\label{zzz5}
    \A(z_0) = \A({\mathfrak j}(z_0),c)  \quad ; \quad   \A_+(z_0) = \A_+({\mathfrak j}(z_0),c)  
\end{equation}
We show in Section \ref{section2} that under all of the above
assumptions: (J), (R), (B), and (C), 
$B_n^{-1}(0)$ is either empty, or else consists of a single curve $\Gamma$
connecting  the two connected components of $\partial\Omega_c$. We denote
the two points of intersection by $z_1$ and $z_2$ and then set

\begin{equation}\label{eq:183a}
   {\mathfrak j}_+ = \inf_{i=1,2}{\mathfrak j}(z_i)\,.
 \end{equation}
We then let
  \begin{equation}\label{eq:183b}
\nu_m({\mathfrak j},c)= \inf_{\lambda\in \sigma(\A_+(\mathfrak j,c))}\Re\lambda\,.
\end{equation}
A straightforward dilation argument, which we detail in Section 6,
shows that
\begin{subequations}
\label{eq:182}
\begin{gather}
  \nu_m({\mathfrak j},c)={\mathfrak j}^{2/3}\nu_m(1,c) \\
  \|\A^{-1}({\mathfrak j},c)\|={\mathfrak j}^{-2/3}\|\A^{-1}(1,c)\| \\
 \sup_{\gamma\in\R} \|(\A_+({\mathfrak j},c)-i\gamma)^{-1}\|= {\mathfrak j}^{-2/3}\sup_{\gamma\in\R} \|(\A_+(1,c)-i\gamma)^{-1}\|
\end{gather}
\end{subequations}

We can now state
\begin{theorem}\label{thm1.3}
Let $\mu_R$ and $\mu_\infty$ be
respectively defined by
  \begin{equation}\label{eq:130a}
   \mu_R=  R^2\inf_{\lambda\in
    \sigma(\LL_{R^3h})}\Re\lambda\, \mbox{ and} \quad \mu_\infty=\liminf_{R\to\infty}\mu_R\,.
  \end{equation}
     Then
\begin{equation}
\label{eq:130}
  \mu_\infty=\lim_{R\to\infty}\mu_R=\nu_m \,,
\end{equation}
with 
\begin{equation}\label{eq:183c}
\nu_m = \nu_m(\mathfrak j_+,c )\,.
\end{equation}

Furthermore, let $\nu<\mu_\infty$.  Then, there exist positive $R_0$ and
$C$, depending only on $\Omega$, $\nu$, and $h$ such that, for $R\geq R_0$,
we have
  \begin{multline}
\label{eq:131}
R^2\sup_{\gamma\in\R}\|(\LL_{R^3h}-R^2\nu-iR^3\gamma)^{-1}\| \leq \\\max
\Big(\sup_{z_0\in\Gamma}\|(\A(z_0)-\nu)^{-1}\|,\sup_{
  \begin{subarray}
\,    \gamma\in\R \\
    i=1,2
  \end{subarray}}\|(\A_+(z_i)-\nu-i\gamma)^{-1}\|\Big)
\Big(1+\frac{C}{R^{1/4}}\Big) \\+ \frac{C}{R^{1/4}} \,.  
  \end{multline}
\end{theorem}
\begin{remark}
  One can deduce from \eqref{eq:131} an upper bound for the critical current
  where the normal state $(0,hA_n,h\phi_n)$ becomes globally stable.    Let
   \begin{equation}
\label{eq:184}
    {\mathfrak j}_m = \inf_{z\in\Gamma}{\mathfrak j}(z)\,,
  \end{equation}
  and let ${\mathfrak j}_+$ be defined by (\ref{eq:183a}).  As is
  proved in Section \ref{section6}, and in particular in
  \eqref{eq:174}, when the domain size is multiplied
  by $R$, the resolvent norm of $\LL_h$ is given by the left-hand-side
  of \eqref{eq:131}. By \eqref{eq:180} it then follows that if both
  the domain and $\kappa$ are sufficiently large, and if
\begin{subequations}
\label{eq:55}
  \begin{equation}
  {\mathfrak j}_m > \|\A^{-1}(1,c)\|^{3/2}
\end{equation}
and
\begin{equation}
   {\mathfrak j}_+ > \sup_{\gamma\in\R} \|(\A_+(1,c)-i\gamma)^{-1}\|^{3/2} \,,
\end{equation}
\end{subequations}
then the normal state must be globally stable. The above conditions
serve as an upper bound for the critical current where the normal state
becomes globally stable. 

An obvious lower bound for this global stability current, is the
critical current for which the normal state becomes linearly unstable.
For large domains such instability is granted when $\mu_R<1$. By
\eqref{eq:130}, for sufficiently large $R$, it follows that the loss
of stability would take place when $\nu_m<1$. Using \eqref{eq:182} it
then follows that whenever
\begin{displaymath}
   {\mathfrak j}_+ < \nu_m(1,c)^{-3/2}\,,
\end{displaymath}
local stability is lost for sufficiently large $R$.  The optimality of
the above bound and of \eqref{eq:55} is left for future research.
\end{remark}

We conclude this work by providing some well-known elliptic-regularity
results for domains with corners in Appendices \ref{sec:a} and \ref{sec:b}. Then in
Appendix \ref{sec:A} we show how to use these results for parabolic
problems. Finally, in Appendix \ref{appendixD} we use the results of the previous
appendices to prove global existence, uniqueness, and regularity for
solutions of \eqref{eq:1}.

\section{Preliminaries}\label{section2}

\subsection{Equivalent boundary conditions.}~\\
Instead of considering the boundary conditions (\ref{eq:1}e,f,g), it
is possible to use an equivalent boundary condition where we prescribe
instead the magnetic field. As in \cite{ti10} we note that by
(\ref{eq:1}b,e,f), on each point on $\partial\Omega$, except for the corners,
we have
\begin{equation}\label{BJ}
  \frac{\partial}{\partial\tau}\curl A (t,\cdot) = \frac{1}{\kappa^2}J (\cdot)\,,
\end{equation}
where $\partial/\partial\tau$ denotes the tangential derivative along $\partial\Omega$ in the
positive trigonometric direction.
For convenience we set
\begin{equation}
J(x)\equiv0 \mbox{ on }\partial\Omega_i\,.
\end{equation}
 Thus, if we introduce on the boundary the function $B$ via
\begin{equation}
\label{eq:2}
  \curl A (t,x) = h\, B(t,x) \quad \text{on } \partial\Omega \,,
\end{equation}
where $h$ denotes a parameter measuring the intensity of
the magnetic field, we first observe that it satisfies 
\begin{equation}
\label{eq:36}
  B(t,x) = B(t,x_0) + \frac{1}{h \kappa^2}\int_{\Gamma(x,x_0)}J (\tilde x) \,ds (\tilde x) \,,
\end{equation}
where $(x,x_0)\in\partial\Omega\times\partial\Omega$, $\Gamma(x,x_0)$ is the portion on the
boundary connecting $x_0$ and $x$ in the positive trigonometric
direction, and $ds$ is a length element.    For later reference, we define 
 the reference current $J_r$
\begin{equation}
\label{eq:187}
hJ_r=J\,.
\end{equation}
Clearly, $J_r(x)$ is as smooth as $J$, i.e. at least
$C^2(\overline{\partial\Omega_c})$. Note that by \eqref{eq:76} we have
\begin{equation}
  \label{eq:3}
\int_{\partial\Omega}\,  J_r(x) \,ds=0 \,.
\end{equation}
One can recover the magnetic field $B(t,\cdot)$ at $x_0$ by integrating of
\eqref{eq:36} over $\pa \Omega$ ($x_0$ remaining fixed). This gives, with
the aid of (\ref{eq:1}g),
\begin{displaymath} 
B(t,x_0) = h_r - \frac{1}{ \kappa^2}\dashint_{\pa \Omega} \left(\int_{\Gamma(\tilde
    x, x_0)} J_r (x)  ds\right)\, ds (\tilde x)\, \mbox{ for } x_0 \in \pa
\Omega\,, 
\end{displaymath}
where $h_r=h_{ex}/h$. \\
We can thus conclude that $B(t,x)$ does
not depend on $t\,$, that is: $B(t,x)=B(x)$.  Switching the order of
integration then yields for $B$:
\begin{equation}
\label{eq:2bis}
   B(x) = h_r  -  \frac{1}{ \kappa^2}\dashint_{\pa \Omega} \, |\Gamma(\tilde x, x)| \, J_r(\tilde x)  ds (\tilde x)\, \mbox{ for } x  \in \pa \Omega\,.
\end{equation}

Note that by (\ref{eq:36}) and \eqref{eq:75}, $B$ must be continuous
along $\partial\Omega$ and, moreover, we have the property:
\begin{equation}
\mbox{ \it The magnetic field $B$ is constant along each component of }\pa \Omega_i\,.
\end{equation}
Hence (pending on the verification of the spaces in which we should
consider the solutions) the system (TDGL1) is equivalent to the system
(TDGL2)
\begin{subequations}  
\label{eq:1v2}
\begin{alignat}{2}
& \frac{\partial \psi}{\partial t} +
  i\phi\psi = \left(\nabla - iA \right)^{2} \psi + \psi 
  \left( 1 - |\psi|^{2} \right)\,,& \quad \text{ in } \R_+\times\Omega\, ,\\
 & \kappa^2\curl^2A + \sigma \left(\frac{\partial A}{\partial t} +
 \nabla\phi\right)  =  \Im(\bar\psi\, \nabla_A\psi) \, , & \quad \text{ in }  \R_+\times\Omega\,,\\
  &\psi=0 \,,&\quad \text{ on }  \R_+\times\partial\Omega_c\,, \\
 &(i\nabla+A)\psi\cdot\nu=0 \,,& \quad \text{ on }  \R_+\times\partial\Omega_i\,,\\
 &\curl A(t,x)  = h B(x) \,, & \text{ on } \R_+\times\partial\Omega\,, \\
&\psi(0,x)=\psi_0(x) \,, & \quad \text{ in } \Omega\,, \\ 
&A(0,x)=A_0(x)\,, & \quad \text{ in } \Omega \,.
\end{alignat}
\end{subequations}
where $B$ is given by \eqref{eq:2bis}.

Conversely, a solution of (TDGL2) must satisfy (TDGL1) with
\begin{displaymath}
 J_r = \kappa^2 \frac{\partial B}{\partial\tau} \,,\, \mbox{ and } 
 h_{r} = \dashint_{\pa \Omega} B(x) ds\,,
\end{displaymath}
having in mind that $J=h J_r$ and $h_{ex} = h h_r$.
\begin{remark}
  Note the above equivalence has only been established formally, as
  the regularity of the solutions has not been addressed yet. We
  return to this point in Subsection~\ref{gaugechoice} where we
  provide a precise definition of the spaces where the solutions
  reside.
\end{remark}
\vspace*{1ex}

\subsection{Stationary states}~\\
For a normal state we have $\psi\equiv 0$ by definition. Furthermore,
denoting the corresponding stationary magnetic and electric potentials
respectively by $hA_n$ and $h\phi_n$ we obtain, after dividing by $h$,
that $(A_n,\phi_n)$ weakly satisfy 
\begin{equation}
  \label{eq:8}
  \begin{cases}
     - c\, \curl^2A_n + \nabla\phi_n = 0 & \text{in } \Omega\,, \\
     \curl A_n = B   & \text{on } \partial\Omega\,,
  \end{cases}
\end{equation}
where 
\begin{equation}
\label{eq:9}
 c=\kappa^2/\sigma
\end{equation}
is a positive parameter and $B$ is defined by 
\eqref{eq:2bis}. \\ 
 As we later discuss, we choose the Coulomb gauge and assume that $A_n$ satisfies:
 \begin{equation}\label{eq:8a}
 A_n\in H^1(\Omega)\,,\, \Div A_n =0\,,\, A_n\cdot  \nu_{/\pa \Omega} =0\,.
 \end{equation}
 We now show that \eqref{eq:8a} combined with \eqref{eq:8} is uniquely
 solvable. We begin by constructing $\phi_n$ as a solution of the
 following problem, which can formally be obtained by taking the
 divergence of \eqref{eq:8},
\begin{equation}\label{eq:9a}
  \begin{cases}
     - \Delta\phi_n = 0 & \text{in } \Omega\,, \\
    \frac{\partial\phi_n}{\partial\nu}  = -\frac{J_r(x)}{\sigma}  & \text{on } \partial\Omega\,.
  \end{cases}
\end{equation}
\begin{remark}\label{sobaubord}
  Let $v$ denote a function in $H^2(\Omega)$. Then the trace of its normal
  derivative is well defined in $H^\frac 12 (\Gamma)$, where $\Gamma$ is any
  regular component of $\pa \Omega$. For convenience of notation we write
  $v\in H^\frac 12(\pa \Omega)$ in the sequel. The reader should thus be
  careful not to adopt the conventional interpretation of this
  notation which may not apply in some cases (consider,for instance,
  the case where $J$ is discontinuous at the corners).
\end{remark}
We seek a solution to the problem \eqref{eq:9a}  in $H^2(\Omega)$ such that 
\begin{equation}
  \label{eq:10}
\int_\Omega\phi_n\,dx=0 \,.
\end{equation}
Since $J_r \in C^2(\overline{\partial\Omega_c})$,  and
\begin{displaymath}
  \int_{\partial\Omega} J_r(x)\, ds = 0 \,.
\end{displaymath}
We can now use Proposition \ref{propoA2} and property (R1) of $\Omega$ to
obtain that $\phi_n$ uniquely exists and that
\begin{equation}
\label{regphin}
\phi_n\in W^{2,p}(\Omega)\,,
\end{equation}  
for all $p\geq 2$. 

Similarly, we construct $B_n= \curl A_n$ as the solution of a problem
which can be formally derived by taking the $\curl$ of \eqref{eq:8},
\begin{equation}\label{eq:9a1}
  \begin{cases}
      \Delta\,  B_n = 0 & \text{in } \Omega\,, \\
    B_n=  B  & \text{on } \partial\Omega\,.
  \end{cases}
\end{equation}
Although one can obtain an explicit formula for $B_n$ in $\Omega$ (which
amounts to extending \eqref{eq:2bis} into $\Omega$) using the strong
regularity of $J$, we prefer to use \eqref{eq:9a1} which allows
us to rely on the regularity results in Appendix \ref{sec:a}. In
particular, by Proposition \ref{propoA1}, we have, as $\Omega$ satisfies
condition (R1), a unique $B_n\in H^2(\Omega)$ solution of \eqref{eq:9a1}.
Moreover
\begin{equation}\label{bnlp}
B_n\in W^{2,p}(\Omega)\,,\, \forall p\geq 2\,.
\end{equation}
We can now proceed to determine $A_n$.  To ensure that $\Div A_n=0$,
we look for $A_n$ in the form:
\begin{displaymath}
  A_n= \nabla_\perp\theta_n\,.
\end{displaymath} 
Since $B_n=\curl A_n$, we have
  \begin{equation}
\label{eq:9a2}
  \begin{cases}
  -\Delta \theta_n = B_n & \text{in } \Omega\,, \\
    \theta_n= 0  & \text{on } \partial\Omega\,.
\end{cases}
\end{equation}
Hence, we set $\theta_n$ to be the variational solution of the above
Dirichlet problem.  The Dirichlet condition $\theta_n$ satisfies
ascertains that $A_n$ meets the condition $A_n\cdot \nu =0 $ on $\pa \Omega$.
From Proposition \ref{propoA1} (see \eqref{eq:136}) it then follows
that $\theta_n \in W^{3,p}(\Omega)$ for all $p<2$, and that $\theta_n\not\in
W^{3,p}(\Omega)$ for $p>2$ unless $B$ vanishes at every corner (a case
which certainly doesn't fall into the (J3) category in
\eqref{eq:75aa}).  Hence,
\begin{equation}
A_n\in W^{2,p}(\Omega,\mathbb R^2)\,,\, \mbox{ for all } p<2\,.
\end{equation}
It remains to show that $(\phi_n,A_n)$ is indeed a
solution of \eqref{eq:8}, as we have only established, so far, that $
V_n := - c\, \curl^2A_n + \nabla\phi_n$ satisfies~:
  \begin{equation*}\label{curldiv} V_n\in L^2(\Omega,\mathbb R^2)\,, \,\Div
    V_n =0\,\mbox{ and }\, \curl V_n =0 \mbox{ in } \Omega\,, \, V_n\cdot \nu
    =0 \mbox{ on } \pa \Omega\,.
\end{equation*}
To obtain the last property of the previous line, we used \eqref{BJ}
and the boundary condition for $\phi_n$ in \eqref{eq:9a}.  We can now
use the decomposition Proposition \ref{decprop} to conclude that
$V_n=0$, and hence $(\phi_n,A_n)$ satisfy \eqref{eq:8}. Finally, note
that $\phi_n$ is unique, by \eqref{eq:9a} and \eqref{eq:10}, and that
the uniqueness of $A_n$ follows from
\eqref{eq:8}. \\

We summarize the above discussion by the following proposition:
\begin{proposition}\label{prop2.1}~\\
 Suppose that $\pa \Omega$ has the property (R1). 
There exists a unique solution $(\phi_n,A_n) \in H^2(\Omega)\times H^1(\Omega)$ satisfying \eqref{eq:8},
  \eqref{eq:8a}, and \eqref{eq:10}.  This solution belongs to
  $W^{2,p}(\Omega)\times  W^{2,q}(\Omega,\mathbb R^2)$, for all $p\geq 2$ and $q<2$,  Moreover $(0,
  h\phi_n,hA_n)$ is a stationary solution of \eqref{eq:1}, and $\curl A_n\in
  W^{2,p}(\Omega,\R^2)$ for all $p\geq 2$.
\end{proposition}
Using Sobolev embeddings, we deduce in particular  that:
\begin{equation}
 \phi_n\,,\,B_n\,,\, \mbox{ and } \curl B_n  \mbox{ belong  to } C^1(\overline{\Omega})\,.
 \end{equation}
\vspace*{1ex}

\subsection{A magnetic Laplacian}\label{ss2.3}
Next, we define
\begin{equation}
\label{eq:11}
  \mu(h) = \inf_{
    \begin{subarray}{2}
      u\in H^{1,\partial\Omega_c}_0 \\
      \|u\|_2 =1
    \end{subarray}} \| \nabla_{hA_n}u\|_2^2 \,,
\end{equation}
where
\begin{displaymath}
  H^{1,\pa \Omega_c}_0= \{ u\in H^1(\Omega,\C)\,|\, u|_{\partial\Omega_c}=0 \} \,, 
\end{displaymath}
in which the boundary data appear in a trace sense. Using the
diamagnetic inequality \cite{fohe09}, it is easy to show
that 
\begin{equation}\label{minmu}
\mu(h) \geq \mu(0) >0\,,
\end{equation}
since $\Omega$ is relatively compact.\\
Under relatively weak assumptions one can obtain that
\begin{equation}\label{Assmu}
\lim_{h\ar +\infty} \mu(h) = +\infty\,.
\end{equation}

One can estimate the rate of divergence of $\mu$, in the large $h$
limit, by assuming first that (R2), (J2) and (J3) hold true. In
that case, $B$ is strictly monotone on each component of $\pa \Omega_c$.
We now argue as in \cite{al08} (proof of Proposition 4.1 there).
Observing that $B_n$ is continuous on $\overline{\Omega}$ and harmonic in
$\Omega$, the maximum principle shows that the minimum $B_{min}$ of $B_n$
in $\Omega$ is attained on one component of $\pa \Omega_i$ and that the
maximum $B_{max}$
is attained at the other component.\\

Assume further that
 \begin{equation}\label{(B)}
 \mbox{ (B)  }  B_n^{-1}(0) = \emptyset  \mbox{ or }
B_{min} < 0 < B_{max}\,,
\end{equation}
which implies that $B_n^{-1}(0)$ lies away from the corners.

\begin{remark}
  This condition can be expressed in terms of the boundary data
  \eqref{eq:1}.  With the aid of (\ref{eq:2bis}) we obtain that
  $0\in(B_{min},B_{max})$ is equivalent to
  \begin{displaymath}
     \frac{1}{ \kappa^2}\dashint_{\pa \Omega} \, |\Gamma(\tilde x, c_{min} )| \,
     J_r(\tilde x)  ds (\tilde x)\,  < h_r  <  \frac{1}{\kappa^2}\dashint_{\pa \Omega} \, |\Gamma(\tilde x, c_{max})| \, J_r(\tilde x)  ds (\tilde x)\, 
  \end{displaymath}
where $c_{min}$ and $c_{max}$ lie both on $\partial\Omega_i$ and satisfy
$B(c_{min})=B_{min}$ and $B(c_{max})=B_{max}$.
\end{remark}

By the maximum principle, we first deduce that $B_{n}^{-1}(\mu)$ cannot
contain a loop for any $\mu\in(B_{min},B_{max})$.  Then, we use the
above-stated monotonicity to conclude that $B_n^{-1}(\mu)\cap\partial\Omega$
consists of precisely two points: one on each connected component of
$\pa \Omega_c$.  Hence $B_n^{-1}(\mu)$ must be a simple smooth curve
joining the two components of $\pa \Omega_c$.  By Hopf's lemma for
harmonic functions (cf. \S\,6.4.2 in \cite{ev98} for instance), it thus
follows that $\nabla B_n\neq 0$ on $\Omega \cap B_n^{-1} (\mu )$.  By the boundary
condition for $B_n$, we obtain $\nabla B_n\neq 0$ on $\partial\Omega_i$ as well. We
have thus proved that
\begin{equation} \label{hypnabla}
\nabla B_n \neq 0 \mbox{ in } B_n^{-1}([B_{min},B_{max}])=\bar{\Omega}
\,,
\end{equation}
and in particular on $B_n^{-1}(0)$.  It is now possible to use the
same methods as in \cite{pakw02} to obtain the existence of some
$\mu_0>0$ such that
\begin{equation}
\label{eq:12}
  \mu(h) \geq  \inf_{
    \begin{subarray}{2}
      u\in H^1(\Omega,\C) \\
      \|u\|_2 =1
    \end{subarray}} \| \nabla_{hA_n}u\|_2^2 \geq \mu_0\,  h^{2/3}, \forall h\geq 1\,.
\end{equation}
\vspace*{1ex}

\subsection{Another spectral  entity}\label{ss2.4}~\\
To be able to state the main result in the next section we need to define yet another entity.
Let then
\begin{equation}
  \label{eq:13}
\lambda= \inf_{
    \begin{subarray}{2}
      V\in\Hg_d \\
      \|V\|_2 =1
    \end{subarray}} \| \curl V\|_2^2 \,,
\end{equation}
where
\begin{equation}
  \Hg_d = \big\{ V\in H^1(\Omega,\R^2)\,|\, \Div V=0 \,, V\big|_{\partial\Omega}\cdot\nu=0 \big\}  \,.
\end{equation}
We next provide an alternative characterization of $\lambda$. \\
 
\begin{proposition} 
\label{prop:dirichlet}
Under condition $(R1)$, 
\begin{equation}
\label{eq:177}
\lambda =\lambda^D\,,
\end{equation}
where  $\lambda^D$ is the ground state energy of the Dirichlet Laplacian $-\Delta^D$.
\end{proposition}
\begin{proof}
We have seen in Proposition \ref{propoA1} that the domain of $\Delta^D$ is $\mathcal H_p:= H^2(\Omega)\cap H_0^1(\Omega)$.\\
 Let $u $
denote  an $L^2$-normalized  ground state of $-\Delta^D$. 
 Then $\nabla^{\perp} u$ belongs to
$\mathcal H_d$, and $\curl \nabla^{\perp} u= -\Delta u = \lambda^D u$.  Hence
$$ \|\curl \nabla^{\perp} u\|^2 = \lambda^D \, \langle u\,,\, \curl \nabla^{\perp} u\rangle = \lambda^D \,  \|
\nabla^{\perp} u\|^2\,.
$$
From the above we deduce that 
$$
\lambda \leq \lambda^D\,.
$$

Conversely, let $V\in\Hg_d$.  Under assumption (R1) there exists $\Phi\in
\Hg_p $ such that $V =-\nabla_\perp\Phi$ (cf. Proposition \ref{decprop}).
Moreover $\nabla^{\perp} $ is a bijection from $\mathcal H_p$ onto $\mathcal
H^d$. Hence, we can rewrite \eqref{eq:13}, in terms of $\Phi$, in
the form
\begin{equation}
\label{eq:143}
  \lambda= \inf_{
    \begin{subarray}{2}
      \Phi\in\Hg_p \\
      \|\nabla\Phi\|_2 =1
    \end{subarray}} \| \Delta\Phi\|_2^2 \,.
\end{equation}
It can be readily verified  that the 
    functional in \eqref{eq:143}  is lower
semicontinuous. Furthermore, it is also coercive in view of
Proposition \ref{propoA1}. We can thus conclude the existence of a
minimizer which we denote by $\Phi_{min}$. Evaluating the first
variation we can conclude that
\begin{displaymath}
  \int_\Omega \Delta\Phi_{min}\, (\Delta\eta+\lambda\eta) \, dx = 0\,,\, \forall \eta \in \mathcal H_p \,.
\end{displaymath}
Clearly, if $\Delta+\lambda:\Hg_p\to L^2(\Omega)$ is invertible, then we must have
$\Delta\Phi_{min}=0$ and since $\Phi_{min} \in H^1_0(\Omega)$, it follows that
$\Phi_{min} \equiv0$, contradicting the requirement that $\|\nabla
\Phi_{min} \|_2=1$.  Consequently, $\lambda$ is an eigenvalue of the
Dirichlet Laplacian $-\Delta^D$, hence satisfying $\lambda \geq \lambda^D$.  
 \end{proof}

\subsection{Gauge equivalence and weak
  solutions}\label{gaugechoice}\label{ss2.5}~\\ 
We assume that $\Omega$ has property (R1) (see  (\ref{propertyR})).
Let $A\in L^2_{loc}([0,\infty);H^1(\Omega,\R^2))$, $\psi\in
L^2_{loc}([0,\infty);H^1(\Omega,\C))$, and $\phi\in
L^2_{loc}([0,\infty);L^2(\Omega,\R^2))$.  Following \cite{chenetal93}, we say that $(\psi',A',\phi)$ is gauge
equivalent to $(\psi,A,\phi)$ if there exists $\omega\in L^2_{loc}([0,\infty);H^2(\Omega))\cap
H^1_{loc}([0,\infty); H^1(\Omega))$ such that 
\begin{equation}
\label{eq:47}
  A'= A + \nabla\omega\,, \; \phi'= \phi-
  \frac{\partial\omega}{\partial t}\,, \; \psi' = \psi e^{i\omega} \,.
\end{equation}
We say that $(\psi',A',\phi')=G_\omega(\psi,A,\phi)$ in that case.
It is easy to show that \eqref{eq:47} is an equivalence relation. We
begin by defining the gauge (cf. \cite{chenetal93})
\begin{equation}
\label{eq:7}
  {\mathbb H}= \{ (u,v)\in H^1(\Omega,\mathbb R^2)\times L^2(\Omega)\,|\, c\, \Div u + v =0\, ; \,
  u\cdot\nu|_{\partial\Omega}=0 \,\} \,.
\end{equation}
Let
$(\psi,A,\phi)\in L^2_{loc}([0,\infty);H^1(\Omega,\C))\times L^2_{loc}([0,\infty);H^1(\Omega,\R^2))\times
L^2_{loc}([0,\infty);L^2(\Omega,\R^2))$, such that $A_0|_{\partial\Omega}\cdot\nu=0$ (where $A_0= A(0,\cdot)$).
We first show that there exists a unique gauge equivalent
$(\psi_d,A_d,\phi_d)\in L^2_{loc}([0,\infty);H^1(\Omega,\C))\times L^2_{loc}([0,\infty);\mathbb{H})$. To
prove this, following \cite{chenetal93}, we first solve
\begin{equation}\label{systchi}
  \begin{cases}
    \frac{\partial\chi}{\partial t} - c\, \Delta\chi = c\, \Div A + \phi & \text{in } (0,\infty)\times\Omega \\
      \frac{\partial\chi}{\partial\nu} = -A\cdot\nu & \text{on } (0,\infty) \times\partial\Omega \\
      \chi|_{t=0}=0 &    \text{in } \Omega \,.
  \end{cases}
\end{equation}
In Appendix \ref{sec:A} we prove that there exists a solution
$\chi\in L^2_{loc}([0,\infty);H^2(\Omega))\cap H^1_{loc}([0,\infty);L^2(\Omega))$ to the above problem. It can
be readily verified that $(\psi_d,A_d,\phi_d)=G_\chi(\psi,A,\phi)$. In the case
where $A_0|_{\partial\Omega}\cdot\nu\neq0$ we define first the gauge function
\begin{displaymath}
   \begin{cases}
    \Delta\chi_0 = 0 & \text{in } \Omega \\
      \frac{\partial\chi_0}{\partial\nu} = -A_0\cdot\nu & \mbox{ on } \partial\Omega \,.
  \end{cases}
\end{displaymath}
Then, we replace $(\psi,A,\phi)$ by $G_{\chi_0}(\psi,A,\phi)$ and proceed as
before. 

Let
\begin{displaymath}
  \Wg_1= \{ V \in H^1(\Omega;\mathbb R^2)\,,\, V\cdot \nu =0\mbox{ on } \pa
  \Omega\} \,,
\end{displaymath}
and
\begin{displaymath}
  \Wg_2= \{\psi \in H^1(\Omega)\,,\, \psi=0 \mbox{ on } \pa \Omega_c\} \,.
\end{displaymath}
Let $A\in L^2_{loc}([0,\infty);\Wg_1))\cap H^1_{loc}([0,\infty); \Wg_1^\prime)$,
$\psi\in L^2_{loc}([0,\infty);\Wg_2)\cap H^1_{loc}([0,\infty); \Wg_2^\prime)$, and
$\phi\in L^2_{loc}([0,\infty);L^2(\Omega,\R^2))$ denote a weak solution of
\eqref{eq:1}. (The reader is referred to \cite{chenetal93} for the
definition of a weak solution.)  It is easy to show that
$G_\omega(\psi,A,\phi)$ is also a weak solution of \eqref{eq:1} for any $\omega\in
L^2_{loc}([0,\infty);H^2(\Omega))\cap H^1_{loc}([0,\infty);L^2(\Omega))$. In Theorem
\ref{Existencetheorem} and Appendix \ref{appendixD} we prove (relying
on \cite{feetal98}) that the solution of \eqref{eq:1v2} in
$L^2_{loc}([0,\infty);H^1(\Omega,\C))\times L^2_{loc}([0,\infty);H^1(\Omega,{\mathbb H}))$ is unique.\\
From the foregoing discussion we can thus conclude that all the weak
solutions of \eqref{eq:1} are gauge equivalent. In particular, this
would mean that all the results established in the next sections for
the decay of $|\psi|$ are valid for all possible solutions of
\eqref{eq:1v2}.  In the next subsections, we concentrate on strong
solutions of \eqref{eq:1v2}.
\vspace*{1ex}

\subsection{The strong solution in the Coulomb gauge}
\label{sec:columb}
In view of the discussion in the previous subsection we fix 
the Coulomb gauge, i.e., we look for global solutions  in
$L^2_{loc}([0,+\infty), H^1( \Omega,\mathbb R^2))$ of \eqref{eq:1} satisfying 
\begin{equation}\label{eq:4} 
  \Div A(t,\cdot)=0 \mbox{ in } L^2_{loc}([0,+\infty), L^2( \Omega))\,, A(t,\cdot)\cdot\nu|_{\pa \Omega} =0\,\mbox{ in }  L^2_{loc}([0,+\infty), H^\frac 12 ( \pa \Omega))\,,
  \end{equation}
  and we also assume:
\begin{equation}
\label{eq:5}
  \int_\Omega\phi(t,x)\,dx =0 \mbox{ in } L^2_{loc}([0,+\infty)) \,.
\end{equation}

Suppose first that the initial condition $A_0$ satisfies
\begin{equation}\label{eq:4aa}
  \Div A_0=0 \mbox{ in }  \Omega\,, A_0\cdot\nu =0\, \mbox{ on }  \pa \Omega\,,
  \end{equation} 
where  
\begin{equation}
\label{assA0}
  A_0\in H^2(\Omega,\mathbb R^2). 
\end{equation}
We further assume that 
  \begin{equation}\label{asspsi0}
  \psi_0\in H^2(\Omega,\mathbb C)\,,
  \end{equation}
  and \eqref{eq:2.2}.
  
We show that the solution $(\psi_d,A_d,\phi_d)$ established in Theorem
\ref{thm:A.1} with
$\widehat A_0=A_0$ and $\widehat \psi_0=\psi_0$  is gauge-equivalent to the solution of
\eqref{eq:1} and \eqref{eq:4}. \\
To this end we define the gauge
function $\omega$ as the solution of
  \begin{equation}
\label{eq:146}
    \begin{cases}
      -\Delta\omega = \Div A_d & \text{in } (0,+\infty)\times  \Omega\,, \\
      \frac{\partial\omega}{\partial\nu} = 0 & \text{on } (0,+\infty)\times\partial\Omega \,,\\
     \int_{\Omega}\omega(t,x)\,dx =0 &  \text{in }  (0,+\infty)\,.    
\end{cases}
  \end{equation}  
  As $A_d\in C([0,+\infty);W^{1+\alpha,2}(\Omega,\R^2))$ for any $0<\alpha<1$, it follows
  by Sobolev embeddings and Proposition \ref{propoA2} that
 $\omega\in C([0,+\infty);W^{1,p}(\Omega))$ for all $p\geq2$. Furthermore, since
 $\Div A_d\in L^2_{loc}([0,+\infty),H^1(\Omega))$  we get by   \eqref{eq:137}
 \begin{equation}\label{regomega}
 \omega\in L^2_{loc}([0,+\infty),H^3(\Omega))\,.
 \end{equation}
 Next, we observe that the projector $\pi_1$ introduced in Proposition
 \ref{decprop} extends (by tensor product) to a projector $\Pi_1$ in
 $H^1_{loc}([0,+\infty);L^2(\Omega,\mathbb R^2))$ and that by the uniqueness of
 the decomposition established in the proposition and \eqref{eq:146}:
  \begin{equation}\label{omad}
  -\nabla \omega = \Pi_1 A_d\,,
  \end{equation}
  in  $\mathcal D'(0,+\infty ;L^2(\Omega,\mathbb R^2))$, where $\mathcal D'(0,+\infty; L^2(\Omega;\mathbb R^2))$ denotes the space of distributions on $(0,+\infty)$ with value in $L^2(\Omega,\mathbb R^2)$.\\
  Note that \eqref{omad} simply reads 
    \begin{equation}\label{omada}
  -\nabla \Big(\int \omega (t,\cdot) \phi(t) dt \Big) = \pi_1 \Big(\int A_d(t,\cdot) \phi(t) dt\Big) \,,
  \end{equation}
  for all $\phi\in C_0^\infty (0,+\infty)$.\\
  The right hand side of \eqref{omad} being in $H^1_{loc}([0,+\infty);L^2(\Omega,\mathbb
  R^2)$, this implies that $\nabla \omega \in H^1_{loc} ([0,+\infty);L^2(\Omega,\mathbb
  R^2))$, and hence 
  \begin{equation} \label{regomega2}
  \pa_t \omega \in L^2_{loc}([0,+\infty);H^1(\Omega,\mathbb R^2))\,.
  \end{equation}
 It can now be readily verified  from \eqref{regomega} and \eqref{regomega2} that
  the Coulomb gauge solution $(\psi_c,A_c,\phi_c)=G_\omega(\psi_d,A_d,\phi_d)$
  satisfies:
 \begin{equation} \label{prop1coul}
 \psi_c\in C([0,+\infty);W^{1+\alpha,2}(\Omega,\C))\cap H^1_{loc}([0,+\infty);L^2(\Omega,\C))\,, \, \forall \alpha <1\,,
 \end{equation}
 \begin{equation}\label{prop2coul}
 A_c\in
C([0,+\infty); W^{1,p}(\Omega,\R^2))\cap H^1_{loc}([0,+\infty);L^2(\Omega,\R^2))\,, \forall p\geq 1\,,
\end{equation}
which follows from the fact that by (\ref{omad}) $\nabla \omega\in  C([0,+\infty);
W^{1,p}(\Omega,\R^2))$, and
\begin{equation}\label{prop3coul}
\phi_c\in
L^2_{loc}([0,+\infty); H^1(\Omega))\,.
\end{equation}

Relying on  Theorem \ref{thm:A.1} and the above discussion we can now state:
\begin{theorem}\label{Existencetheorem}
   Suppose that $\Omega$ satisfies condition (R1) and that $B$ is in
   $H^\frac 12(\pa \Omega)$ (see Remark \ref{sobaubord}). Suppose further
   that  $(\psi_0,A_0)$ satisfies \eqref{assA0}, \eqref{eq:4aa}, \eqref{asspsi0} and \eqref{eq:2.2}.\\
   Then, there exists a unique weak solution $(\psi_c ,A_c,\phi_c)$ of
   (TDGL2)  in the Coulomb gauge. Moreover, this solution is strong in
   the sense that it satisfies \eqref{prop1coul}-\eqref{prop3coul} and
  \begin{equation}\label{eq:2.2a}
 \|\psi_c(t,\cdot)\|_\infty \leq 1\,,\, \forall t >0\,.
 \end{equation}
Finally, let $A_1 = A_c-hA_n$ where $A_{n}$ satisfies
\eqref{eq:8} and \eqref{eq:8a}. Then 
\begin{equation}
A_1\in L^2_{loc}([0,+\infty);H^2(\Omega,\R^2))\,.
\end{equation}
\end{theorem}
The last statement of the theorem is a consequence of \eqref{eq:D4},
\eqref{regomega} and Theorem \ref{thm:A.1}. We can indeed write:
$$
A_1= \widehat A_1  - \nabla \omega  - h (A_n - A_{n,d} )\,,
$$
where $ \widehat A_1$ is defined in \eqref{eq:186}

\begin{remark}
If $A_0$ does not satisfy \eqref{eq:4aa},  we let $\omega_0\in H^2(\Omega)$ be solution of
  \begin{displaymath}
    \begin{cases}
      -\Delta\omega_0 = \Div A_0 & \text{in } \Omega\,, \\
      \frac{\partial\omega_0}{\partial\nu} =- A_0\cdot \nu & \text{on } \partial\Omega \,,\\
     \int_{\Omega}\omega_0(x)\,dx =0\,. &      
\end{cases}
  \end{displaymath}
It follows by  Proposition \ref{propoA2},  that, for $\Omega$ satisfying
property (R),  $\omega_0\in H^2(\Omega))$. We then consider $G_{\omega_0}(\psi,A,\phi)$ which
is a solution of \eqref{eq:1} with initial conditions satisfying
\eqref{eq:4aa}.
\end{remark}

We can now return to the solution of (TDGL1).
\begin{theorem}\label{Existencetheoremcomp}
  Under the assumptions of the previous theorem, assuming that $J$ satisfies
  \eqref{eq:75}-\eqref{eq:76}, and $B$ by \eqref{eq:36}, the
  solution of (TDGL2) has the additional property that $\phi_c \in
  C([0,+\infty);W^{1,p}(\Omega))$ for all finite $p$, and is a solution of
  (TDGL1).
\end{theorem} 
\begin{proof} 
  Let $(\psi_c,A_c,\phi_c)$ denote a solution of (TDGL2) and \eqref{eq:4}.
  One has to clarify first the sense in which the trace condition
  (\ref{eq:1}e)-(\ref{eq:1}f) is satisfied.  By Theorem
  \ref{Existencetheorem} we have that $\pa_t A_c + \nabla \phi_c$ belongs to
  $ L^2_{loc}([0,+\infty),L^2(\Omega,\mathbb R^2))$.  Hence, we can use the
  fact (see for example Theorem 2.2 in \cite{gira79}) that for a
  vector field $V$ in $L^2_{loc}(0,+\infty;L^2(\Omega;\mathbb R^2))$ \\with $\Div V \in L^2_{loc}([0,\infty);L^2(\Omega))$, the normal component of
its trace, $V\cdot \nu|_{\partial\Omega}$, belongs to  $L^2_{loc}([0,+\infty); H^{-\frac 12}
(\partial\Omega))$.\\
Consider then $V= \pa_t A_c + \nabla\phi_c $.  
By (\ref{eq:1v2}b) and \eqref{eq:4} we obtain:
\begin{equation}
\label{neumannpourphi}
\sigma \Div V= \sigma \Div \nabla \phi_c =\Im  \Div (\bar \psi_c \cdot \nabla_{A_c} \psi_c)\,.
\end{equation}
It is easy to show that the left hand side is in $L^2_{loc}([0,+\infty);
L^2(\Omega))$.  As $\Delta_{A_c} \psi_c \in L^2_{loc}([0,+\infty);L^2(\Omega))$ we can use
\eqref{eq:2.2a} to conclude that $\psi_c\Delta_{A_c} \psi_c \in
L^2_{loc}([0,+\infty);L^2(\Omega))$. Furthermore,  $\nabla\psi_c \in
C([0,+\infty); L^4(\Omega,\mathbb R^2))$ and  $A_c\in C([0,+\infty)\times \overline \Omega)$ in view of
\eqref{prop1coul} and \eqref{prop2coul}\,, hence
$\nabla\psi_c\cdot\nabla_{A_c}\psi_c\in L^2_{loc}([0,+\infty);L^2(\Omega))$. Consequently, $V\cdot \nu$ is
well defined in $L^2_{loc}([0,\infty);H^{-1/2}(\partial\Omega))$, and we can discern
that
 $$
 V\cdot \nu|_{\pa \Omega}=\pa_\nu \phi\,,
 $$ 
 due to the fact that $\pa_t A_c\cdot \nu =0$ in $\mathcal D'(0,+\infty; H^\frac
 12(\pa \Omega))$ by \eqref{eq:4}. \\
Consider again (\ref{eq:1v2}b). Each term of the equality has a
meaningful normal component for its trace and hence, as the right hand
side has a zero "normal" trace,
 \begin{equation}
 V\cdot \nu = \kappa^2 \pa_\tau \curl A_c = \kappa^2 h\,  \pa_\tau B = J\,,
 \end{equation}
 in $L^2_{loc}([0,+\infty); H^{-\frac 12}(\pa \Omega))$, 
 as expected.
  \end{proof}

\begin{remark}
  Although the main focus of this work  is on Coulomb gauge
  solutions, it seems worthwhile to note that any weak solution for
  which $A\cdot \nu=0$ on $\pa \Omega$, $\Div A \in L^2_{loc}([0,+\infty);
  H^1(\Omega))$, $\phi \in L^2 _{loc}([0,+\infty); H^1(\Omega))$, and the initial
  data are regular, is also a strong solution. To prove this, we
  return to Equation \eqref{systchi} to determine the regularity of
  $\chi$. Then, one has to show that $\chi$ satisfies \eqref{regomega} and
  $\eqref{regomega2}$, which we prove by considering the
  equation satisfied by $V=\nabla \chi$. Differentiating \eqref{systchi} yields,
\begin{equation}\label{equationpourqa}
  \begin{cases}
    \frac{\partial V}{\partial t} - ( \nabla \Div + \nabla^{\perp} \curl) V = c \nabla \Div A + \nabla \phi & \text{in } (0,+\infty) \times\Omega \\
      \curl V  = 0 & \text{on } (0,+\infty)\times\partial\Omega \\
     V\cdot \nu =0& \text{on } (0,+\infty) \times\partial\Omega \\
      V_{t=0}= 0 &    \text{in } \Omega \,.
  \end{cases}
\end{equation}
The above problem possesses a unique solution $V \in L^2_{loc}([0,+\infty);
D(\mathcal L^{(1)}))$ and $\pa_t V\in L^2_{loc}([0,+\infty);L^2(\Omega,\mathbb
R^2))$ (cf. Proposition \ref{propoA3}). Since by
Proposition~\ref{propoA4}, $D(\mathcal L^{(1)})\subset H^2(\Omega, \mathbb
R^2)$, the desired regularity of $\chi$ readily follows.
\end{remark}

\section{Asymptotic contraction properties  of the   semi-group}\label{section3}
In this section we obtain our simplest estimate for the critical
current, for which the normal state $(0,A_n,\phi_n)$, given by
\eqref{eq:8}, becomes globally stable. We concentrate here on currents
for which the (non-linear) semi-group associated with \eqref{eq:1}  and
  \eqref{eq:4} becomes a contraction for sufficiently long times.

\subsection{Analysis of the linearized problem}\label{ss3.1}~\\
Consider first the linearized version of (\ref{eq:1}):
\begin{equation}
\label{eq:144}
  \begin{cases}
     \frac{\partial u}{\partial t} 
  = \LL u\,,  & \text{
    in } \R_+\times\Omega\, ,\\
(i\nabla+hA_n)u\cdot\nu=0\,,  &  \text{ on }  \R_+\times\partial\Omega_i\,,\\
u=0\,, &  \text{ on }  \R_+\times\partial\Omega_c\,, \\
u(0,\cdot)=u_0(\cdot)\,,  &  \text{ in } \Omega\,.
  \end{cases}
\end{equation}
In the above
\begin{displaymath}
  \LL = \left(\nabla - ihA_n \right)^{2} + ih\phi_n + 1 \,.
\end{displaymath}
It is easy to show using integration by parts that for any
$v\in D(\LL)$ we have
\begin{displaymath}
  \langle v,\LL v\rangle =   - \big\|\nabla_{hA_n}v\big\|_2^2 + \|v\|_2^2 \,.
\end{displaymath}
By \eqref{eq:11} we have that
\begin{displaymath}
  \langle v,\LL v\rangle \leq -(\mu-1)\|v\|_2^2 \,.
\end{displaymath}
Note that if $v$ is a ground state of $\LL$ the above inequality
becomes an identity.  Hence, it follows that the operator $\LL$ is
dissipative if and only if $\mu\geq1$.  Consequently, it is easy to show
from the Lumer-Phillips Theorem (Theorem 8.3.5 in \cite{da07}) that
the semigroup associated with \eqref{eq:144} is a contraction
semigroup if and only if $\mu\geq1$. If $\mu>1$ one can easily show that
any solution of \eqref{eq:144} decays exponentially fast (with a decay like $\exp - (\mu-1)t\, $)  and hence,
that $u\equiv0$ is asymptotically stable.

If we now consider the linearized part of  (\ref{eq:1}b), (after taking its
curl), we get the equation for the first variation $w$ of $\curl A$
  $$
  \begin{cases}
  \sigma \pa_t w - \kappa^2  \Delta w=0\, \mbox{ in } \mathbb R^+\times   \Omega\,, \\
  w(t, \cdot )= 0 \, \mbox{ on }\mathbb R^+\times  \pa \Omega\,,\\
  w(0, \cdot) = w_0(\cdot)\, \mbox{ on } \Omega\,.
  \end{cases}
  $$
  From the above we can conclude an $\OO(e^{ - \lambda_D ct})$-decay for
  $w(t,\cdot)$, where $\lambda^D$ is defined below \eqref{eq:177} and
  $c=\frac{\kappa^2}{\sigma}$ (cf \eqref{eq:9}).  We recall from Proposition
  \ref{prop:dirichlet} that $\lambda =\lambda^D$.\\ 
  We say that $f:\R\to\R_+$ has a monotone $\OO(e^{-\alpha t})$ decay, if
  $e^{\alpha t}f$ is monotone.  From the foregoing discussion it therefore
  follows that we cannot hope for a better monotone decay than $e^{-
    \min (\mu-1,\lambda c) t}$ for the asymptotic behavior of the nonlinear
  problem that we consider in the next subsection. (See Formula
  \eqref{eq:15}.)

  We note that if, instead of imposing that both
  $e^{\alpha t}\|\psi(t,\cdot)\|_2$ and $e^{\alpha t}\|A(\cdot,t)-hA_n\|_2$ become
  monotone for sufficiently large $t$, we impose the weaker
  requirement that they are both bounded, we can obtain greater values
  of $\alpha$. This is precisely the focus  of Sections \ref{section4} and  \ref{section5}.

  \subsection{Asymptotic analysis}\label{ss3.2}
  In light of the above discussion we expect that asymptotic stability
  of $\psi\equiv0$ could be achieved for $\mu>1$ (at least in some asymptotic
  regimes). We also expect that the semigroup associated with
  \eqref{eq:1} and \eqref{eq:4} would turn asymptotically into a
  contraction semigroup, when $\|\psi(t,\cdot)\|_2$ becomes very small. The
  statement and the proof of theses intuitive observations are made in
  part, in the following theorem for fixed values of $\sigma$ and $\kappa$.  A
  particular attention is devoted to the limit $\kappa\to\infty$ with fixed
  $c=\kappa^2/\sigma$.
\begin{theorem}
\label{thm:3.1}
Let $(\psi,A,\phi)\in\Hg$ denote a solution of \eqref{eq:1} and
\eqref{eq:4} satisfying \eqref{eq:2.2}. Then, whenever
   \begin{equation}
    \label{eq:14}
\mu(h)>1+ \frac{\gamma}{\kappa^2} + \frac{\gamma^2}{\kappa^4} \,,
  \end{equation}
where
\begin{displaymath}
  \gamma = (\lambda c+2)\Big[\frac{2}{c\lambda^3}\Big]^{1/2} \,,
\end{displaymath}
there exist  $C= C(\lambda,\kappa,\|\psi_0\|_2,\|A_0\|_2,h)>0$  and $\lambda_m= \lambda_m(c,\kappa,\mu(h),\lambda) >0 $\, 
such that, for all $t>0$, 
 we have:
\begin{equation}
    \label{eq:15}
\|\psi\|_2 + \|A-hA_n\|_2 \leq Ce^{-\lambda_mt} \,,
  \end{equation}
  where $A_n$ is the stationary normal solution introduced in
  \eqref{eq:8}-\eqref{eq:9a} and satisfying \eqref{eq:4}. \\
  
   Moreover, when $\mu>1$, there
  exists $C=C(\mu,\lambda,c)>0$  and $\kappa_0(\mu,\lambda,c)$ such that, for $\kappa\geq \kappa_0(\mu,\lambda,c)  $,
\begin{equation}
\label{eq:152}
  \lambda_m\geq   \min((\mu-1),\lambda c) -  \frac{C}{\kappa}
\end{equation}

More precisely, 
\begin{enumerate}
\item
If  $ 0< \mu -1<  \lambda c $ then 
 \begin{equation}\label{eq:152a}
 \lambda_m \geq   (\mu -1) - 4 c \mu (\lambda c  - \mu +1) ^{-1} \kappa^{-2} +\Og(\kappa^{-4})\,. 
\end{equation}
\item
If  $\mu-1 = \lambda c$
\begin{equation}\label{eq:152b}
 \lambda_m \geq   \lambda c -    \frac{ 2\, c^\frac 12   \, 3^{-\frac 14} \,  (\lambda c +1)^{\frac 12} }{ \kappa }  + \Og (\kappa^{-2})\,.
\end{equation}
\item 
If 
 $\mu -1 >  \lambda c $
 \begin{equation}\label{eq:152c}
 \lambda_m \geq   \lambda c -  \frac{ cd}{\kappa^2} + \Og (\kappa^{-4})\,.
\end{equation}
where
\begin{equation} \label{calculded}
d = 4  \mu^\frac 12 (\lambda c+1)^\frac 12 (\mu-1- \lambda c)^{-1}\,.
\end{equation}
\end{enumerate}
\end{theorem}

{\it Proof:}~\\
Set
  \begin{equation}
\label{eq:17}
    A_1 = A-hA_n \quad ; \quad \phi_1=\phi-h\phi_n \,.
  \end{equation}
Note that by \eqref{eq:5} and \eqref{eq:10} we have 
\begin{equation}\label{eq:17a}
\int_\Omega \phi_1(t,x) \,dx=0\,.
\end{equation}
  Substituting into (\ref{eq:1}b,e,f)  yields with the
  aid of \eqref{eq:4},
\begin{subequations}  
\label{eq:18}
\begin{alignat}{2}
& \sigma\frac{\partial A_1}{\partial t} +\kappa^2\curl^2A_1 +
 \sigma\nabla\phi_1  =  \Im(\bar\psi\nabla_A\psi) & \quad \text{ in }  \R_+\times\Omega \,,\\
 & \frac{\partial\phi_1}{\partial\nu}=0 &\quad \text{ on } \R_+\times\partial\Omega \,.
\end{alignat}
\end{subequations}
Taking the scalar
  product in $L^2(\Omega,\R^2)$  of (\ref{eq:18}a) with $\nabla \phi_1$ yields
  after integration by parts, with the aid of \eqref{eq:4},
\begin{equation}
\label{eq:19}
  \|\nabla\phi_1(t,\cdot)\|_2 \leq \frac{1}{\sigma}\|\Im(\bar\psi\nabla_A\psi)(t,\cdot)\|_2 \,.
\end{equation}

We next multiply (\ref{eq:18}a) by $A_1$ and integrate over $\Omega$.  Observing that by
  {   (\ref{eq:1v2}e)}, \eqref{eq:8}, and \eqref{eq:17} we have 
\begin{equation}
\label{eq:151}
 \curl A_1(t,x)=0 \mbox{ for any } (t,x) \in \mathbb R^+\times  \pa \Omega\,,
  \end{equation}
we obtain
\begin{displaymath}
 \frac{1}{2} \sigma\frac{d\, \|A_1(t,\cdot)\|_2^2}{dt}  + \kappa^2\, \|\curl A_1(t,\cdot)\|_2^2 = \langle
  A_1(t,\cdot)\,,\,\Im(\bar\psi\nabla_A\psi)(t,\cdot) -  \sigma\, \nabla\phi_1(t,\cdot) \rangle \,,
\end{displaymath}
where $\langle\cdot,\cdot\rangle$ denotes the $L^2(\Omega,\R^2)$ inner product. In view of
\eqref{eq:19} and \eqref{eq:13} we have, as $\|\psi\|_\infty\leq 1$,
\begin{equation}
\label{eq:20}
   \frac{1}{2}  \frac{d\, \|A_1(t,\cdot)\|_2^2}{dt}  + \lambda c\, \|A_1(t,\cdot)\|_2^2 \leq   \frac{2c}{\kappa^2}\;  \|A_1(t,\cdot)\|_2\,\|\nabla_A\psi(t,\cdot)\|_2 \,.
\end{equation}
With the aid of Cauchy's inequality, we then obtain for any $\alpha >0$
\begin{equation}
\label{eq:21}
     \frac{1}{2} \frac{d\, \|A_1(t,\cdot)\|_2^2}{dt}  + (\lambda c - \frac{c \alpha }{\kappa^2})  \, \|A_1(t,\cdot)\|_2^2 \leq
    \frac{c}{\alpha \kappa^2}\, \|\nabla_A\psi(t,\cdot)\|_2^2 \,,
  \end{equation}
  where $\lambda$ has been introduced in \eqref{eq:13}.\\
For later reference we note that by setting  $\alpha = \kappa^2 \lambda \,$,  we
obtain a weaker estimate:
\begin{equation}
\label{eq:21c}
     \frac{1}{2} \frac{d\, \|A_1(t,\cdot) \|_2^2}{dt}   \leq
    \frac{c}{\lambda\kappa^4}\, \|\nabla_A\psi(t,\cdot) \|_2^2 \,.
\end{equation}
Multiplying (\ref{eq:1}a) by $\bar{\psi}$   and integrating by parts
we obtain for the real part
\begin{equation}
\label{eq:22}
    \frac{1}{2}  \frac{d\, \|\psi(t,\cdot)\|_2^2}{dt}  + \|\nabla_A\psi (t,\cdot)\|_2^2 \leq  \|\psi (t,\cdot)\|_2^2 \,.
\end{equation}
From this, with the aid  of Cauchy's inequality, we
obtain that,  for all $\epsilon>0$,
\begin{displaymath}
  \frac{1}{2}  \frac{d\,\|\psi\|_2^2}{dt}  + (1-\epsilon) \|\nabla_{hA_n}\psi\|_2^2 \leq  \|\psi\|_2^2
  +\frac{1}{\epsilon} \|A_1\psi\|_2^2 \,.
\end{displaymath}
From \eqref{eq:11} we then conclude
\begin{equation}
\label{eq:23}
  \frac{1}{2}  \frac{d\, \|\psi\|_2^2}{dt}  + [\mu(h) (1-\epsilon)-1] \|\psi\|_2^2 \leq  
  \frac{1}{\epsilon} \|A_1\psi\|_2^2 \,.
\end{equation}
Using \eqref{eq:2.2a}, we infer from the above that
\begin{equation}
\label{eq:24}
  \frac{1}{2}  \frac{d\, \|\psi\|_2^2}{dt}  + [\mu(h) (1-\epsilon)-1] \|\psi\|_2^2 \leq  
  \frac{1}{\epsilon} \|A_1\|_2^2 \,.
\end{equation}
We next combine \eqref{eq:21} and \eqref{eq:22}, so that $\|\nabla _A
\psi\|^2$ is eliminated, to obtain
\begin{equation}\label{estu}
 \frac{d}{dt}\Big(\|A_1\|_2^2+\frac{c}{\alpha\kappa^2}\|\psi\|_2^2\Big) +
2 ( \lambda c  -  \frac {c\alpha}{\kappa^2}) \|A_1\|_2^2  \leq \frac{2c}{\alpha\kappa^2} \, \|\psi\|_2^2\,.
\end{equation}
Then, substituting into the above  the variable transformation
\begin{equation}
\label{eq:25}
  u(t)= \|A_1(t,\cdot) \|_2^2+\frac{c}{\alpha\kappa^2}\|\psi(t,\cdot)\|_2^2 \quad ; \quad v(t) = \|\psi(t,\cdot)\|_2^2
  \,,
\end{equation}
we obtain the following vector inequality
\begin{equation}
\label{eq:26}
   \frac{d}{dt}
   \begin{bmatrix}
     u \\
     v
   \end{bmatrix}
\leq 
\begin{bmatrix}
  - 2\lambda c + \frac{2c \alpha }{\kappa^2}   &  \frac{2c(\lambda c+1)}{\alpha \kappa^2}  - \frac{2c^2}{\kappa^4}  \\[1.5ex]
\frac{2}{\epsilon} & -2[\mu(h)(1-\epsilon)-1]  -\frac{2c}{\epsilon \alpha \kappa^2}
\end{bmatrix}
\begin{bmatrix}
  u \\
  v
\end{bmatrix}\,,
\end{equation}
which is merely an alternative representation of \eqref{eq:24} and
\eqref{estu}.  Denote the matrix on the right-hand-side by $M(\epsilon,\alpha)$.  
Next, we consider for $a>0$ the quantity   $\langle (au,v) |M(\epsilon,\alpha)  |(u,v) \rangle $ which can
be represented in the form   $\langle (a^\frac 12 u,v) | M_a(\epsilon,\alpha)   |(a^\frac 12 u,v) \rangle $, where
\begin{equation*}\label{eq:29}
M_a(\epsilon,\alpha):=
\begin{bmatrix}
  - 2\lambda c + \frac{2c \alpha }{\kappa^2}   &a^\frac 12 \left(  \frac{2c(\lambda c+1)}{\alpha \kappa^2}  - \frac{2c^2}{\kappa^4} \right) \\[1.5ex]
\frac{2}{\epsilon} a^{-\frac 12} & -2[\mu(h)(1-\epsilon)-1]  -\frac{2c}{\epsilon \alpha \kappa^2}
\end{bmatrix} \,.
\end{equation*}
We then choose such $a$, for which $ M_a$ is symmetric.  Suppose that 
\begin{equation}\label{careful}
\alpha c < (\lambda c +1)\kappa^2 \,.
\end{equation}
We then set
\begin{equation}
\label{eq:85}
  a^{-1} = \epsilon \, \left(  \frac{c(\lambda c+1)}{\alpha \kappa^2}  - \frac{c^2}{\kappa^4} \right)\,.
\end{equation}
Let $\hat \lambda_1(\epsilon,\alpha) $ and $\hat \lambda_2(\epsilon,\alpha) $ denote the eigenvalues
of $M_a$, which are identical with the eigenvalues of $M$.
Without loss of generality we may assume
that 
$$\hat \lambda_1 \leq \hat \lambda_2\,.
$$
We obtain that 
\begin{equation*}
  \frac{1}{2} \frac{d}{dt}(a|u|^2 +|v|^2) \leq \hat \lambda_2 \, (a|u|^2 + |v|^2) \,.
\end{equation*}
This proves that
$$
a|u(t) |^2 +|v(t) |^2\leq e^{2 \hat \lambda_2 t} (a|u(0) |^2 +|v(0) |^2)\,,
$$
which  implies
 \eqref{eq:15} with 
\begin{equation}\label{deflambdam} 
\lambda_{m} =- \hat \lambda_2\,.
\end{equation}
One can easily obtain from the above an explicit formula for $
C(\lambda,\kappa,\|\psi_0\|_2,\|A_0\|_2,h)$.
\\
We now determine the conditions under which $\hat \lambda_2 < 0$. We search
for the optimal values of $\alpha$ and $\epsilon$ which achieve that goal. We
separately obtain conditions for which $\hat{\lambda}_2$ is negative for
arbitrary values of $\kappa$ and in the large $\kappa$ limit. We first obtain
results of the former type that are, naturally, non-optimal. Then, we
consider the regime $\kappa \ar + \infty$ where we obtain results that are
asymptotically close to those of the linearized problem discussed in
the previous subsection.

{\em Arbitrary $\kappa$}~\\
We set
\begin{equation}
 \label{eq:27}
  \epsilon= \frac{2b}{\kappa^2}\,,\,  \alpha = \frac{\lambda}{2} \kappa^2\,,
\end{equation}
where
\begin{displaymath}
  b= \Big[\frac{c}{2\mu\lambda}\Big]^{1/2} \,.
 \end{displaymath}
(Note that \eqref{careful} is satisfied for the above value of $\alpha$.)
 From \eqref{eq:85} we get 
 \begin{displaymath}
  a=\frac{\lambda\kappa^6}{2b(\lambda c^2+2c)} \,,
\end{displaymath}
and hence,
\begin{displaymath}
  M_a = \begin{bmatrix}
-\lambda c  & \frac{\beta}{\kappa}  \\[1.5ex]
\frac{\beta}{\kappa} & - 2(\mu-1)
\end{bmatrix}\,,
\end{displaymath}
where
\begin{displaymath}
  \beta^2 = 2  \frac{\lambda c^2 +2c}{\lambda b}
\end{displaymath}
Clearly, under Assumption 
\eqref{eq:14},  $\mu >1$ and ${\rm Tr}(M_a)<0$. Hence for $M_a$ to be
strictly negative we must require that $\det
M_a >0$. As
\begin{displaymath}
  \det  M_a =  2\lambda c(\mu-1) - \frac{\beta^2}{\kappa^2} = \det M  \,. 
\end{displaymath}
A straightforward computation  shows that \eqref{eq:14} guarantees
that $\det M_a$ is positive. The rate of decay $\lambda_m$ can now be
obtain from \eqref{deflambdam} and 
\begin{displaymath}
 \hat  \lambda_2 = \frac{{\rm Tr}(M)}{2} +  \sqrt{ \frac{{\rm
        Tr}(M)^2}{4}-\det  M } \,.
\end{displaymath}

{\em  Large $\kappa$}~\\
We now compute the asymptotics of $\hat \lambda_2$ as $\kappa
\ar +\infty$. In this limit  we can distinguish between three zones:
\begin{enumerate}
\item $0<\mu-1 <\lambda c$\,;
\item $\mu-1 =\lambda c$\,;
\item $\mu -1 > \lambda c$\,.
\end{enumerate} 
In cases 1 and 3,  $|\mu-1-\lambda c|$ (as well as $c$ and $\lambda$) must be
bounded away from $0$ as $\kappa\to\infty$. In all cases we search for the
values of $\epsilon$ and $\alpha$ for which the minimal value of $\hat \lambda_2$ is
obtained and then prove that it is negative.

{\em The case   $\lambda c\,< \mu -1$ }~\\

We first note that in view of the discussion in Subsection \ref{ss3.1} one cannot have
$\lambda_2<\lambda c\,$.  For convenience we divide all elements of $M_a$ by $2$ to
obtain the matrix
$$
\begin{bmatrix}
  - \lambda c + \frac{c \alpha }{\kappa^2}   &  \frac{c(\lambda c+1)}{\alpha \kappa^2}  - \frac{c^2}{\kappa^4}  \\[1.5ex]
\frac{1}{\epsilon} & -[\mu(h)(1-\epsilon)-1]  -\frac{c}{\epsilon \alpha \kappa^2}
\end{bmatrix}
= 
\begin{bmatrix}
  \rho_1   &  \frac{c(\lambda c+1)}{\alpha \kappa^2}  - \frac{c^2}{\kappa^4}  \\[1.5ex]
\frac{1}{\epsilon} & \rho_2
\end{bmatrix}
$$
whose characteristic polynomial in $\nu$ is given by
\begin{equation}\label{rho1rho2a}
(\rho_1-\nu) (\rho_2 - \nu) = \frac 1 \epsilon \left( \frac{c(\lambda c+1)}{\alpha \kappa^2}  - \frac{c^2}{\kappa^4} \right)\,.
\end{equation}
Suppose that $\epsilon$ and $\alpha$ are chosen so that the right-hand-side
vanishes as $\kappa\to\infty$. Then $\nu$ is confined to a close neighborhood of
$\{\rho_1,\rho_2\}$. We seek an estimate of $\nu$ in the neighborhood of
$\rho_1$. Thus, we neglect the terms $(\rho_1-\nu)^2$ and
$\frac{c^2}{\kappa^4}$ to obtain that
\begin{equation}\label{rho1rho2ba}
(\rho_1-\nu) (\rho_2 - \rho_1) \sim  \frac{c(\lambda c+1)}{\epsilon \alpha \kappa^2}\,.
\end{equation}
We now write
$$
\epsilon (\rho_1-\rho_2) = (\mu -1 - \lambda c)\epsilon  + \frac{c \alpha }{\kappa^2}\epsilon - \mu \epsilon^2 +\frac{c}{ \alpha \kappa^2}\,,
$$
and maximize it with respect to $\epsilon$, for fixed $\alpha$, to
obtain 
$$
2 \mu \epsilon = (\mu -1 - \lambda c) + \frac{c \alpha }{\kappa^2}\sim  (\mu -1 - \lambda c) \,.
$$
Setting $\epsilon = \frac{\mu-1- \lambda c}{2 \mu}$\,, we obtain that
$$
\epsilon (\rho_1-\rho_2)= \frac{(\mu-1- \lambda c)^2}{4  \mu} + \frac{c \alpha  (\mu-1 - \lambda c)}{2 \mu \kappa^2}+  \frac{c}{ \alpha \kappa^2} \,,
$$
and hence
$$
\nu  \sim - \lambda c + \frac{c\alpha} {\kappa^2} +\frac{c(\lambda c+1)}{\alpha \kappa^2}  \frac{4 \mu}{(\mu-1- \lambda c)^2} \,.
$$
Minimizing $\nu$ over $\alpha$ then yields
$$
\nu = - \lambda c  + \frac{cd}{\kappa^2} + \Og (\kappa^{-4})\,,
$$
where
$$
d = \inf_\alpha \left( \alpha + \frac{4 \mu (\lambda c+1)}{(\mu-1- \lambda c)^2} \frac 1 \alpha  \right)\,.
$$
This leads to 
$$
\alpha = \sqrt{ \frac{4 \mu (\lambda c+1)}{(\mu-1- \lambda c)^2}}
$$
and to the value of $d$ given in  \eqref{calculded}.

Once $\alpha$ and $\epsilon$ had been set, it can be easily established that
$$
\hat \lambda_2= - 2 \lambda c +  \frac{2 cd}{\kappa^2} + \Og (\kappa^{-4})\,,
$$
which verifies \eqref{eq:152c}.
Moreover, by \eqref{eq:85} we obtain that
$$
a \sim \frac{4 \mu^{\frac 32} }{c (\mu-1-\lambda c)^2 (\lambda c +1)^{\frac 12}}  \, \kappa^2 \,.
$$

{\em The case: $ \lambda c > \mu-1$}~\\
In this case, since $\rho_2>\rho_1$, we look for the solution of
\eqref{rho1rho2a} which lies in the close vicinity of
$\rho_2$. Neglecting the term $(\rho_2-\nu)^2$ leads to
\begin{equation}\label{rho1rho2b2}
(\rho_2 -\nu) (\rho_1 - \rho_2) =  \frac{c(\lambda c+1)}{\epsilon \alpha \kappa^2} - c^2 \epsilon^{-1}  \kappa^{-4} \,.
\end{equation}
We next set $$\alpha = \hat \alpha \kappa^2\mbox{ and }\epsilon = \hat \epsilon \kappa^{-2}\,,$$
which is in line with the choice we made for arbitrary values of
$\kappa$. Then we can rewrite \eqref{rho1rho2b2} in the form
\begin{equation}\label{rho1rho2b3}
(\rho_2 -\nu) (\rho_1 - \rho_2) \sim  \frac{c(\lambda c+1)}{\hat \epsilon  \hat  \alpha} \kappa^{-2}  - c^2 \hat \epsilon^{-1} \kappa^{-2}\,.
\end{equation}
It follows that
$$
\nu \sim  -[\mu(h)-1]  + \hat \epsilon \mu \kappa^{-2}  -\frac{c}{\hat \epsilon \hat \alpha } \kappa^{-2} +  \frac{c(\lambda c+1)- c^2 \hat \alpha }{(\hat \epsilon  \hat  \alpha) (\lambda c  - \mu +1- c \hat \alpha)  }   \kappa^{-2}\,.
$$
In view of the above we have to minimize the coefficient of $\kappa^{-2}$
which is given by
$$
q(\hat \alpha, \hat \epsilon):=  \hat \epsilon \mu   +\frac{c \mu  }{(\hat \epsilon  \hat  \alpha) (\lambda c  - \mu +1- c \hat \alpha)  }  \,.
$$
To minimize $q$, we first minimize over $\alpha$ to obtain
$$\hat \alpha = \frac{ 1}{2c} ( \lambda c  - \mu +1)$$
(note that \eqref{careful} is satisfied), which leads to the minimization, this time over $\epsilon$,  of
$$
 \hat \epsilon \mu   +\frac{4c^2 \mu  }{\hat \epsilon  (\lambda c  - \mu +1)^2  } 
$$
and finally to
$$
\inf_{\hat \epsilon,\hat \alpha} q(\hat \alpha, \hat \epsilon) = 4c \mu (\lambda c  - \mu +1) ^{-1}\,,
$$
where the minimum is obtained for
$$
\hat \epsilon = \frac {2c}{\lambda c -\mu +1}\,.
$$
Hence, we have shown that when $\mu-1 <  \lambda c$ then for $\kappa$ large
enough, the largest eigenvalue $\hat \lambda_2$ satisfies 
\begin{equation}
\hat \lambda_2 = - 2 (\mu -1) +8c \mu (\lambda c  - \mu +1) ^{-1} \kappa^{-2} +\Og(\kappa^{-4})\,. 
\end{equation}

A simple computation shows that there exists a
constant $C(\lambda,\mu,c)$ and $\kappa_0 (\lambda,\mu,c)$ such that if $\kappa\geq \kappa_0$
and
$$
\mu >1 + \frac{4}{\lambda\kappa^2}  +  C \kappa^{-4}\,,
$$
then $\hat \lambda_2<0\,$. \\
Note that the above condition is weaker that \eqref{eq:14} in the
sense that $4/\lambda\leq\gamma$, but unlike \eqref{eq:14} its validity is
limited for large values $\kappa$. \\ 
The asymptotic behavior of $a$ in this case is given by
$$
a^{-1} \sim  2c^3 \, (\lambda c + \mu +1)^{-1} \,  \kappa^{-6}\,,
$$
which is of the same magnitude as in the estimate for arbitrary values 
of $\kappa$.\\

{\em The case $\mu-1 = \lambda c$}~\\
In this case we search for the optimal eigenvalues of
$$
\begin{bmatrix}
  - \lambda c + \frac{c \alpha }{\kappa^2}   &  \frac{c(\lambda c+1)}{\alpha \kappa^2}  - \frac{c^2}{\kappa^4}  \\[1.5ex]
\frac{1}{\epsilon} & -[(\lambda c +1)(1-\epsilon)-1]  -\frac{c}{\epsilon \alpha \kappa^2}
\end{bmatrix}
= 
\begin{bmatrix}
  \rho_1   &  \frac{c(\lambda c+1)}{\alpha \kappa^2}  - \frac{c^2}{\kappa^4}  \\[1.5ex]
\frac{1}{\epsilon} & \rho_2
\end{bmatrix}
$$
whose characteristic equation is given by
\begin{equation}\label{rho1rho2a1}
(\rho_1-\nu) (\rho_2 - \nu) = \frac 1 \epsilon \left( \frac{c(\lambda c+1)}{\alpha \kappa^2}  - \frac{c^2}{\kappa^4} \right)\,.
\end{equation}
Here we set $\alpha = \tilde \alpha \kappa$ and $\epsilon = \tilde \epsilon \kappa^{-1}\,$. Then,
neglecting lower order terms  the characteristic equation becomes
$$
(\nu + \lambda c - \frac{c \tilde \alpha}{\kappa}) (\nu + \lambda c - \frac{\tilde \epsilon(\lambda c+1)}{\kappa}) \sim  \frac{c(\lambda c +1)}{\tilde \alpha \tilde \epsilon \kappa^2}\,.
$$
Hence,
$$
\nu \sim - \lambda c  + \frac 1 \kappa \left(  \frac 12( c \tilde \alpha +  \tilde \epsilon(\lambda c+1) ) + \sqrt{ \frac{c (\lambda c +1)}{\tilde \alpha \tilde \epsilon} + \frac 14 ( c \tilde \alpha +  \tilde \epsilon(\lambda c+1))^2}\right)\,.
$$
The minimum of the coefficient of $\kappa^{-1}\,$ is achieved for $$c
\tilde \alpha = \tilde \epsilon (\lambda c +1)= ( c(\lambda c +1) )^{\frac 12} 3^{-\frac
  14}\,,$$ leading to 
 $$
\nu = - \lambda c  + \frac{ 2\, c^{\frac 12}   \, 3^{-\frac 14} \,  (\lambda c +1)^{\frac 12} }{ \kappa } + \Og (\kappa^{-2})\,.
$$
The asymptotic behavior of $a$ is given by
$$
a \sim  c^{-2}  \, (\lambda c + \mu +1) \, (\lambda c+1) \,  \kappa^{4}\,.
$$
This completes the proof of the theorem.
\qed

\subsection{$L^2$ estimate of $A_1(t,\cdot)$}~\\
We continue this section  by the following simple bound for
$\|A_1(t,\cdot)\|_2$.
\begin{lemma}
 Let $A_1$ be defined by \eqref{eq:17}. Then, for any positive values of
 $h$, $\kappa$, $c$, and $\lambda$,  we have
  \begin{equation}
    \label{eq:43}
\|A_1(t,\cdot) \|_2^2 \leq  \|A_1(0,\cdot) \|_2^2 \, e^{-\lambda ct}  + \frac{2c}{\lambda\kappa^4}\ +  \Big(\frac{2c}{\lambda \kappa^4} +
   \frac{4}{\lambda^2\kappa^4}\Big) \,|\Omega| \,.
  \end{equation}
\end{lemma}
\begin{proof}
 By \eqref{eq:26} with the choice \eqref{eq:27} for $(\alpha,\epsilon)$ and
 Gronwall's inequality we have, with $v$ and $u$ introduced in
 \eqref{eq:25}, 
 \begin{equation}
\label{eq:43a}
   u(t)  \leq u(0) e^{-\lambda ct} +  \Big(\frac{2 c^2}{\kappa^4} +
   \frac{4c}{\lambda\kappa^4}\Big)\int_0^t
   e^{-\lambda c(t-\tau)}v(\tau)\,d\tau \,.
 \end{equation}
By \eqref{eq:2.2a} we have $v(t)\leq |\Omega|$ and  hence
  \begin{equation}\label{eq:43b}
   u \leq u(0) e^{-\lambda ct} +  \Big(\frac{2c}{\lambda \kappa^4} +
   \frac{4}{\lambda^2\kappa^4}\Big) \,|\Omega|\,.
 \end{equation}
 The lemma readily follows.   
\end{proof}

\subsection{ $L^\infty$ decay of $A_1(t,\cdot)$.}\label{ss3.4}~\\
We now show how to deduce from an $L^2$-estimate of $A_1$, an
  $L^\infty$ estimate for it.
\begin{proposition}
\label{lemma3.1}
Let  $M >0$, $(\psi,A,\phi)$ denote a solution of \eqref{eq:1} and
\eqref{eq:4}, and  $A_1$ be defined
by \eqref{eq:17}. We further require that
\begin{equation}
\label{a1init}
\| A_1(0,\cdot)\|_2 \leq M\,.
\end{equation}
 Let further,  for any $\kappa\in (0,+\infty)$,
\begin{equation}\label{eq:145}
t^*(\kappa,M) =\max \left(  \frac{2}{\lambda c} \ln (\kappa^2 M ),1 \right) \,.
\end{equation}
For any  $\kappa_0 >0$, and any $\delta \in (0,1)$, there exists $C=C(\Omega,c,\delta,\kappa_0)>0$, such that,
for all $\kappa \in [\kappa_0,+\infty)$, and $t\geq t^*(\kappa,M)+2\delta$, we have
      \begin{equation}
    \label{eq:30}
\|A_1(t, \cdot)\|_\infty \leq C \Big(\|A_1(t-\delta,\cdot )\|_2 +\frac{1}{\kappa^2}\|\psi(t-\delta,\cdot )\|_2\Big) \,.
  \end{equation}
\end{proposition}
\begin{remark} \label{remsupp}
Note that in view of \eqref{eq:145},  
 $A_1(t,\cdot)$ satisfies for $t\geq t^*(\kappa,M)$ 
\begin{equation}
\label{eq:159}
  \|A_1(t,\cdot)\|_2\leq \frac{C(\Omega,c)}{\kappa^2} \,.
\end{equation}
 \end{remark}

\begin{proof}~\\
  Let $0<\tilde{\delta}<1/3$, and $t_0\geq t^*(\kappa,M)$.  Let further
  $T=t_0+4\tilde{\delta}$, $t_1=t_0+\tilde{\delta}$, $t_2=t_0+2\tilde{\delta}$,
  and $t_3=t_0+3\tilde{\delta}$. Hence we have $0<t_0<t_1<t_2 <T$.  We
  begin by applying \eqref{eq:147abstract} in the case of $\mathcal
  L^{(1)}$, with $X=A_1$ on the interval $(t_0,T)$ and 
  $X_0=A_1(t_0,\cdot)$ (recalling that $A_1$ satisfies \eqref{eq:4}),
  Sobolev embedding then yields
\begin{equation}\label{eq:3.25}
       \|A_1\|_{L^2(t_1,T; L^\infty(\Omega,\mathbb R^2))}^2 +\|A_1\|_{L^\infty(t_1,T; H^1(\Omega,\mathbb R^2))}^2
      \leq C\Big[\frac{1}{\kappa^2}\|\Im(\bar{\psi}\nabla_A\psi)\|_{L^2(t_0,T;
        L^2(\Omega,\C))}^2+ \|A_1(t_0,\cdot )\|_2^2 \Big]\,,
\end{equation}
for some $C>0$ depending on $\tilde{\delta}$, $c$, and $\Omega$.  Integrating \eqref{eq:22} on $(t_0,T)$
  yields, in view of \eqref{eq:2.2a}, that
\begin{multline}\label{eq:16a} 
  \|\Im(\bar{\psi}\nabla_A\psi)\|_{L^2(t_0,T; L^2(\Omega,\C))}^2  \leq
  \|\nabla_A\psi\|_{L^2(t_0,T; L^2(\Omega,\C))}^2 \\\leq
  \int_{t_0}^T\Big[\|\psi(\tau, \cdot )\|_2^2  -\frac{1}{2}  \frac{d\,
    \|\psi(\tau, \cdot )\|_2^2}{d\tau}\Big] \,d\tau  \leq \|\psi\|_{L^2(t_0,T;L^2(\Omega,\C))}^2 + \frac 12
      \|\psi(t_0,\cdot )\|_2^2\,.
\end{multline}
We then obtain that
\begin{displaymath}
\|A_1\|^2_{L^2(t_1,T; L^\infty(\Omega))} + \|A_1\|_{L^\infty(t_1,T; H^1(\Omega))}^2
          \leq C   \Big[\frac{1}{\kappa^2}  \|\psi\|_{L^2(t_0,T;L^2(\Omega,\C))}^2+\|A_1(t_0,\cdot )\|_2^2 +
      \frac{1}{\kappa^2 } \|\psi(t_0,\cdot )\|_2^2\Big]\,.
      \end{displaymath}
Using \eqref{eq:22} again yields for all $t_0\leq\tilde{t}<t\leq T$
\begin{equation}
  \label{eq:16}
\| \psi(t,\cdot )\|_2 \leq e^{(t-\tilde{t})}\| \psi(\tilde{t}, \cdot )\|_2 \leq e^{4\tilde{\delta}}
  \|\psi(\tilde{t},\cdot )\|_2  \,. 
\end{equation}
Hence,
\begin{equation}
\label{eq:149}
  \|\psi\|_{L^2(t_0,T;L^2(\Omega,\C))}^2 \leq  4\tilde{\delta}e^{8\tilde{\delta}}\,
\|\psi(t_0, \cdot )\|_2^2   \,.
\end{equation}
Consequently,
\begin{equation}
\label{eq:31}
\|A_1\|^2_{L^2(t_1,T; L^\infty(\Omega))} + \|A_1\|_{L^\infty(t_1,T; H^1(\Omega))}^2
          \leq C (\tilde \delta,c,\Omega)\,  \Big[\|A_1(t_0,\cdot )\|_2^2 +
      \frac{1}{\kappa^2 } \|\psi(t_0,\cdot )\|_2^2\Big]\,.
\end{equation}
For later reference we mention that by \eqref{eq:21c} and \eqref{eq:16a}
we have, for all $t_0<t\leq T$,
\begin{equation}
  \label{eq:37}
\|A_1(t,\cdot)\|_2^2 \leq \|A_1(t_0,\cdot)\|_2^2 + \frac{C(\tilde{\delta},c,\Omega)}{\kappa^4}\|\psi(t_0,\cdot)\|_2^2
\,.
\end{equation}

Let $t_1\leq\tilde{t}_1\leq t_2$. We next apply \eqref{eq:148abstracta} in
the case of the Dirichlet-Neumann problem to (\ref{eq:1}a) on the
interval $(\tilde{t}_1,s)$, for any $\tilde{t}_1<s\leq T$, to obtain that
\begin{multline}
\label{eq:32}
  \|\psi\|_{L^\infty(\tilde{t}_1,s;H^1(\Omega,\C))} 
  \leq C(\Omega)\,
  \big[\|\psi(\tilde{t}_1,\cdot )\|_{1,2} + \|A\cdot\nabla\psi\|_{L^2(\tilde{t}_1,s; L^2(\Omega,\C))}+\\ \||A|^2\psi\|_{L^2(\tilde{t}_1
    ,s;L^2(\Omega,\C))}+\|\phi\psi\|_{L^2(\tilde{t}_1
    ,s; L^2(\Omega,\C))} + \| \psi(1-|\psi|^2)\|_{L^2(\tilde{t}_1,s;
    L^2(\Omega,\C))} \big]\,.
 \end{multline}
 
For the last term inside the brackets on the right-hand-side we obtain
with the aid of \eqref{eq:2.2a} and \eqref{eq:16} that 
\begin{equation}
  \label{eq:155}
\| \psi(1-|\psi|^2)\|_{L^2(\tilde{t}_1,s;
    L^2(\Omega,\C))}^2  \leq 3\,\tilde{\delta} \, e^{8\tilde{\delta}}\; \|\psi(\tilde t_1, \cdot)\|_2^2\,.
\end{equation}

For the second  term inside  the brackets on the right-hand-side we have
\begin{equation}
\label{eq:33}
  \|A\cdot\nabla\psi\|_{L^2(\tilde{t}_1,s;L^2(\Omega))}^2 \leq 2
  \int_{\tilde{t}_1}^s\left(\|A_n \|_\infty^2+\|A_1(\tau,\cdot )\|_\infty^2\right)\,\|\nabla\psi(\tau,\cdot )\|_2^2 \,d\tau\,.
\end{equation}
For the third term inside the brackets we have, using Sobolev
embedding and the fact that $A_n$ is bounded in $L^8(\Omega\,;\mathbb R^2)$
by a constant depending only on $c$ and $\Omega$,
\begin{multline*}
 \| |A|^2\psi\|_{L^2(\tilde{t}_1,s,L^2(\Omega,\C))}^2 \leq C(\Omega) \, \int_{\tilde{t}_1}^s (\|A_n\|_8^4+ 
 \|A_1(\tau,\cdot )\|_8^4)\|\psi(\tau,\cdot )\|_4^2 \,d\tau\\ \leq  \widehat{C}
 (c,\Omega)\, \Big[
 \|\psi\|_{L^2(\tilde{t}_1,s; H^1(\Omega,\C))}^2 + \int_{\tilde{t}_1}^s \|A_1(\tau,\cdot)\|_{1,2}^4\, \|\psi(\tau,\cdot )\|_{1,2}^2 \,d\tau\Big] \,.
\end{multline*}
With the aid of \eqref{eq:31}  we then obtain that
\begin{multline*}
   \| |A|^2\psi\|_{L^2(\tilde{t}_1,s; L^2(\Omega,\C))}^2\leq \\ \quad\quad \leq C(\tilde \delta,c,\Omega)\, \Big\{
 \|\psi\|_{L^2(\tilde{t}_1,s;H^1(\Omega,\C))}^2 +\Big[\|A_1(t_0,\cdot)\|_2^2 +
      \frac{1}{\kappa^2}\|\psi(t_0,\cdot)\|_2^2\Big]^2 \int_{\tilde{t}_1}^s \|\psi(\tau,\cdot)\|_{1,2}^2
      \,d\tau\Big\} 
\end{multline*}
 From \eqref{eq:159}, we then get
\begin{equation}
\label{eq:34}
   \| |A|^2\psi\|_{L^2(\tilde{t}_1,s; L^2(\Omega,\C))}^2\leq C(\tilde \delta,c,\Omega,\kappa_0)\,  \|\psi\|_{L^2(\tilde{t}_1,s;H^1(\Omega,\C))}^2 \,.
\end{equation}

 For the fourth term on the right-hand-side of \eqref{eq:32} we obtain,
using \eqref{eq:2.2a} and Sobolev embedding
\begin{multline}
\label{eq:157}
  \|\phi\psi\|_{L^2(\tilde{t}_1,s;L^2(\Omega,\C))}^2 \leq 2\, 
  \|\phi_n\|_\infty^2\|\psi\|_{L^2(\tilde{t}_1,s;L^2(\Omega,\C))}^2 + 2 \int_{\tilde{t}_1}^s
  \|\phi_1(\tau,\cdot)\|_4^2  \|\psi(\tau,\cdot)\|_4^2 \,d\tau  \\ \leq  C(\tilde \delta,c,\Omega,\kappa_0)\,  \left( \|\psi\|_{L^2(\tilde{t}_1,s;L^2(\Omega,\C))}^2+ \int_{\tilde{t}_1}^s
  \|\phi_1(\tau,\cdot)\|_{1,2}^2   \,d\tau \right) \,.
\end{multline}

To obtain the last inequality we had to use \eqref{eq:9a}, to conclude
that $$\| \phi_n\|_\infty \leq C(\Omega) /\sigma \leq \hat C (c,\Omega,\kappa_0)\,,$$ for all $\kappa \geq
\kappa_0$.  Using \eqref{eq:19} and Poincar\'e's inequality for $\phi_1$
 then yields,
\begin{displaymath}
   \|\phi\psi\|_{L^2(\tilde{t}_1,s;L^2(\Omega,\C))}^2 \leq
    C(\tilde \delta,c,\Omega,\kappa_0) \,   \, \Big[\|\psi\|^2_{L^2(\tilde{t}_1,s;L^2(\Omega,\C))}+ \frac{1}{\kappa^2}\int_{\tilde{t}_1}^s
  \|\nabla_A\psi(\tau,\cdot)\|_2^2   \,d\tau \Big] \,. 
\end{displaymath}
By \eqref{eq:22} (or, more precisely, its integrated version over
$(t_1,s)$), we then obtain
\begin{displaymath}
   \|\phi\psi\|_{L^2(\tilde{t}_1,s;L^2(\Omega,\C))}^2 \leq  C(\tilde
   \delta,c,\Omega,\kappa_0)\, \big[\|\psi\|_{L^2(\tilde{t}_1,s;L^2(\Omega,\C))}^2
   + \|\psi(\tilde{t}_1,\cdot)\|_2^2 \big] \,,
\end{displaymath}
which together with  \eqref{eq:16}  gives way to
\begin{equation}\label{eq:157a}
   \|\phi\psi\|_{L^2(\tilde{t}_1,s;L^2(\Omega,\C))} \leq
    C(\tilde \delta,c,\Omega,\kappa_0)\,  \,  \|\psi(\tilde{t}_1,\cdot)\|_2 \,.
\end{equation}
Substituting \eqref{eq:33}, \eqref{eq:34}, and \eqref{eq:157a}  into  \eqref{eq:32} then yields
\begin{displaymath}
   \|\psi\|_{L^\infty(\tilde{t}_1,s;H^1(\Omega,\C))}^2 \leq  C(\tilde \delta,c,\Omega,\kappa_0)\,   \, \big[
   \int_{\tilde{t}_1}^s(1+\|A_1(\tau,\cdot)\|_\infty^2)\|\psi(\tau,\cdot)\|_{1,2}^2 \,d\tau + \|\psi(\tilde{t}_1,\cdot)\|_{1,2}^2 \big]\,.
\end{displaymath}
We now apply a variant  of
  Gronwall's inequality.\\  Let $f(t)=\|\psi(t,\cdot)\|_{1,2}^2$,
  $g(t)=C(\tilde \delta,c,\Omega,\kappa_0)(1+\|A_1(t,\cdot)\|_\infty^2)$, and $C_0= C(\tilde
  \delta,c,\Omega,\kappa_0)\|\psi(\tilde{t}_1,\cdot)\|_{1,2}^2$. By the inequality
  below \eqref{eq:157a} we obtain that
  \begin{displaymath}
    f(s) \leq C_0 + \int_{\tilde{t}_1}^s f(\tau)g(\tau)\,d\tau \quad
    \forall\tilde{t}_1<s\leq T\,. 
  \end{displaymath}
One can now apply Gronwall's inequality (Theorem III.1.1 in
\cite{ha02}) to obtain that
\begin{displaymath}
    f(s) \leq C_0 \exp \Big\{  \int_{\tilde{t}_1}^s g(\tau)\,d\tau \Big\} \leq
    C_0 \exp \Big\{  \int_{\tilde{t}_1}^T g (\tau)\,d\tau \Big\} \,,
\end{displaymath}
for all $\tilde{t}_1<s\leq T$. By
taking the supremum over $s\in(\tilde{t}_1,T]$, we obtain:
\begin{displaymath}
   \|\psi\|^2_{L^\infty(\tilde{t}_1,T;H^1(\Omega,\C))} \leq  C \|\psi(\tilde{t}_1,\cdot)\|_{1,2}^2 \exp
   \Big\{ C\int_{\tilde{t}_1}^T(1+\|A_1(\tau,\cdot)\|_\infty^2) \,d\tau\Big\} \,,
\end{displaymath}
which together with \eqref{eq:31} and  \eqref{eq:159} yields (recall
that $t_0\geq t^*(\kappa,M)$),
\begin{equation}
\label{eq:156}
  \|\psi\|_{L^\infty(\tilde{t}_1,T;H^1(\Omega,\C))}  \leq C(\tilde \delta,\Omega,\kappa_0)\,  \|\psi(\tilde{t}_1,\cdot)\|_{1,2}\,.
\end{equation}

To find $\tilde t_1$ for which we can estimate $\|\psi(\tilde{t}_1,\cdot)\|_{1,2}$
we first observe that
\begin{displaymath}
  \|\nabla\psi(\cdot,t)\|_2 \leq \|\nabla_A\psi(t,\cdot)\|_2 + \|A\psi(t,\cdot)\|_2 \leq
  \|\nabla_A\psi(t, \cdot)\|_2 + \|A_n\|_\infty\|\psi(t,\cdot)\|_2+ \|A_1(t,\cdot)\|_2 \,.
\end{displaymath}
Integrating the above between $t_1$ and $t_2$ yields, with the aid of
\eqref{eq:22} that
\begin{displaymath}
    \|\psi\|_{L^2(t_1,t_2;H^1(\Omega,\C))} \leq C\big(
   \|\psi\|_{L^2(t_1,t_2;L^2(\Omega,\C))} +
   \|A_1\|_{L^2(t_1,t_2;L^2(\Omega,\C))} +  \|\psi(t_1, \cdot )\|_2 \big)
\end{displaymath}
With the aid of \eqref{eq:37} and \eqref{eq:16} we then obtain that
\begin{displaymath}
   \|\psi\|_{L^2(t_1,t_2;H^1(\Omega,\C))} \leq C\big( \|A_1(t_1, \cdot )\|_2 +  \|\psi(t_1,\cdot)\|_2 \big)
\end{displaymath}
We can, thus, conclude that there exists $\tilde{t}_1\in [t_1,t_2]$  such that
\begin{displaymath}
   \|\psi(\tilde{t}_1,\cdot)\|_{1,2} \leq C \tilde \delta^{-\frac 12} \big( \|A_1(t_1,\cdot)\|_2 +
   \|\psi(t_1,\cdot )\|_2 \big) \,.
\end{displaymath}
In conjunction with \eqref{eq:156}, \eqref{eq:16}, and \eqref{eq:37}
the above inequality yields the existence of some constant $C$ such
that: 
\begin{equation}
  \label{eq:35}
  \|\psi\|_{L^\infty(t_2,T;H^1(\Omega,\C))}  \leq C \big( \|A_1(t_0, \cdot )\|_2 +
   \|\psi(t_0, \cdot )\|_2 \big) \,.
\end{equation}

Let $t_3=t_2+\tilde{\delta}$. We continue by applying
\eqref{eq:148abstract} (in the case of the operator $\mathcal
L^{(1)}$) to the first line of \eqref{eq:18} (recalling \eqref{eq:4}
and \eqref{eq:151}) in $(t_2,T)$ to obtain, with the aid of
\eqref{eq:19}
\begin{equation}
\label{eq:38}
     \|A_1\|_{L^\infty(t_3,T;L^\infty(\Omega,\R^2))}^2
      \leq C \Big[\frac{1}{\kappa^2}\|\Im(\bar{\psi}\nabla_A\psi)\|_{L^2(t_2,T;H^1(\Omega,\R^2))}^2+\|A_1(t_2,\cdot )\|_2^2 \Big]\,.
\end{equation}
By \eqref{eq:2.2a} we have
\begin{equation}
\label{eq:39}
  \|\nabla\Im(\bar{\psi}\nabla_A\psi)(t,\cdot)\|_2 \leq \|\nabla\psi(t,\cdot)\|_4^2+
  \|D^2\psi(t,\cdot)\|_2 + 2\||A|\nabla\psi(t,\cdot)\|_2+\|\nabla A(t,\cdot)\|_2 \,,
\end{equation}
for all $t\geq t_2$, where $D^2\psi$ denotes the Hessian matrix of $\psi$. \\
By (\ref{eq:148abstracta}) (in the case of the Dirichlet-Neumann
Laplacian) applied to (\ref{eq:1}a) in $(t_2,T)$ we have:
\begin{multline}
\label{eq:41}
      \|\psi\|_{L^2(t_2,T;H^2(\Omega))}  \leq C(\Omega)
     \left[\|A\cdot\nabla\psi\|_{L^2(t_2,T;L^2(\Omega,\C))}+ \|\phi\psi\|_{L^2(t_2
    ,T;L^2(\Omega,\C))}\right. \\ \left.
    + \||A|^2\psi\|_{L^2(t_2
    ,T;L^2(\Omega,\C))}+\|\psi(t_2,\cdot)\|_{1,2} \right]\,.
\end{multline}
Using \eqref{eq:35} yields
\begin{multline*}
\|A\cdot\nabla\psi\|_{L^2(t_2,T;L^2(\Omega,\C))}^2
\leq\|\,|A|\,|\nabla\psi| \,\|_{L^2(t_2,T;L^2(\Omega,\C))}^2\leq \int_{t_2}^T
\|A(\tau, \cdot )\|_\infty^2 \|\nabla\psi(\tau, \cdot )\|_2^2\,d\tau \leq  \\ \|\psi\|_{L^\infty(t_2,T;H^1(\Omega,\C))}\int_{t_2}^T
\|A(\tau, \cdot )\|_\infty^2\,d\tau \leq 
C(\|\psi(t_0,\cdot)\|_2^2+  \|A_1(t_0, \cdot )\|_2^2)\int_{t_2}^T
\|A(\tau, \cdot )\|_\infty^2\,d\tau \,,
\end{multline*}
which together with \eqref{eq:31} and  \eqref{eq:16} yields
\begin{displaymath}
  \|\, |A|\,|\nabla\psi|\,\|_{L^2(t_2,T; L^2(\Omega,\C))}^2 \leq C(\|\psi(t_0,\cdot)\|_2^2+  \|A_1(t_0, \cdot )\|_2^2)[1+\|\psi(t_0,\cdot)\|_2+\|A_1(t_0,\cdot)\|_2]^2\,.
\end{displaymath}
Hence, by \eqref{eq:159} and \eqref{eq:2.2a}  we get
\begin{equation}
\label{eq:42}
 \||A|\nabla\psi\|_{L^2(t_2,T; L^2(\Omega,\C))} \leq C \big( \|A_1(t_0, \cdot )\|_2 +
   \|\psi(t_0, \cdot )\|_2 \big)  \,.
\end{equation}
By \eqref{eq:34} and \eqref{eq:35} we have that
\begin{displaymath}
   \| |A|^2\psi\|_{L^2(t_1,T;L^2(\Omega,\C))}^2 \leq C\, 
   \|\psi\|_{L^2(t_1,T;H^1(\Omega,\C))}^2 \leq 
   \hat  C \, \|\psi(t_0,\cdot)\|_2^2 \,.
\end{displaymath}
Combining the above with \eqref{eq:42}, \eqref{eq:41}, \eqref{eq:16},
\eqref{eq:37}, \eqref{eq:157}, and \eqref{eq:43} yields
\begin{equation}
\label{eq:158}
     \|\psi\|_{L^2(t_2,T;H^2(\Omega))}^2 \leq C \big( \|A_1(t_0, \cdot )\|_2^2 +
   \|\psi(t_0, \cdot )\|_2^2 \big)  \,. 
\end{equation}

For the first term on the right-hand-side of \eqref{eq:39} we have by \eqref{eq:158}
and Sobolev embeddings that
\begin{equation}
\label{eq:40}
   \||\nabla\psi|\|_{L^2(t_2,T,L^4(\Omega))}^2 \leq
   C\, \|\psi\|_{L^2(t_2,T;H^2(\Omega))}^2\leq  \hat{C}\, \big( \|A_1(t_0, \cdot)\|_2^2 +
   \|\psi(t_0, \cdot )\|_2^2 \big) \,.
\end{equation}
The second term on the right-hand-side of \eqref{eq:39} can similarly
be estimated
\begin{displaymath}
   \|D^2\psi\|_{L^2(t_2,T,L^2(\Omega))}^2 \leq  C \big( \|A_1(t_0, \cdot )\|_2^2 +
   \|\psi(t_0, \cdot) \|_2^2 \big) \,.
\end{displaymath}
Combining the above, \eqref{eq:40}, \eqref{eq:42}, \eqref{eq:31}, \eqref{eq:159}, and
\eqref{eq:39} yields
\begin{displaymath}
  \|\Im(\bar{\psi}\nabla_A\psi)\|_{L^2(t_2,T;H^1(\Omega,\R^2))}\leq C \big( \|A_1(t_0, \cdot )\|_2 +
   \|\psi(t_0, \cdot )\|_2 \big) 
  \,. 
\end{displaymath}
Substituting the above into \eqref{eq:38} yields \eqref{eq:30} for $\delta=3\tilde{\delta}$.
\end{proof}

\subsection{Decay estimate for $\phi_1$ and asymptotic contraction}~\\  
We conclude this section by establishing an exponential rate of decay
for $\phi_1$. 
\begin{proposition}
\label{prop:3.4}
Under the same assumptions of Theorem \ref{thm:3.1} and Proposition
\ref{lemma3.1}, there exist $\kappa_0$ and $C(\kappa)$ such that, for $\kappa \geq
\kappa_0$ and  $t\geq t^*(\kappa,M)+1$, where $t^*(\kappa,M)$ is given by
\eqref{eq:145}, we have 
  \begin{equation}
    \label{eq:154}
\|\phi(t,\cdot) -h\phi_n(\cdot) \|_2 \leq C(\kappa) \, e^{-\lambda_mt} \,,
  \end{equation}
and $[t^*(\kappa,M)+1,+\infty)\ni  t\mapsto \|\psi (t,\cdot) \|_2$ is monotone decreasing.
\end{proposition}
\begin{proof}
 Recalling \eqref{eq:19},  we use \eqref{eq:2.2a}, \eqref{eq:35}, and
  \eqref{eq:30} to obtain that, for any $0<\delta<1$, there exists 
  $C(\Omega,\delta)>0$  such that
\begin{multline}
\label{eq:86}
  \|\nabla\phi_1(t,\cdot) \|_2 \\
  \quad \leq C(\Omega,\delta)\big[ \big(\|A_n\|_\infty +\|A_1(t-\delta,\cdot )\|_2 
  +\|\psi(t-\delta,\cdot )\|_2\big)\|\psi(t,\cdot )\|_2 +\qquad \qquad  \\ \|A_1(t-\delta,\cdot)\|_2
  +\|\psi(t-\delta,\cdot )\|_2 \big]  \,. 
  \end{multline}
Using (\ref{eq:15}) for $\psi$ and $A_1$ together with Poincar\'e's
inequality (recall that $\int_\Omega\phi_1(t,x) \,dx=0$) completes the proof of (\ref{eq:154}). \\

We now show that 
$\psi(t,\cdot)$ becomes decreasing for 
sufficiently large $\kappa$ and\break  $t \geq t^*(\kappa,M)+1$. To this end we combine (\ref{eq:15}),
and \eqref{eq:23} to obtain that for all $\epsilon\in
  (0,1)$
\begin{displaymath}
   \frac{1}{2}  \frac{d\|\psi\|_2^2}{dt}  + [\mu(h) (1-\epsilon)-1] \|\psi(t,\cdot)\|_2^2 \leq  \frac 1 \epsilon 
  \, \|A_1(t,\cdot)\|_\infty^2\, \|\psi\|_2^2 \,.
\end{displaymath}
Using the choice of $\epsilon$ 
\begin{equation}
\label{eq:27aa}
  \epsilon= \frac{1}{2}\frac{\mu(h)-1}{\mu(h)}\,,
\end{equation} 
  then yields
\begin{displaymath}
  \frac{1}{2}  \frac{d\|\psi (t,\cdot)\|_2^2}{dt}  + \Big[\frac{\mu-1}{2}-
  \frac{2\mu}{\mu-1}\|A_1(t,\cdot)\|_\infty^2 \Big] \|\psi\|_2^2 \leq0 \,.
\end{displaymath}
Hence, for sufficiently large $\kappa_0$, we get from \eqref{eq:30} and
\eqref{eq:15} that
\begin{displaymath}
  \frac{d\|\psi(t,\cdot)\|_2^2}{dt}<0 \,.
\end{displaymath}
\end{proof}

\section{Stable semigroup - coherence length scale}
\label{section4}
Let
\begin{equation}
  \label{eq:44}
\LL_h = -\nabla_{hA_n}^2 + ih\phi_n \,,
\end{equation}
where $(A_n,\phi_n)$ are  defined by \eqref{eq:8}, and $h$ denotes the
external current's intensity.  We define $D(\LL_h)$ as
\begin{displaymath}
  D(\LL_h) = \{ u\in H^2(\Omega) \, | \; u|_{\partial\Omega_c}=0 \; ; \;
  \nabla u\cdot\nu|_{\partial\Omega_i}=0\,\} \,.
\end{displaymath}
Let $Q_L$ denote the domain of the sesquilinear form $q_L$ associated
with $\LL_h\,$, that is $$Q_L:= \{ u\in H^1(\Omega)\,,\, u|_{\partial\Omega_c}=0
\}\,.$$ Then the domain of the operator $\widehat \LL_h$, which can be
obtained using the Lax-Milgram's procedure in \cite{aletal13}, is the
space:
$$
D(\widehat \LL_h): = \{u\in Q_L\,,\,
\nabla u\cdot\nu|_{\partial\Omega_i}=0\,,\,( -\nabla_{hA_n}^2 + ih\phi_n) u \in L^2\}\,.
$$
To show that $D( \widehat \LL_h) = D(\LL_h)$, we use \eqref{eq:173} to
prove $H^2$ regularity for every $u\in D(\widehat \LL_h)$.  Note that
since $A_n\cdot\nu=0$ on $\partial\Omega_i$, $u$ satisfies a Neumann 
boundary condition on $\partial\Omega_i$.

We seek a lower bound for the critical value $J^d$, so that the normal
state is globally stable whenever $\|J\|_{L^\infty(\partial\Omega_c)}>J^d$ . Let
$c$ be defined by \eqref{eq:9}, 
$$
J_r(x)=\kappa^2J_0(x)\,,
$$
and $h_r$ be fixed. We assume that $c$ is fixed and that $J_0$ is independent of
$\kappa$.  Hence (see (\ref{eq:45}e))
$\LL_h$ is independent of $\kappa$ as well.  We begin by the following
statement on the steady-state version of \eqref{eq:1}.
\begin{proposition}
\label{prop:stationary}
Let $(\psi,A,\phi)$ denote a steady-state solution of \eqref{eq:1}, i.e.
\begin{subequations}  
\label{eq:45}
\begin{alignat}{2}
- & \nabla_A^{2} \psi + i\phi\psi =  \psi\left( 1 - |\psi|^{2} \right)\,,  & \quad \text{ in } \Omega\,, \\
- & \curl^2A + \frac{1}{c}\nabla\phi\, =  \frac{1}{\kappa^2}\Im(\bar\psi\nabla_A\psi) & \quad \text{ in }  \Omega\,, \\
  &\psi=0\,, &\quad \text{ on }  \partial\Omega_c \,,\\
 &\nabla_A\psi\cdot\nu=0\,, & \quad \text{ on }  \partial\Omega_i\,, \\
 & \frac{\partial\phi}{\partial\nu} = -hcJ_0(x)\,, &\quad \text{ on } \partial\Omega_c\,, \\
&\frac{\partial\phi}{\partial\nu}=0\,,  &\quad \text{ on }  \partial\Omega_i \,, \\[1.2ex]
&\dashint_{\partial\Omega}\curl A (x) \, ds = h\cdot h_r \,.&
\end{alignat}
\end{subequations}
There exists $\kappa_0>0$ and $C_1>0$ such that if 
\begin{equation}
  \label{eq:46}
\|\LL^{-1}_h\| <  1- \frac{C_1}{\kappa^2} \,,
\end{equation}
for some $\kappa>\kappa_0\,$, $(0,hA_n,h\phi_n)$ is the unique solution of
\eqref{eq:45} satisfying \eqref{eq:4}.
\end{proposition}
\begin{proof}~\\
  Set $A_1$ and $\phi_1$ to be as in \eqref{eq:17}. Then,
  \begin{displaymath}
 \begin{aligned}
&- \curl^2A_1  =  \frac{1}{\kappa^2}\Im(\bar\psi\nabla_A\psi) -  \frac{1}{c}\nabla\phi_1& \quad \text{ in }  \Omega\,, \\
& \frac{\partial\phi_1}{\partial\nu}=0 &\quad \text{ on } \partial\Omega\,, \\
&\curl A_1=0 & \mbox{ on } \pa \Omega\,.
\end{aligned}
  \end{displaymath}
The last boundary condition is a consequence of   (\ref{eq:45}c-f) and \eqref{eq:8}.\\
Since $\phi_1$ satisfies \eqref{eq:19}, and since $\|\psi\|_\infty\leq 1$ by
\eqref{eq:2.2a}, it follows that
\begin{displaymath}
  \|\curl^2A_1\|_2 \leq \frac{2}{\kappa^2} \, \|\nabla_A\psi\|_2 \,.
\end{displaymath}
Observing that $\curl A_1$ vanishes on the boundary, we obtain that
\begin{equation}
\label{eq:48}
  \| \curl A_1\|_{H^1(\Omega)} \leq  \frac{C_\Omega }{\kappa^2}\,  \|\nabla_A\psi\|_2 \,,
\end{equation}
which leads, via Sobolev's injection, to
\begin{equation}
\label{eq:48a}
  \| \curl A_1\|_{L^p(\Omega)} \leq  \frac{C_\Omega^p }{\kappa^2} \, \|\nabla_A\psi\|_2 \,,
\end{equation}
for $p\in [2,+\infty)$.

Let $\Phi$ denote the solution in $H^1(\Omega)$ of
\begin{equation}
  \label{eq:49}
\begin{cases}
\Delta\Phi=\curl A_1 & \text{in } \Omega\,, \\
\Phi=0 & \text{on } \partial\Omega  \,.
\end{cases}
\end{equation}
We now use \eqref{eq:173D} and \eqref{eq:48a} 
to obtain that $\Phi \in W^{2,p}(\Omega)$ and 
\begin{equation}
  \label{eq:50}
\|\Phi\|_{2,p} \leq \frac{C(p,\Omega)}{\kappa^2} \|\nabla_A\psi\|_2 \,,
\end{equation}
for all $2\leq p<\infty$.

We next show  that $A_1=-\nabla_\perp\Phi\,$. 
We set first $u=A_1+\nabla_\perp\Phi$. Since $A_1\cdot\nu=0$ on $\partial\Omega$, and since
by the Dirichlet condition $\Phi$ satisfies we additionally have  
$\nabla_\perp\Phi\cdot \nu=0$ on $\partial\Omega$,  it follows that $u\in H^1(\Omega,\R^2)$ satisfies
\begin{displaymath}
  \begin{cases}
   \curl u = 0 & \text{in } \Omega \\
 \Div u = 0 & \text{in } \Omega \\
 u\cdot\nu = 0 & \text{on } \partial\Omega \,.   
  \end{cases}
\end{displaymath}
Since $u\equiv0$ is the unique solution to the above problem (see
Proposition \ref{decprop}), we obtain that $A_1=-\nabla_\perp\Phi\,$.

 From the embedding of $W^{1,p}(\Omega)$ in $L^\infty(\Omega)$ for any $p\in (2,+\infty)$,  we then get
\begin{equation}
  \label{eq:51}
 \|A_1\|_\infty\leq C(\Omega,p)\, \|A_1\|_{1,p} \leq  \widehat C(\Omega,p) \, \|\Phi\|_{2,p}\leq \frac{\widetilde C}{\kappa^2} \,  \|\nabla_A\psi\|_2 \,.
\end{equation}
Multiplying (\ref{eq:45}a) by $\bar{\psi}$ and integrating by parts we obtain
\begin{equation}
\label{eq:52}
    \|\nabla_A\psi\|_2^2 = \|\psi\|_2^2 - \|\psi\|_4^4 \;  \leq  \|\psi\|_2^2 \,. 
\end{equation}
Substituting the above into \eqref{eq:51} yields
\begin{equation}
\label{eq:53}
  \|A_1\|_\infty \leq \frac{C}{\kappa^2}  \|\psi\|_2 \,.
\end{equation}
We next write (\ref{eq:45}a) in the form
\begin{displaymath}
  \LL_h \psi = 2iA_1\cdot\nabla_A\psi - |A_1|^2 \psi -i\phi_1\psi + \psi\left( 1 - |\psi|^{2} \right) \,.
\end{displaymath}
By \eqref{eq:2.2a},   \eqref{eq:19}, and \eqref{eq:52}, we have
\begin{displaymath}
  \|\phi_1\psi\|_2 \leq \|\phi_1\|_2 \leq \frac{C}{\kappa^2}\, \|\nabla_A\psi\|_2 \leq \frac{C}{\kappa^2}\, \|\psi\|_2 \,.
  \end{displaymath}
Thus, we obtain that
\begin{displaymath}
  \|\LL_h \psi\|_2 \leq \Big(\frac{C_0}{\kappa^2}+1\Big) \|\psi\|_2 \,.
\end{displaymath}
The proposition now easily follows since, if we choose $C_1=C_0$ in
\eqref{eq:46}, we have that 
\begin{displaymath}
  \|\psi\|_2 \leq \|\LL_h^{-1}\| \|\LL_h \psi\|_2 \leq \Big(1-
  \frac{C_0^2}{\kappa^4}\Big)\|\psi\|_2 \,.
\end{displaymath}
From the above we readily conclude $\psi\equiv0\,$. One can then show that $A=hA_n$ and $\phi=h\phi_n$ from
\eqref{eq:8} and the discussion which follows.
\end{proof}

We now move to consider the time-dependent problem \eqref{eq:1}.
\begin{proof}[Proof of Theorem \ref{thm:1.2}]~\\
  We begin the proof by defining some useful entities. We first
  rewrite (\ref{eq:1}a) in the form
 \begin{equation}
\label{eq:56}
   \frac{\partial\psi}{\partial t} + \LL_h\psi = 2iA_1\cdot\nabla_A\psi - |A_1|^2\psi+ i\phi_1\psi 
   +(1-|\psi|^2)\psi  \,.
 \end{equation}
Set then
\begin{displaymath}
  F=2iA_1\cdot\nabla_A\psi - |A_1|^2\psi+ i\phi_1\psi 
   +(1-|\psi|^2)\psi  \,.
\end{displaymath}
We next define the Laplace transform. Let
$u\in e^{\omega\cdot}L^2(\R_+\,;\,L^2(\Omega,\C))$. Then
\begin{displaymath}
  \widehat{u}(s,x) = \int_0^\infty e^{st}u(t, x)\,dt
\end{displaymath}
denotes the Laplace transform of $u$, which is well defined whenever
$\Re s\leq\omega$. Denote then
\begin{equation}
\label{eq:57}
  \Gamma_\omega = \{ s\in\C\,|\, \Re s=\omega\}\,.
\end{equation}
It is well known that $\widehat{u}\in L^2(\Gamma_\omega\,;\,L^2(\Omega,\C))$. Finally we
define the cutoff function 
\begin{equation}
\label{eq:58}
  \chi_{T,\epsilon}=
  \begin{cases}
    0 & t<\frac{1}{\epsilon} \\
    1 & \frac{1}{\epsilon}<t<T \\
    e^{-\epsilon(t-T)} & t\geq T \,,
  \end{cases}
\end{equation}
for some $T>1/\epsilon$.

{\em Step 1:} Let $t^*(\kappa,M) $ be given by \eqref{eq:145}. We
prove that there exists $C>0$ such that for sufficiently large $\kappa$
and $\epsilon$ satisfying
\begin{equation}\label{condepsilon}
 \frac 1 \epsilon > t^*(\kappa,M)\,, 
 \end{equation}
 we have
\begin{equation}
  \label{eq:59}
\|\widehat{\chi_{T,\epsilon}F}\|_{L^2(\Gamma_0\,;\,L^2(\Omega,\C))}^2 \leq \Big(1+\frac{C}{\kappa^2}
\Big)\big[\|\widehat{\chi_{T,\epsilon}\psi}\|_{L^2(\Gamma_0\,;\,L^2(\Omega,\C))}^2 +
\|\psi(\epsilon^{-1},\cdot )\|_2^2 \big]\,.
\end{equation}
By Parseval's identity we have
\begin{displaymath}
  \frac{1}{2\pi}\|\widehat{\chi_{T,\epsilon}F}\|_{L^2(\Gamma_0\,;\,L^2(\Omega,\C))}^2 =
  \|\chi_{T,\epsilon}F\|_{L^2(\R_+\,;\,L^2(\Omega,\C))}^2 \,.
\end{displaymath}
We, therefore, attempt to estimate the norm of the  right hand side. 
 Let then 
 \begin{displaymath}
   F = F_1+F_2\,,
 \end{displaymath}
 where
 \begin{displaymath}
   F_1 = (1-|\psi|^2-|A_1|^2)\psi 
 \end{displaymath}
 and
 \begin{displaymath}
   F_2 = i\phi_1\psi +2iA_1\cdot\nabla_A\psi \,.
 \end{displaymath}
 To bound the norm of $\chi_{T,\epsilon}F_1$\,,  we recall \eqref{eq:159}  which holds for every $t\geq t^*(\kappa,M)$.
 Hence, for sufficiently large $\kappa$, we have that $
 \|A_1(t,\cdot)\|_\infty\leq1$ for all $t>t^*$, from which we conclude that
 \begin{displaymath}
  |F_1(t,x)|\leq  |\psi (t,x)| \,.
 \end{displaymath}
We thus conclude  that there exists $C(\Omega,c)$ such that, for
   sufficiently large $\kappa$ large and for any $\epsilon$  satisfying \eqref{condepsilon},  we have 
\begin{equation}
\label{eq:61}
  \int_0^\infty \chi_{T,\epsilon}^2(t)  \|F_1(t,\cdot)\|_2^2\,dt \leq
  \Big(1+\frac{C(\Omega,c)}{\kappa^2}\Big) \int_0^\infty \chi_{T,\epsilon}^2(t)  \|\psi(t,\cdot)\|_2^2\,dt
\end{equation}

To bound the norm of $\chi_{T,\epsilon}F_2$\,,  we use \eqref{eq:159}, \eqref{eq:19}, and
Poincar\'e inequality, applied to $\phi_1$,  to obtain that for
sufficiently large $t$
\begin{displaymath}
   \||F_2(t,\cdot)\|_2 \leq \|\phi_1\|_2 +  \frac{C}{\kappa^2} \, \|\nabla_A\psi(t,\cdot)\|_2
   \leq \frac{\widehat C}{\kappa^2}\,  \|\nabla_A\psi(t,\cdot)\|_2 \,.
\end{displaymath}
Hence, for  $\epsilon$ satisfying \eqref{condepsilon}, we have
 \begin{displaymath}
   \int_0^\infty \chi_{T,\epsilon}^2(t)  \|F_2(t,\cdot)\|^2_2\,dt \leq  \frac{\widehat C^2 }{\kappa^4}  \int_0^\infty\chi_{T,\epsilon}^2 (t) 
 \|\nabla_A\psi(t,\cdot)\|^2_2\,dt \,.
 \end{displaymath}
 We next use \eqref{eq:22} to obtain that
 \begin{displaymath}
   \|\chi_{T,\epsilon}F_2\|^2_{L^2(\R_+;\,L^2(\Omega,\C))} \leq  \frac{C(\Omega,c)}{\kappa^4}  \int_{\frac 1 \epsilon}^\infty
  \chi_{T,\epsilon}^2(t)  \Big[\|\psi(t,\cdot)\|_2^2-\frac{1}{2}  \frac{d\|\psi(t,\cdot)\|_2^2}{dt}\Big]  \,dt \,.
 \end{displaymath}
 Integration by parts then yields, using the fact that $\chi_{T,\epsilon}^\prime (t) 
 \leq 0$ for all $t>\epsilon^{-1}$,
 \begin{displaymath}
     \|\chi_{T,\epsilon}F_2\|^2_{L^2(\R_+;\,L^2(\Omega,\C))} \leq
     \frac{C(\Omega,c)}{\kappa^4}  \Big[\int_0^\infty \chi_{T,\epsilon}^2(t)   \|\psi(t,\cdot)\|_2^2
     \,dt + \frac{1}{2}\|\psi(\epsilon^{-1}, \cdot )\|_2^2 \Big]\,.
 \end{displaymath}
 In conjunction with  \eqref{eq:61}, the above inequality readily
 yields \eqref{eq:59}.

{\em Step 2:} We now prove \eqref{eq:181}. Multiplying \eqref{eq:56} by
$\chi_{T,\epsilon}$ and then applying the Laplace transform with $s=\nu+i\gamma$ yields
\begin{displaymath}
  (\LL_h-\nu-i\gamma) \widehat{\chi_{T,\epsilon}\psi} -\widehat{\chi_{T,\epsilon}^\prime\psi} -
  e^{(\nu+i\gamma)/\epsilon}\psi (\epsilon^{-1},\cdot) 
    = \widehat{\chi_{T,\epsilon}F}\,.
\end{displaymath}
$\chi_{T,\epsilon}^\prime$ denotes here the extension to $[0,\infty)$ of the
derivative of $\chi_{T,\epsilon}$ in $ (\frac {1}{ \epsilon},+\infty ) $, which is
obtained by setting $\chi_{T,\epsilon}^\prime=0$ on $[0,\frac {1}{ \epsilon}]$. \\

 By
\eqref{eq:180} it follows that, for any $\gamma \in \mathbb R$,
 \begin{multline*}
   \| \widehat{\chi_{T,\epsilon}\psi} (\nu+i\gamma,\cdot)\|_2 \leq  \|(\LL_h-\nu-i\gamma)^{-1}\| \big[
     \|\widehat{\chi_{T,\epsilon}F}(\nu+i\gamma,\cdot)\|_2 + \|\widehat{\chi_{T,\epsilon}^\prime\psi}(\nu+i\gamma,\cdot)\|_2
     + e^{\nu/\epsilon}\|\psi(\epsilon^{-1},\cdot)\|_2\} \big] \leq\\ \quad \leq  \Big(1- \frac{C_1}{\kappa^2}\Big)\big[
     \|\widehat{\chi_{T,\epsilon}F}(\nu+i\gamma,\cdot)\|_2 + \|\widehat{\chi_{T,\epsilon}^\prime\psi}(\nu+i\gamma,\cdot)\|_2
     \big] +  e^{\nu/\epsilon}\|(\LL_h-\nu-i\gamma)^{-1}\|  \|\psi(\epsilon^{-1},\cdot )\|_2 \,,
 \end{multline*}
where the precise value of $C_1$ will be determined later. 
Hence, by Cauchy's inequality, we have for every $\delta>0$,
\begin{multline}
\label{eq:62}
  \| \widehat{\chi_{T,\epsilon}\psi} (\nu+i\gamma,\cdot)\|_2^2 \leq \Big(1- \frac{C_1}{\kappa^2}\Big)\Big[
    (1+2\delta)\|\widehat{\chi_{T,\epsilon}F}(\nu+i\gamma,\cdot)\|_2^2 + \frac{1}{\delta}\|\widehat{\chi_{T,\epsilon}^\prime\psi}(\nu+i\gamma,\cdot)\|_2^2
    \Big]\\  + \frac{e^{\nu/\epsilon}}{\delta}\|(\LL_h-\nu-i\gamma)^{-1}\|^2 \,\|\psi(\epsilon^{-1},\cdot)\|_2^2\,.
\end{multline}
Since $\phi_n\in H^2(\Omega)\subset L^\infty(\Omega)$,  it is easy to show from the identity 
$$
 \Im\langle u,\LL_hu\rangle = h \langle u,\phi_nu\rangle\,,
 $$
 that, for some positive $C$,
 \begin{displaymath}
  \|(\LL_h-\nu-i\gamma)^{-1}\|\leq \frac{C}{1+|\gamma|} \,.
\end{displaymath}
Consequently, we may integrate \eqref{eq:62} over $\Gamma_\nu$ to obtain
\begin{displaymath}
\begin{array}{l}
  \| \widehat{\chi_{T,\epsilon}\psi}\|_{L^2(\Gamma_\nu\,;\,L^2(\Omega,\C))}^2 \\
  \qquad \leq \Big(1- \frac{C_1}{\kappa^2}\Big)\Big[
    (1+2\delta)\|\widehat{\chi_{T,\epsilon}F}\|_{L^2(\Gamma_\nu\,;\,L^2(\Omega,\C))}^2\\
   \qquad\qquad +  \frac{1}{\delta}\|\widehat{\chi_{T,\epsilon}^\prime\psi}\|_{L^2(\Gamma_\nu\,;\,L^2(\Omega,\C))}^2
    \Big] + \frac{Ce^{\nu/\epsilon}}{\delta}\|\psi(\epsilon^{-1},\cdot )\|_2^2 \,.
    \end{array}
\end{displaymath}
 We next use Parseval's identity together with \eqref{eq:59} and the
 fact that 
 $$- \chi_{T,\epsilon}^\prime \leq \epsilon\chi_{T,\epsilon}\,,$$ to obtain that, for any $\delta >0$,
  \begin{multline*}
     \|e^{\nu t} \chi_{T,\epsilon}\psi\|_{L^2(\R_+\,;\,L^2(\Omega,\C))}^2\\ \leq  \Big(1-
     \frac{C_1}{\kappa^2}\Big) (1+2\delta)\Big(1+\frac{C}{\kappa^2}+
     C\frac{\epsilon^2}{\delta}
  \Big)\| e^{\nu t}\chi_{T,\epsilon}\psi\|_{L^2(\R_+\,;\,L^2(\Omega,\C))}^2 \\
    + \frac{Ce^{\nu/\epsilon}}{\delta}\|\psi(\epsilon^{-1},\cdot)\|_2^2 \,.
  \end{multline*}
Finally, we choose $C_1=8(C+1)$, $\delta=C/\kappa^2$, and $\epsilon=\min(1/\kappa^2,1/t^*)$, to
obtain that
\begin{displaymath}
  \| e^{\nu t}\chi_{T,\epsilon}\psi\|_{L^2(\R_+\,;\,L^2(\Omega,\C))}^2 \leq Ce^{\nu/\epsilon}\kappa^4\|\psi(\epsilon^{-1},\cdot)\|_2^2 \,.
\end{displaymath}
The theorem now easily follows by taking the limit $T\to\infty$.
\end{proof}

\begin{remark}
  It can be easily verified that \eqref{eq:46}, implies that any
  $\mu\in B(0,(1- C_1/\kappa^2)^{-1})$ is in the resolvent set of $\LL_h$. As
  $\langle u,\LL_hu\rangle\geq0$ for all $u\in D(\LL_h)$, we can thus conclude that
  $(-\infty, (1- C_1/\kappa^2)^{-1}\,)\cap\sigma(\LL_h)=\emptyset $. A similar conclusion
  can be  reached from the lower bound \eqref{eq:14} on $\mu(h)$  in
  Section \ref{section3}. Indeed, it can be easily verified that $\|\mathcal L_h^{-1}\| \leq
  1/\mu(h)$. 
  The formulation of \eqref{eq:180} is reminiscent of conditions
  appearing in the statement of the Gearhart-Pruss Theorem (see Theorem~1.11 (p. 302-304) in \cite{EnNa}) or more quantitatively
  in Helffer-Sj\"ostrand \cite{HeSj} (see Remarks 1.3 and 1.4 there).  
   \end{remark}

\section{Penetration depth scaling}\label{section5}
We begin by considering a  steady-state solution for domains scaled with
respect to the penetration depth. If one applies the transformation
\begin{equation}
\label{eq:141}
  x \to \frac{x}{\kappa} \quad ; \quad \phi\to\kappa\phi \,,
\end{equation}
the system \eqref{eq:45} becomes 
\begin{subequations}  
\label{eq:63}
\begin{alignat}{2}
- & \nabla_{\kappa A}^2\psi+i\kappa\phi \psi=\kappa^2(1-|\psi|^2)\psi & \quad \text{ in } \Omega \,,\\
- &\kappa^2\curl^2A+\sigma\nabla\phi =\kappa\,\Im(\bar\psi\nabla_{\kappa A}\psi) & \quad \text{ in }  \Omega\,, \\
  &\psi=0 &\quad \text{ on }  \partial\Omega_c\,, \\
 &\nabla_{\kappa A}\psi\cdot\nu=0 & \quad \text{ on }  \partial\Omega_i\,, \\
 & \frac{\partial\phi}{\partial\nu} = -h\frac{\kappa^4J_0(x)}{\sigma} &\quad \text{ on } \partial\Omega_c\,, \\
&\frac{\partial\phi}{\partial\nu}=0  &\quad \text{ on }  \partial\Omega_i \,, \\[1.2ex]
&\dashint_{\partial\Omega}\curl A \, ds = h\cdot h_r \,.&
\end{alignat}
\end{subequations}
We consider again fixed $c=\kappa^2/\sigma$ and $J_0$. The
normal state for this scaling is defined once again by
\eqref{eq:8}. We define the linear operator (with Dirichlet-Neumann conditions)
\begin{equation}
  \label{eq:64}
\Lg_\kappa = \LL_{\kappa^3 h}  :=-  \nabla_{\kappa^3hA_n}^2 + i \kappa^3h \phi_n   \,.
\end{equation}
Note that unlike $\LL_h$, $\Lg_\kappa$ does depend on $\kappa$.

\begin{lemma}
For any $\beta<1$ there exists $C_\beta(\Omega,h)>0$ and $\kappa_\beta (\Omega,h)$ such that for
 $\kappa \geq \kappa_\beta$ we have
\begin{equation}
\label{eq:77}
   \|\psi\|_2 \leq C_\beta\kappa^{-\frac \beta 2} \,.
\end{equation}
\end{lemma}
\begin{proof}~\\
  From Assumption (B) (implying \eqref{hypnabla}) and Assumption (R1) on
  $\Omega$ we get that $\pa \Omega \cup \{\curl A_n =0\}$ is a union of
  piecewise $C^1$ curves. For $\alpha \in (0,1)$, whose value is to be
  determined later, we define
\begin{equation}\label{eq:70}
S_{\kappa,\alpha} = \{x\in \Omega\,,\, \{|\curl A_n(x)| \geq \kappa^{-\alpha} \}\cap \{ d(x,\pa \Omega)\geq \kappa^{-\alpha}\}.
\end{equation}
Note that $|\Omega \setminus S_{\kappa,\alpha} |\leq C \kappa^{-\alpha}$ for some fixed $C>0$.  We can
now consider (for $\kappa$ large enough) a covering of $S_{\kappa,\alpha}$ by
balls of size $\kappa^{-\alpha'}$ with $1> \alpha' > \alpha$ with support in $\Omega$.

In each of the balls (for $\kappa$ large enough) we can assume that the
sign of $\curl A_n$ is constant. Associated with this (finite)
covering, we can associate a partition of unity $\eta_j$ such that
 \begin{equation}\label{eq:70a}
 \sum_j \eta_j^2 =1 \mbox{ on } S_{\kappa,\alpha}\,,\,
 \sum_j  |\nabla \eta_j|^2 \leq C \kappa^{2\alpha'} \mbox{ and } \supp \eta_j \subset \Omega\,.
 \end{equation}
 
Multiplying (\ref{eq:63}a) by
$\eta_{j}^2\bar{\psi}$ and integrating by parts, taking into account
the boundary conditions (\ref{eq:63}c,d) yields for the
real part
\begin{displaymath}
  \|\nabla_{\kappa A}(\eta_{j}\psi)\|_2^2= \|\psi\nabla\eta_{j}\|_2^2+ \kappa^2(\|\eta_{j}\psi\|_2^2-
  \|\eta_{j}^{1/2}\psi\|_4^4) \,.
\end{displaymath}
With the aid of Theorem
4 in \cite{mo95} we then obtain
\begin{equation}
\label{eq:71}
  \kappa\,  \Big| \int  \curl A \,
  |\eta_{j}\psi|^2\, dx \Big|  \leq C \|(\nabla \eta_j) \psi\|_2^2+
  \kappa^2\|\eta_{j}\psi\|_2^2   \,.
\end{equation}
Using the fact that $\curl A_n$ does not change
sign in each of the connected components of  $\Sg_{\kappa,\alpha}$ we obtain
that for $\kappa$ large enough
\begin{equation}
\label{eq:73}
  \kappa^{3-\alpha}h\int_{\Sg_{\kappa,\alpha}} |\psi|^2 \, dx\, \leq \kappa\int_\Omega 
  |\curl A_1|\,|\psi|^2\,dx\;  + \; C\, \kappa^2\, \int_\Omega |\psi|^2\, dx \,.
\end{equation}
From \eqref{eq:68} below we get that 
\begin{equation}\label{eq:73a}
  \|\curl A_1\|_2 \leq C(\Omega)\,  \|\psi\|_2 \,,
\end{equation}
which, when substituted into \eqref{eq:73} yields, with the aid of
\eqref{eq:2.2a} and Cauchy-Schwarz inequality, the existence of
$C(h,\Omega) >0$ such that, for sufficiently large  $\kappa$,
\begin{equation}
\label{eq:74}
  \|\psi\|_{L^2(\Sg_{\kappa,\alpha})}^2 \leq
  \frac{C}{\kappa^{1-\alpha}} \|\psi\|_2^2 \,.
\end{equation}
Choosing  $\alpha < 1$, keeping in mind the control of the measure of $\Omega
\setminus S_{\kappa,\alpha}$, yields
$$
 \int_{S_{\kappa,\alpha}} |\psi|^2 \, dx\, \leq   \;  C_h \kappa^{-1}\,,
 $$
 and completing   by the integral over the complementary of $S_{\kappa,\alpha}$, we obtain
 $$
 \int_{\Omega} |\psi|^2 \, dx\, \leq   \;  C_h  \kappa^{-\alpha}\,.
 $$
 To complete the proof, we can  take $\alpha =\beta$ and then $\alpha' = (1+\beta)/2$.
Consequently,  we obtain \eqref{eq:77}. This proves the lemma.
\end{proof}
We can now state our steady-state estimate for the critical current
\begin{proposition}
  Suppose that  there exist  $\gamma <1/2$,  $C>0$ and $\kappa_0>0$ such that
  \begin{equation}\label{eq:5.4}
  \|\Lg^{-1}_\kappa\| < \frac{1}{\kappa^2}\Big[1- \frac{C}{\kappa^{\gamma}}\Big]\,,
  \end{equation}
    for all $\kappa>\kappa_0$.
  Then, there exists $\kappa_1\geq\kappa_0$ such that $\psi\equiv0$ is the unique solution
  of \eqref{eq:63} for all $\kappa\geq \kappa_1$.
\end{proposition}
\begin{proof}
  As in the proof of Proposition \ref{prop:stationary} we set
  \begin{equation}
\label{eq:65}
    A_1= A - h\kappa^2\, A_n \quad ; \quad \phi_1 = \phi - h\kappa^2\, \phi_n \,.
  \end{equation}
It is easy to show that $A_1$ and $\phi_1$ satisfy 
  \begin{equation}
\label{eq:66}
 \begin{aligned}
&  -\kappa^2\, \curl^2A_1  =  \kappa\, \Im(\bar\psi\nabla_{\kappa A}\psi) -  \sigma\, \nabla\phi_1& \quad \text{ in }  \Omega\,, \\
& \frac{\partial\phi_1}{\partial\nu}=0 &\quad \text{ on } \partial\Omega\,, \\
& \curl A_1=0 &\quad \text{ on } \partial\Omega \,.
\end{aligned}
  \end{equation}
Taking the scalar product with $\nabla \phi_1$ of the first line of \eqref{eq:66}  yields
\begin{equation}
\label{eq:67}
   \|\nabla\phi_1\|_2 \leq \frac{\kappa}{\sigma}\, \|\Im(\bar\psi\nabla_{\kappa A}\psi)\|_2 \,,
\end{equation}
and hence, 
\begin{displaymath}
  \|\curl^2A_1\|_2 \leq \frac{2}{\kappa} \, \|\nabla_{\kappa A} \psi\|_2 \,.
\end{displaymath}
By (\ref{eq:63}a) we then have (see the proof of \eqref{eq:52})
\begin{equation}
  \label{eq:68}
\|\curl^2A_1\|_2 \leq 2\,  \|\psi\|_2 \,.
\end{equation}
We can now set $A_1=-\nabla_\perp\Phi$, and follow the same route as in the
proof of \eqref{eq:51} (see \eqref{eq:53}), to obtain, with the aid of Sobolev embedding
that
\begin{equation}
  \label{eq:69}
\|A_1\|_\infty \leq C\,  \|\psi\|_2 \,.
\end{equation}
~\\
By \eqref{eq:67}, \eqref{eq:51} and \eqref{eq:77} we obtain that for
all $\hat \gamma >1/2$ there exists $C_{\hat \gamma} (h,\Omega)$ such that for sufficiently large
$\kappa$ we have
\begin{displaymath}
   \|2i\kappa A_1\cdot\nabla_{\kappa A}\psi - \kappa^2|A_1|^2 \psi + i\kappa\phi_1\psi\|_2 \leq
   C_{\hat \gamma} \, \kappa^{1+\hat \gamma}\, \|\psi\|_2 \,.
\end{displaymath}
Consequently, we obtain that
\begin{displaymath}
  \|\psi\|_2 \leq \big(\kappa^2+ C_{\hat \gamma} \kappa^{1+\hat \gamma}\big) \|\Lg_\kappa^{-1}\|\,
  \|\psi\|_2 \,,
\end{displaymath}
from which the proposition readily follows by choosing $1/2<\hat \gamma < 1-\gamma\,$.
\end{proof}

We next prove the decay, in the long time limit, of solutions of the
time-dependent version of \eqref{eq:63}, i.e.,
\begin{subequations}  
\label{eq:78}
\begin{alignat}{2}
&\frac{\partial\psi}{\partial t}- \nabla_{\kappa A}^2\psi+i\kappa\Phi\psi=\kappa^2(1-|\psi|^2)\psi & \quad \text{ in } \Omega\,,\\
- &\kappa^2\curl^2A+\sigma\Big(\nabla\Phi +\frac{\partial A}{\partial t}\Big)=\kappa\,\Im(\bar\psi\nabla_{\kappa A}\psi) & \quad \text{ in }  \Omega\,, \\
  &\psi=0 &\quad \text{ on }  \partial\Omega_c \,,\\
 &\nabla_{\kappa A}\psi\cdot \nu=0 & \quad \text{ on }  \partial\Omega_i \,,\\
 & \nabla\Phi \cdot\nu = -h\frac{\kappa^4J_0(x)}{\sigma} &\quad \text{ on } \partial\Omega_c\,, \\
&\nabla\Phi \cdot\nu = 0 &\quad \text{ on }  \partial\Omega_i \,, \\[1.2ex]
&\dashint_{\partial\Omega}\curl A \, ds = h\kappa\cdot h_r \,.& 
\end{alignat}
\end{subequations}
We begin by  proving a few auxiliary estimates.
\begin{lemma}
   Let $A_1$ be defined by \eqref{eq:65}. 
There exists $C(c,\Omega)$ and $\kappa_0$ such that, for $\kappa\geq \kappa_0$,
   \begin{equation}
     \label{eq:79}
\|A_1(t,\cdot)\|_2^2 \leq \Big(\|A_1(0,\cdot)\|_2^2+\frac{C(c,\Omega)}{\kappa^2}\Big)
e^{-\lambda ct} + C(c,\Omega)  \int_0^te^{-\lambda c(t-\tau)}\|\psi(\tau,\cdot)\|_2^2 \,d\tau \,,
   \end{equation}
and  a constant $C_1(c,\Omega)$ such that:
\begin{equation}
  \label{eq:80}
\int_0^t e^{-2\lambda c(t-\tau)} \|\curl A_1(\tau,\cdot)\|_2^2\,d\tau \leq C_1(c,\Omega)\, \Big\{ \Big(\|A_1(0,\cdot)\|_2^2+\frac{1}{\kappa^2}\Big)
e^{-\lambda ct} + \int_0^te^{-\lambda c(t-\tau)}\|\psi(\tau,\cdot)\|_2^2 \,d\tau\Big\}  \,.
\end{equation}
\end{lemma}
\begin{proof}
   In a similar manner to the one used to derive \eqref{eq:20} we
   obtain that
   \begin{displaymath}
       \frac{1}{2}\frac{d\|A_1(t,\cdot)\|_2^2}{dt}  +
       \frac{\lambda\kappa^2}{\sigma}\|A_1(t,\cdot)\|_2^2 \leq  \frac{C}{\kappa}\, \|A_1(t,\cdot)\|_2\, \|\nabla_{\kappa A}\psi(t,\cdot)\|_2 \,,
   \end{displaymath}
from which we readily obtain that
\begin{displaymath}
  \frac{d\|A_1(t,\cdot)\|_2}{dt}  +
       \lambda c \, \|A_1(t,\cdot)\|_2 \leq  \frac{C}{\kappa}\, \|\nabla_{\kappa A}\, \psi(t,\cdot) \|_2 \,.
\end{displaymath}
Hence,
\begin{displaymath}
\begin{array}{ll}
  \|A_1(t,\cdot)\|_2 &\leq \|A_1(0,\cdot)\|_2\, e^{-\lambda ct} +
  \frac{C}{\kappa}\int_0^te^{-\lambda c(t-\tau)}\|\nabla_{\kappa A}\psi (\tau,\cdot)\|_2 \,d\tau \\ &
  \leq
\|A_1(0,\cdot)\|_2\, e^{-\lambda ct} +  \frac{C}{\kappa}\, \Big[\int_0^te^{-\lambda c(t-\tau)}\|\nabla_{\kappa A}\psi(\tau,\cdot)\|_2^2 \,d\tau\Big]^{1/2}\,.
\end{array}
\end{displaymath}
Note that $C$ depends on $\Omega$ and $c$. 
Since, as in \eqref{eq:22}
\begin{equation}
\label{eq:82}
      \|\nabla_{\kappa A}\psi\|_2^2 \leq  \kappa^2\|\psi\|_2^2-\frac{1}{2}
      \frac{d\|\psi\|_2^2}{dt} \,,
\end{equation}
we obtain after an integration by parts
\begin{displaymath}
   \|A_1(t,\cdot)\|_2 \leq \|A_1(0,\cdot)\|_2\, e^{-\lambda ct}+
   C\, \Big[\int_0^te^{-\lambda c(t-\tau)}\|\psi\|_2^2 \,d\tau
 +\frac{1}{2\kappa^2}\|\psi(0,\cdot)\|_2^2e^{-\lambda ct}\Big]^{1/2}\,,
\end{displaymath}
from which \eqref{eq:79} readily follows.

To prove \eqref{eq:80} we use the fact that, by straightforward
integration by parts we have,
\begin{displaymath}
    \frac{1}{2}\frac{d\|A_1(\tau ,\cdot)\|_2^2}{d\tau}  +
       c\,  \|\curl A_1(\tau,\cdot)\|_2^2 \leq
       \frac{C}{\kappa}\, \|A_1(\tau,\cdot)\|_2\, \|\nabla_{\kappa A}\psi(\tau,\cdot)\|_2 \,.
\end{displaymath}
Multiplying by $e^{-2\lambda c(t-\tau)}$ and integrating with respect to $\tau$
yields after integration by parts
\begin{multline*}
    \int_0^te^{-2\lambda c(t-\tau)} \|\curl A_1(\tau,\cdot)\|_2^2\,d\tau  \leq C
    \int_0^te^{-2\lambda c(t-\tau)}
    \Big[\|A_1(\tau,\cdot)\|_2^2+\frac{1}{\kappa^2}\|\nabla_{\kappa
      A}\psi(\tau,\cdot)\|_2^2\Big] \,d\tau\\ + \frac{1}{2}e^{-2\lambda
      ct}\|A_1(0,\cdot)\|_2^2 + \lambda c \int_0^te^{-2\lambda c(t-\tau)}
    \|A_1(\tau,\cdot)\|_2^2\,d\tau \,.
\end{multline*}
The above, together with \eqref{eq:79} and \eqref{eq:82},  leads to
\eqref{eq:80}.
\end{proof}

 \begin{lemma}\label{lem:5.l2bound}
  Let $A_1$ be defined by \eqref{eq:65} and $0<\beta<1/2$. Suppose that
  $\|A_1(\cdot,0)\|\leq M$. Then, there
  exist  $\kappa_\beta (h,c,\Omega)$ and  $C_\beta(h,c,\Omega)$, and,  for $\kappa \geq \kappa_\beta$,  $t_\beta(M,\kappa)$ such that, for all $t\geq t_\beta (\kappa) $, 
  \begin{equation}
\label{eq:81}
    \|A_1(t,\cdot)\|_2 + \|\psi(t,\cdot)\|_2 \leq \frac{C_\beta}{\kappa^\beta} \,.
  \end{equation}
\end{lemma}
\begin{proof}
Let $S_{\kappa,\alpha}$ be defined by \eqref{eq:70} and $\eta_{j}$ by
  \eqref{eq:70a}. We multiply (\ref{eq:78}a) by $\eta_{j}^2\bar{\psi}$
  and integrate over $\Omega$ to obtain
\begin{displaymath}
  \frac{1}{2}\frac{d\|\eta_{j}\psi\|_2^2}{dt} + \|\nabla_{\kappa A}(\eta_{j}\psi)\|_2^2=
  \|\psi\nabla\eta_{j}\|_2^2+ \kappa^2\, (\|\eta_{j}\psi\|_2^2- \|\eta_{j}^{1/2}\psi\|_4^2) \,.
\end{displaymath}
In the same manner used to prove \eqref{eq:73} we can show that
\begin{displaymath}
   \|\nabla_{\kappa A}(\eta_{j}\psi)\|_2^2 \geq \kappa^{3-\alpha}h\, \|\eta_{j}\psi\|^2 \,,
\end{displaymath}
and hence
\begin{displaymath}
   \frac{1}{2}\frac{d\|\eta_{j}\psi\|_2^2}{dt} + \kappa^{3-\alpha}h\, \|\eta_{j}\psi\|^2 \leq
 \kappa\, \int_\Omega  |\curl A_1|\,|\eta_j \psi|^2  +  \int | \psi\,\nabla\eta_{j}|^2\,dx +  C\kappa^2\, \int_\Omega |\eta_j \psi|^2\, dx \,.
\end{displaymath}
Summing over $j$, we obtain with the aid of the Cauchy-Schwarz
inequality and \eqref{eq:2.2a}:
\begin{equation}
\label{eq:153}
  \frac{1}{2}\frac{d\big(\sum_j \|\eta_{j}\psi\|_2^2\big)}{dt} + \kappa^{3-\alpha}h\, (\sum_j  \|\eta_{j}\psi\|^2)
  \leq \kappa\,  \|\curl A_1\|_2\Big(\sum_j\|\eta_j\psi\|_2\Big)^{\frac 12} + C\kappa^2\, \|\psi\|_2^2  \,.
\end{equation}
We rewrite this inequality in the following way, for $\kappa$ large enough,
\begin{equation}
\label{eq:153a}
  \frac{1}{2}\frac{d (\sum_j \|\eta_{j}\psi\|_2^2)}{dt} + \frac 12 \kappa^{3-\alpha}h\, (\sum_j  \|\eta_{j}\psi\|^2)
  \leq  \kappa\,  \|\curl A_1\|_2(\sum_j\|\eta_j\psi\|_2^2)^{\frac 12} +  C \kappa^{2-\alpha}  \,.
\end{equation}
We next define 
$$
\Theta (t):=\Theta_{\kappa,\alpha} (t) = \sum_j \|\eta_{j}(\cdot) \psi(t,\cdot)\|_2^2\,.
$$
From \eqref{eq:153a} we deduce that
\begin{displaymath}
  \Theta(t)  \leq
  \Theta (0) e^{-h\kappa^{3-\alpha}t}+  2\,  \int_0^te^{-h\kappa^{3-\alpha}(t-\tau)}\Big[\kappa\|\curl
  A_1(\tau,\cdot)\|_2\, \Theta (\tau)^{\frac 12}+C\kappa^{2-\alpha}\Big]\,d\tau \,,
\end{displaymath}
from which  we obtain 
\begin{multline}
\label{eq:150}
  \Theta(t)  \leq
  \Theta (0) e^{-h\kappa^{3-\alpha}t}+ \\ +2 \kappa \Big[\int_0^te^{-2\lambda c(t-\tau)}\|\curl
  A_1(\tau,\cdot)\|_2^2\,d\tau \Big]^{1/2}\Big[\int_0^te^{-2(h\kappa^{3-\alpha}-\lambda c)(t-\tau)}\Theta(\tau)\, d \tau \Big]^{1/2}  
+ C \kappa^{-1}   \,.
\end{multline}
Let $t^*=t^*(\kappa,M)$ be given by \eqref{eq:145}. Suppose now that $t\geq t^*$.  It
readily follows from \eqref{eq:80} and   \eqref{eq:145} that
\begin{displaymath}
   \int_0^te^{-2\lambda c(t-\tau)}\|\curl
  A_1(\tau,\cdot)\|_2^2\,d\tau \leq C\Big(\frac{1}{\kappa^2} + \|\psi\|_{L^\infty (0,t;L^2(\Omega,\C))}^2\Big) \,.
\end{displaymath}
Substituting the above into \eqref{eq:150} yields, with the aid of
\eqref{eq:2.2a}, for sufficiently large $\kappa$,
\begin{multline*}
   \Theta(t) \leq \Theta(0) e^{- t^* h \kappa^{2}} +  C\Big(\frac{1}{\kappa}+\frac{1}{\kappa^{(1-\alpha)/2}}\|\psi\|_{L^\infty
       (0,t;L^2(\Omega,\C))}^2 + \frac{1}{\kappa^{1-\alpha}}\|\psi\|_{L^\infty
       (0,t;L^2(\Omega,\C))}^2\Big) \\ \leq \hat{C}\Big(\frac{1}{\kappa^{(1-\alpha)/2)}}\|\psi\|_{L^\infty
       (0,t;L^2(\Omega,\C))}^2+\frac{1}{\kappa}\Big)\,.
\end{multline*}

Hence, for all $t\geq t^*$ we have
\begin{equation*}
  \|\psi(t,\cdot)\|_2^2 \leq  \Theta(t) + |\Omega\setminus\Sg_{\kappa,\alpha}| \leq  \Theta(0) \exp - t^*\, h \kappa^{2} +C\Big(\frac{1}{\kappa^{(1-\alpha)/2)}}\|\psi\|_{L^\infty
       (0,t;L^2(\Omega,\C))}^2+\frac{1}{\kappa^\alpha}\Big)\,.
\end{equation*}
leading to:
\begin{equation}
\label{eq:160}
  \|\psi(t,\cdot)\|_2^2 \leq    e^{- t^*\, h \kappa^{2} } +C\Big(\frac{1}{\kappa^{(1-\alpha)/2)}}\|\psi\|_{L^\infty
       (0,t;L^2(\Omega,\C))}^2+\frac{1}{\kappa^\alpha}\Big)\,.
\end{equation}
Optimizing over $\alpha$ gives, for $\alpha =\frac 13$, the existence of
$\kappa_0$ and $t^*$ such that, for $\kappa \geq \kappa_0$ and $t\geq t^*$, we have:
\begin{equation}
\label{eq:160a}
  \|\psi(t,\cdot)\|_2^2 \leq  C \kappa^{-\frac 13}
  \end{equation}
  Here we have chosen $\kappa_0$ such that:
  $$ 
  e^{- t^*\, h \kappa^{2} } \leq \kappa^{-2}\,,\, \forall \kappa \geq \kappa_0\,.
  $$

  We can now apply the above procedure for $t\geq nt^*$ with $n\geq1$ to
  obtain a generalization of \eqref{eq:160}. Suppose that we
  have found for some $n$ and for any $p\leq n$, two increasing
  sequences $\kappa_p$ and $\alpha_p$, and constants $C_p$ such that $\kappa \geq
  \kappa_p$ and for $t\geq p t^*(\kappa,M)$, we have
\begin{equation}\label{eq:160ap}
  \|\psi(t,\cdot)\|_2^2 \leq  C_p \kappa^{-\alpha_p}\,.
  \end{equation}
 Then, we show that the above is true for $p= n+1$.\\ 
 We rewrite \eqref{eq:160} with initial point $n t^*$ in the form:
 $$
  \|\psi(t,\cdot)\|_2^2 \leq    e^{- t^*\, h \kappa^{2} } +C\Big(\frac{1}{\kappa^{(1-\alpha)/2)}}\|\psi\|_{L^\infty
       (nt^*,t;L^2(\Omega,\C))}^2+\frac{1}{\kappa^\alpha}\Big)\,,\, \forall t \geq (n+1) t^*\,.
$$
Using the recursion argument, we obtain, for any $\alpha\in (0,1)$ and $\kappa\geq \kappa_n$,
$$
  \|\psi(t,\cdot)\|_2^2 \leq    \kappa^{-2} +C(\alpha)  \max (C_n,1) \, \Big(\frac{1}{\kappa^{\frac{(1-\alpha)}{2}+ \alpha_n}}+\frac{1}{\kappa^\alpha}\Big)\,,\, \forall t \geq (n+1) t^*\,.
$$
Optimizing over $\alpha$ leads to the choice: 
\begin{displaymath}
  \alpha_{n+1} = \frac{1}{3} + \frac{2}{3} \alpha_n \,.
\end{displaymath}
and to \eqref{eq:160ap} for $p= (n+1)$.\\

It can be easily shown that
\begin{displaymath}
  \alpha_n =1-\Big(\frac{2}{3}\Big)^n \,.
\end{displaymath}

Hence, for any $\beta<1/2$ there exists $n\in\N$ such that $\alpha_n\geq 2\beta$. It
follows that for all $t\geq n(\beta)t^*$,
\begin{equation}
  \label{eq:161}
 \|\psi(t,\cdot)\|_2\leq \frac{C_\beta}{\kappa^\beta} \,.
\end{equation}

We next substitute \eqref{eq:161}, in conjunction with \eqref{eq:2.2a},
into \eqref{eq:79} to obtain that for all $t \geq n_\beta t^*+\frac{2}{\lambda c}\ln
\kappa$
\begin{displaymath}
  \|A_1(t,\cdot)\|_2^2 \leq \frac{C}{\kappa^2} +   C \int_0^{n_\beta t^*}e^{-\lambda c(t-\tau)}\,d\tau
  + \int_{n_\beta t^*}^t e^{-\lambda c(t-\tau)} \frac{C_\beta^2}{\kappa^{2\beta}} \,d\tau \,.
\end{displaymath}
Hence, for  new constants $C, \hat C, C_\beta$,  we get:
\begin{displaymath}
   \|A_1(t,\cdot)\|_2^2 \leq \frac{C}{\kappa^2} + \frac{C_\beta}{\kappa^{2\beta}}\leq \frac{\hat C_\beta}{\kappa^{2\beta}} \,.
\end{displaymath}
The above, together with \eqref{eq:161},  readily verifies  \eqref{eq:81} with $t_\beta(\kappa) = n_\beta t^*(\kappa)  +\frac{2}{\lambda c}\ln \kappa$.
\end{proof}

\begin{corollary}
\label{cor:5.5}
Let $A_1$ be defined by \eqref{eq:65} and satisfy $\|A_1(\cdot,0)\|_2\leq
M$ for some $M>0$. Let further $0<\gamma<1/2$, and $t_\beta(\kappa,M)$ be defined
as in Lemma \ref{lem:5.l2bound}. For sufficiently large $\kappa$, there
exists $C_\gamma (h,c,\Omega)>0$ such that
  \begin{equation}
\label{eq:87}
  \|A_1(t,\cdot)\|_\infty  \leq \frac{C_\gamma}{\kappa^\gamma} \,,
  \end{equation}
for all $t\geq t_\beta+1$.
 \end{corollary}
\begin{proof}
  The proof is similar to the proof of Proposition \ref{lemma3.1}, and
  we therefore bring only a brief summary of it here. We first
  formulate the problem for
  $A_1$
\begin{subequations}  
\label{eq:163}
\begin{alignat}{2}
& \frac{\partial A_1}{\partial t} + c\, \curl^2A_1 =
\frac{1}{\kappa}\Im(\bar\psi\nabla_{\kappa A}\psi) -\nabla\phi_1 & \quad \text{ in }  \R_+\times\Omega
\,,\\
& \Div A_1 = 0 & \text{ in }  \R_+\times\Omega \,,\\
& A_1\cdot\nu=0 &\quad \text{ on } \R_+\times\partial\Omega \,,\\&
\dashint_{\partial\Omega}\curl A_1(t,x) \,ds =0 &\quad \mbox{ for all }  t>0\,.
\end{alignat}
\end{subequations}
As in \eqref{eq:67} we obtain here that
\begin{equation}
\label{eq:138}
  \|\nabla\phi_1\|_2 \leq \frac{c}{\kappa}\|\nabla_{\kappa A}\psi\|_2 \,.
\end{equation}
We can now apply \eqref{eq:147abstract} and
  \eqref{eq:148abstracta} to \eqref{eq:163} (in the case of $\mathcal
  L^{(1)}$) with $X=A_1$ on the interval $[t_0,T]$ for some
$t_0>t_\beta$, $t_1=t_0+1$, and $T=t_0+2$, to obtain 
\begin{multline*}
 \Big\|\frac{\partial A_1}{\partial t}\Big\|_{L^2(t_1,T; L^2(\Omega))} +
\|A_1\|_{L^2(t_1,T; H^2(\Omega))} + \|A_1\|_{L^\infty(t_1,T; H^1(\Omega))} 
          \\\leq C   \Big[\|A_1(t_0,\cdot )\|_2 +\Big \|\nabla\phi_1\|_{L^2(t_0,T;L^2(\Omega)}
          + \frac{1}{\kappa^2}\|\Im(\bar\psi\nabla_{\kappa A}\psi) \|_2\Big]
      \end{multline*}
With the aid of \eqref{eq:138}, \eqref{eq:82}, and \eqref{eq:2.2a} we
then conclude that there exists
$C_\beta(c,\Omega,h)$ such that for all $0<\beta<1/2$,
\begin{multline}
\label{eq:165}
 \Big\|\frac{\partial A_1}{\partial t}\Big\|_{L^2(t_1,T; L^2(\Omega))} +
\|A_1\|_{L^2(t_1,T; H^2(\Omega))} + \|A_1\|_{L^\infty(t_1,T; H^1(\Omega))} 
          \\\leq C   \Big[\|A_1(t_0,\cdot )\|_2 + \|\psi\|_{L^2(t_0,T;L^2(\Omega)}
          + \frac{1}{\kappa^2}\|\psi(t_0,\cdot)\|_2\Big]\leq \frac{C_\beta}{\kappa^\beta}\,.
      \end{multline}
   
We next obtain a bound on $ \|\nabla_{\kappa A}\psi\|_{L^\infty(t_0+1,T; L^2(\Omega))}$.  By
\eqref{eq:82} we have that
\begin{equation}
\label{eq:164}
  \int_{t_0}^T\|\nabla_{\kappa A}\psi(t,\cdot)\|_2^2 \,dt \leq
  \kappa^2\|\psi\|_{L^2(t_0,T;L^2(\Omega)}^2 -
  \frac{1}{2\kappa^2}\|\psi(t_0,\cdot)\|_2^2 \leq C\kappa^{2(1-\beta)}\,.
\end{equation}
It is easy to show that
\begin{displaymath}
  \frac{d}{dt}\|\nabla_{\kappa A}\psi(t,\cdot)\|_2^2 = - \Big\langle \frac{\partial\psi}{\partial t},
  \nabla_{\kappa A}^2\psi \Big\rangle - \Big\langle i\kappa\frac{\partial A}{\partial t}\psi,
  \nabla_{\kappa A}\psi \Big\rangle \,.
\end{displaymath}
We next use (\ref{eq:78}a) to obtain from the above that
\begin{multline*}
   \frac{d}{dt}\|\nabla_{\kappa A}\psi(t,\cdot)\|_2^2 =
   -\Big\|\frac{\partial\psi}{\partial t}\Big\|_2^2 +
   \kappa^2\frac{d}{dt}\Big(\frac{1}{2}\|\psi(t,\cdot)\|_2^2-
 \frac{1}{4}\|\psi(t,\cdot)\|_4^4\Big) \\ - \Big\langle i\kappa\frac{\partial A}{\partial t}\psi,
  \nabla_{\kappa A}\psi \Big\rangle -\Big\langle \frac{\partial\psi}{\partial t},
  i\kappa\phi\psi \Big\rangle\,.
\end{multline*}
Hence,
\begin{displaymath}
   \frac{d}{dt}\|\nabla_{\kappa A}\psi(t,\cdot)\|_2^2 \leq 2\frac{d}{dt}\Big(\frac{1}{2}\|\psi(t,\cdot)\|_2^2-
 \frac{1}{4}\|\psi(t,\cdot)\|_4^4\Big) +
 \kappa^2\Big\|\frac{\partial A_1}{\partial t}\Big\|_2^2 + \|\nabla_{\kappa A}\psi\|_2^2+
 \frac{\kappa^2}{2}\|\phi\|_2^2\,.
\end{displaymath}
By \eqref{eq:9a} and \eqref{eq:138} we have:
\begin{equation}
\label{eq:142}
  \|\phi\|_{1,2}^2 \leq C(\kappa^4 +  \|\nabla_{\kappa A}\psi\|_2^2)\,.
\end{equation}
Let $t_0<\tilde{t}_0<t_1$. Integrating the above on $(\tilde{t}_0,t)$
yields, with the aid of \eqref{eq:164} and \eqref{eq:165}, that for every
$t\in[t_1,T]$ it holds:
\begin{displaymath}
  \|\nabla_{\kappa A}\psi(t,\cdot)\|_2^2 \leq \|\nabla_{\kappa A}\psi(\tilde{t}_0,\cdot)\|_2^2 +
  C\kappa^6 \,.
\end{displaymath}
By \eqref{eq:164}, there exists $\tilde{t}_0\in(t_0,t_1)$ such that
\begin{displaymath}
  \|\nabla_{\kappa A}\psi(\tilde{t}_0,\cdot)\|_2^2\leq C\kappa^{2(1-\beta)} \,,
\end{displaymath}
which readily yields
\begin{equation}
\label{eq:166}
   \|\nabla_{\kappa A}\psi\|_{L^\infty(t_1,T; L^2(\Omega))}\leq C\kappa^3\,. 
\end{equation}

We next apply (\ref{eq:147abstract}) to (\ref{eq:1}a) on the interval
$(t_1,T)$ to obtain that
\begin{multline}
\label{eq:167}
  \|\psi\|_{L^\infty(t_1,T;H^1(\Omega,\C))} + \|\psi\|_{L^2(t_1,T;H^2(\Omega,\C))} 
  \leq C(\Omega)\,
  \big[\|\psi(t_1,\cdot )\|_{1,2} + \kappa\|A\cdot\nabla_{\kappa A}\psi\|_{L^2(t_1,T;
    L^2(\Omega,\C))}+\\ \kappa^2\||A|^2\psi\|_{L^2(t_1
    ,T;L^2(\Omega,\C))}+\kappa\|\phi\psi\|_{L^2(t_1
    ,T; L^2(\Omega,\C))} + \kappa^2\| \psi(1-|\psi|^2)\|_{L^2(t_1,T;
    L^2(\Omega,\C))} \big]\,.
 \end{multline}
We now estimate the various terms on the right-hand-side of
\eqref{eq:167}. 
For the first term we have, in view of \eqref{eq:165} and
\eqref{eq:166}, that
\begin{multline}
  \label{eq:169}
\|\psi(t_1,\cdot )\|_{1,2} \leq \|\psi(t_1,\cdot)\|_2 +
\|\nabla_{\kappa A}\psi(t_1,\cdot)\|_2+\kappa\|A\psi(\cdot,t_1)\|_2\leq \\ C\big[\kappa^{-\beta}+ \kappa^{1-\beta} +\kappa\big(\|A_1\|_{L^\infty(t_0,T;
    L^2(\Omega))}+\kappa^2\|A_n\|_\infty\big) \big] \leq \hat C\, \kappa^3 \,.
\end{multline}
For the second term we have by \eqref{eq:166} and
\eqref{eq:164} 
\begin{equation}
\label{eq:168}
   \kappa\|A\cdot\nabla_{\kappa A}\psi\|_{L^2(t_1,T;
    L^2(\Omega,\C))} \leq \kappa\,  \|\nabla_{\kappa A}\psi\|_{L^\infty(t_1,T;
    L^2(\Omega))}\big(\|A_1\|_{L^2(t_0,T; L^\infty(\Omega))}+\kappa^2\|A_n\|_\infty\big) \leq C\kappa^6 \,.
\end{equation}
For the third term we have, with the aid of \eqref{eq:2.2a},
\eqref{eq:165}, and Sobolev embedding 
\begin{equation}
\label{eq:170}
  \kappa^2\||A|^2\psi\|_{L^2(t_1
    ,T;L^2(\Omega,\C))} \leq \kappa^2\big(\|A_1\|_{L^\infty(t_0,T;
    H^1(\Omega))}^2+\|A_n\|_\infty^2\big)  \leq C\kappa^6 \,.
\end{equation}
For the fourth term we have, in view of \eqref{eq:142}
\begin{displaymath}
  \kappa\|\phi\psi\|_{L^2(t_1,T; L^2(\Omega,\C))} \leq C\kappa^3 \,.
\end{displaymath}
Finally, for the last term on the right hand-side we have
\begin{displaymath}
   \kappa^2\| \psi(1-|\psi|^2)\|_{L^2(t_1,T;
    L^2(\Omega,\C))}\leq C\kappa^{2-\beta} \,.
\end{displaymath}
Combining the above with \eqref{eq:169}, \eqref{eq:168}, and
\eqref{eq:170} yields
\begin{equation}
  \label{eq:171}
\|\psi\|_{L^\infty(t_1,T;H^1(\Omega,\C))} + \|\psi\|_{L^2(t_1,T;H^2(\Omega,\C))}  \leq
C\kappa^6 \,.
\end{equation}

We continue from here in precisely the same manner as in the proof of Lemma 
\ref{lemma3.1} to obtain that,
\begin{displaymath}
      \|A_1\|_{L^\infty(t_1,T; H^2(\Omega,\R^2))}
      \leq C \Big[\frac{1}{\kappa}\|\Im(\bar{\psi}\nabla_{\kappa A}\psi)\|_{L^2(t_2,T;H^1(\Omega,\R^2))}+\|A_1(t_2,\cdot )\|_2^2 \Big]\,.
\end{displaymath}
With the aid of \eqref{eq:171} it can then be proved that
\begin{displaymath}
   \|A_1\|_{L^\infty(t_1,T;H^2(\Omega,\R^2))} \leq C\kappa^9 \,,
\end{displaymath}
and hence
$$
  \|\nabla A_1\|_{L^\infty(t_1,T;H^1(\Omega,\R^2))} \leq C\kappa^9 \,.
$$
Note that by \eqref{eq:165} 
$$
\|\nabla A_1\|_{L^\infty (t_1,T;L^2(\Omega))} \leq \frac{C_\beta}{\kappa^\beta}\,.
$$
We next apply a standard interpolation inequality (cf. Section 7.1 in
\cite{gitr01}) together with the above to obtain
\begin{equation}
\label{eq:172}
  \|\nabla A_1(t,\cdot)\|_{2+\delta}\leq \|\nabla A_1(t,\cdot)\|_2^s\|\nabla A_1(t,\cdot)\|_p^{1-s}\leq C\, \kappa^{-s\beta+9(1-s)} \,,
\end{equation}
where
\begin{displaymath}
  s= \frac{\frac{1}{2+\delta}-\frac{1}{p}}{\frac{1}{2}-\frac{1}{p}} \geq 1-
  \delta\frac{p}{2p-4} \,.
\end{displaymath}
For $p\geq4$ we get
\begin{displaymath}
  s\geq1-\delta \,,
\end{displaymath}
which combined with \eqref{eq:172} yields
\begin{displaymath}
    \|\nabla A_1(t,\cdot)\|_{2+\delta}\leq C\kappa^{-\beta+\delta(9-\beta)} \,.
\end{displaymath}
Since we may choose $\delta$ to be arbitrarily small the corollary easily
follows with the aid of Sobolev embeddings. 
 \end{proof}
We can now conclude long-time decay of solutions of \eqref{eq:78}.
\begin{proposition}
\label{prop:5.6}
Suppose that for some $0<\alpha<1$, there exists some $C>0$ and $\kappa_0>0$
  we have 
  \begin{equation}
\label{eq:88}
    \sup_{\gamma\in\R}\|(\Lg_\kappa-i\gamma)^{-1}\| < \frac{1}{\kappa^2}\Big[1- \frac{C}{\kappa^{1-\alpha}}\Big]
  \end{equation} 
  for all $\kappa>\kappa_0$.  Then, there exists $\kappa_1\geq\kappa_0$ such that
  \eqref{eq:181} with $\nu=0$ holds true.
\end{proposition}
\begin{proof}
  The proof is almost identical with the proof of Theorem
  \ref{thm:1.2}, and hence we bring therefore only a brief summary of
  it. We first set
 \begin{displaymath}
  F=2i\kappa A_1\cdot\nabla_{\kappa A}\psi - |\kappa A_1|^2\psi+ i\kappa\phi_1\psi 
   +\kappa^2(1-|\psi|^2)\psi  \,.
\end{displaymath}
We then recall the definition of $\chi_{T,\epsilon}$ in \eqref{eq:58}, of
$\Gamma_\omega$ in \eqref{eq:57}, and of the Laplace transform, and then, in
the same manner we prove \eqref{eq:59} above, we show that for any
$0<\gamma<1$ there exists $C_\gamma>0$ such that
\begin{equation}
\label{eq:89}
  \|\widehat{\chi_{T,\epsilon}F}\|_{L^2(\Gamma_0\,;\,L^2(\Omega,\C))}^2 \leq
  \kappa^2\Big(1+\frac{C_\gamma}{\kappa^{1-\gamma}}
\Big)\big[\|\widehat{\chi_{T,\epsilon}\psi}\|_{L^2(\Gamma_0\,;\,L^2(\Omega,\C))}^2 +
\|\psi(\epsilon^{-1},\cdot)\|_2^2\big]\,.
\end{equation}

We next write (\ref{eq:78}a) in the form
\begin{displaymath}
   \frac{\partial\psi}{\partial t} + \Lg_\kappa\psi = F  \,,
 \end{displaymath}
and then take its Laplace transform. Then, we make use of \eqref{eq:89}
in conjunction with \eqref{eq:88} and Parseval's identity to obtain 
 \begin{multline*}
    \| \chi_{T,\epsilon}\psi\|_{L^2(\R_+\,;\,L^2(\Omega,\C))}^2 \leq   \Big[1-
    \frac{C}{\kappa^{1-\alpha}}\Big] (1+2\delta)\Big(1+\frac{C_\gamma}{\kappa^{1-\gamma}}+
    C\frac{\epsilon^2}{\delta}
 \Big)\| \chi_{T,\epsilon}\psi\|_{L^2(\R_+\,;\,L^2(\Omega,\C))}^2 \\
   + \frac{C}{\delta}\|\psi(\epsilon^{-1},\cdot )\|_2^2 \,.
 \end{multline*}
Finally, we choose $\delta=\epsilon=1/\kappa$ and $\gamma<\alpha$ to
obtain that
\begin{displaymath}
  \| \chi_{T,\epsilon}\psi\|_{L^2(\R_+\,;\,L^2(\Omega,\C))}^2 \leq C\kappa^4\, \|\psi(\epsilon^{-1},\cdot )\|_2^2 \,,
\end{displaymath}
and take the limit $T\to\infty$ to complete the proof.
\end{proof}
\newpage

\section{Resolvent estimates in the large domain limit}  
\label{section6}
\subsection{Presentation of the problem}
\rm
In the previous sections we have obtained sufficient conditions for
the stability of the semi-group associated with \eqref{eq:1}. These
conditions were phrased in terms of the resolvent norm of the linear
operator $\LL_h$, which is defined in (\ref{eq:44}). As the operator
is defined on a general class of domains in $\R^2$, we attempt to
estimate the resolvent $(\LL_h-\lambda)^{-1}$, in the large domain limit,
by approximate operators defined on $\R^2$ and $\R^2_+$, $h$
being   fixed and strictly positive.\\
Let then
$R>0$. We denote by $\Omega_R$ the image of $\Omega$ under the dilation
\begin{equation}
\label{eq:147}
  x\to R\, x \,.
\end{equation}
We assume that the domain $\Omega$ has the property (R1)-(R2) and that
assumptions   (J1)-(J3), (B)  and (C) are met.  \\
Denote the
transformed electric field by $\phi_R$. It satisfies the problem
\begin{displaymath}
  \begin{cases}
      \Delta\phi_R = 0 & \text{in } \Omega_R\,, \\
      \frac{\partial\phi_R}{\partial\nu}  = -\frac{J_R(x)}{\sigma}  & \text{on } \partial\Omega_R\,, 
  \end{cases}
\end{displaymath}
where the current density $J_R$ remains fixed except for the dilation 
$$
J_R(x)=J_r(x/R)\,,
$$
in which $J_r(x)$ is the reference current density defined in
\eqref{eq:187}.\\
Note that $$ \phi_R(x)=R\, \phi_n(x/R)\,.
$$
The transformed magnetic
potential, which we denote by $A_R$ then satisfies
\begin{equation}
\label{eq:90}
    \begin{cases}
     - \curl^2A_R + \frac{1}{c}\nabla\phi_R = 0 & \text{in } \Omega_R\,, \\
     \curl A_R = B_R(x) & \text{on }\partial\Omega_R\,,
  \end{cases}
\end{equation}
where $$B_R(x)=R\, B(x/R)\,.$$
 It can be easily proved that $$ A_R(x) =
R^2\, A_n(x/R)\,.
$$
Let then
\begin{equation}
\label{eq:91}
  \LL_h^R = -\nabla_{hA_R}^2 + ih\phi_R \,.
\end{equation}
The form domain associated with $\LL_h^R$ is given by
\begin{displaymath}
  H^{1,\pa \Omega_{R,c}}_0= \Big\{ u\in H^1(\Omega_R,\C)\,\Big|\, u|_{\partial\Omega_{R,c}}=0 \Big \} \,, 
\end{displaymath}
and  the domain of  $\LL_h^R$  is
\begin{displaymath}
  \Dg_R= \Big\{ u\in H^2(\Omega_R,\C)\,\Big|\, u|_{\partial\Omega_{R,c}}=0 \,, \; \frac{\partial u}{\partial\nu}\Big|_{\partial\Omega_{R,i}}=0\Big \} \,.
\end{displaymath}

We attempt to estimate $\sup_{\gamma\in\R}\| (\LL_h^R-\mu- i\gamma)^{-1}\|$ as
$R\to\infty$. Once the problem has been defined, we apply the inverse
transformation of \eqref{eq:147} to \eqref{eq:91} to obtain that
\begin{equation}
\label{eq:174}
  \sup_{\gamma\in\R}\| (\LL_h^R-\mu- i\gamma)^{-1}\| = R^2\sup_{\gamma\in\R}\| (\LL_{R^3h}-\mu R^2- i\gamma)^{-1}\| \,.
\end{equation}
We attempt to estimate the right-hand side in the sequel, as it is
more in line with the standard practice in semi-classical analysis
than the estimate of the left-hand-side. Furthermore, estimating the
right-hand-side would be valuable also for the analysis of the
penetration-scale problem presented in the previous section. In
particular, \eqref{eq:88}, which is a sufficient condition for global
stability of the normal state can be written in the form
\begin{equation}
  \kappa^2\sup_{\gamma\in\R}\| (\LL_{\kappa^3h}- i\gamma)^{-1}\| \leq 1 - \frac{C}{\kappa^{1-\alpha}}\,.
\end{equation}
Thus, an estimate of the right-hand-side is valuable for this problem
as well as for verifying that \eqref{eq:180} is satisfied. Additionally,
we can use the estimate of \eqref{eq:174} with $\mu=1$ to find the
critical current where the normal state looses its stability, which
amounts to determining the values of $h$ for which \eqref{eq:174}
becomes infinite.

Let $\gamma\in\R$, $B_n=\curl A_n$ and
 \begin{equation}\label{defF}
 F=\phi_n+icB_n\,.
 \end{equation}
 From \eqref{eq:8}, which is nothing else as the Cauchy-Riemann equations for $F$, we see that
$F$ is holomorphic in $\Omega$ as a function of $x_1+ix_2$.
\begin{lemma}
\label{uniquesol}
Under assumptions (B), (J1)-(J3) and (R1)-(R2), for any $\gamma$, $F-\gamma
/h$ has at most one simple zero in $\Omega$.
 \end{lemma}
 
  \begin{proof}
    Since $\gamma$ is real, any zero of $F-\gamma/h$ must lie in $B_n^{-1}
    (0)$.  It has been established in \eqref{hypnabla} that
    $B_n^{-1}(0)$ is either empty or that it is a regular curve
    $\Gamma$ joining the two components of $\pa \Omega_i$ on which $\nabla B_n \neq
    0$. By \eqref{eq:8}, $\nabla \phi_n \neq 0$ and is tangent to $\Gamma$. Hence
    $\phi_n$ is strictly monotone on $\Gamma$, which completes the proof of
    the lemma.
   \end{proof}
To obtain the supremum with respect to $\gamma$ of the resolvent as is
   clear from \eqref{eq:174}, we allow for dependence of $\gamma$ on $R$.\\
   
   {\em Suppose first} that a zero of $F-\gamma(R)/h$ exists in $\Omega$, and
   let $z_0(R)=(x_0,y_0)$ denote it.  We distinguish between two
   different cases:
\begin{subequations}
\label{eq:92}
  \begin{gather}
  d(z_0,\partial\Omega)\geq 2R^{\alpha -1  } \,,\\
d(z_0,\partial\Omega)< 2R^{\alpha -1  } \,,
\end{gather}
where $\alpha \in (0,1)$ will be determined later.\\
In the case (\ref{eq:92}a) we approximate $\|(\LL_{R^3h}-\lambda)^{-1}\|$ by
the norm of a similar operator on $L^2(\R^2)$, whereas in case
(\ref{eq:92}b), we approximate it by the norm of an operator on
$L^2(\R^2_+)$.\\

{\em Suppose next} that $F\neq\gamma(R)/h$ for all $x\in\overline{\Omega}$. Let $\Gamma$ denote
the set in $\Omega$ where $B_n=0$ ($\Gamma=B_n^{-1}(0)$). By Assumption (B) together with (R2),
this set is a single curve joining the two components of $\partial\Omega_c$, away
from the corners.  Denote the points of intersection of $\Gamma$ with $\partial\Omega$ by
$z_1$ and $z_2$.  Without loss of generality we assume that
$\phi_n(z_1)<\phi_n(z_2)$  and that $\gamma/h<\phi_n(z_1)$ (the case $\gamma/h > \phi_n(z_2)$ can be treated similarly).  We distinguish then
between two cases
\begin{gather}
-\gamma/h+\phi_n(z_1) \geq C^*R^{\alpha -1  } \,,\\
-\gamma/h + \phi_n(z_1)< C^*R^{\alpha -1  } \,
\end{gather}
\end{subequations}
where $C^*$ is determined below.

In the case (\ref{eq:92}d) we set
\begin{equation*}
  z_0 = z_1 +\frac{\frac{\gamma}{h}-F(z_1)}{F^\prime(z_1)}\quad  \,,
\end{equation*}
which is clearly well defined, since $|F^\prime(z_j)|=|\nabla\phi_n(z_j)|$ which
is strictly positive, as has already been stated in the proof of Lemma
\ref{uniquesol}.  Let
\begin{equation}
  d_0 = d(z_0,\Omega) \,.
\end{equation}
By Assumption (C) $\nabla\phi_n(z_1)$ is perpendicular to $\partial\Omega$ at $z_1$,
and so is, by the Cauchy-Riemann relations, $F^\prime(z_1)$. For sufficiently large $R$, 
we thus have $z_0\not\in \Omega$, 
\begin{equation}
\label{eq:185}
  d_0=|z_0-z_1|\,.
\end{equation}
We thus choose $$C^*=2|\nabla\phi_n(z_1)|$$ which guarantees that $$ d_0<2R^{\alpha
  -1 }\,,
  $$
  for $R$ large enough, when (\ref{eq:92}d) is
  met. \\
  Hence in Cases (\ref{eq:92}a,b,d), we have constructed a point
  $z_0(R)$ which for $R$ large enough is in a fixed neighborhood
  $\mathcal V(\Omega)$ of $\Omega$. We shall consolidate in the sequel the
  treatment the cases (\ref{eq:92}d) and (\ref{eq:92}b). In the case
  (\ref{eq:92}c) we shall prove that $\lambda_R$ is not in the spectrum and
  that $\|(\LL_{R^3h}-\lambda_R)^{-1}\|\to0$ as $R\to\infty$, where
\begin{equation}
\label{defgammaR}
\lambda_R = \nu(R) R^2
+i\gamma(R) R^3\,.
\end{equation}

 \subsection{A resolvent estimate}

 We seek an estimate for $\|(\LL_{R^3h}-\lambda_R)^{-1}\|$ for the cases
 (\ref{eq:92}a,b,d). Let then $\chi,\tilde \chi \in C^\infty(\R_+,[0,1])$ form a
 partition of unity satisfying
\begin{equation}
\label{eq:93}
  \chi(t) =
  \begin{cases}
    1 &\mbox{ for }  t \leq 1\,, \\
    0 & \mbox{ for }  t \geq  2\,,
  \end{cases}  
\end{equation}
and
\begin{displaymath}
  \chi^2 + \tilde \chi^2 =1 \mbox{ on } \mathbb R^+\,.
\end{displaymath}
We then introduce   $$\chi_1(x)=\chi(R^{1-\alpha  }|x-z_0|)  \mbox{ and }
\chi_2(x)= \tilde \chi(R^{1-\alpha  } |x-z_0|) \,,$$
where \begin{equation}\label{condalpha}
0<\alpha<1/3
\end{equation}
 is kept fixed throughout the sequel.\\

We establish the  following auxiliary result:
\begin{lemma}
\label{lem:5.1}
Suppose that  for some $\ell$ and $R_0$, $\nu(R)\leq \ell $ in \eqref{defgammaR}  for all $R\geq R_0$ and  let $z_0(R)$ as defined in the previous subsection for the cases (\ref{eq:92}a,b,d)). Then, there exist $R_1\geq R_0$ and $C>0$, such
that, if $R\geq R_1$ and $\lambda_R$ belongs to the
resolvent set $\rho(\LL_{R^3h})$ of $ \LL_{R^3h}$, then
  \begin{equation}
    \label{eq:94}
\| \chi_2(\LL_{R^3h}-\lambda_R)^{-1} \| \leq C\, R^{-\alpha}\,\big(R^{-2}+\|  (\LL_{R^3h}-\lambda_R)^{-1}\|\big) \,.
  \end{equation}
\end{lemma}
\begin{remark}
  In case (\ref{eq:92}c), we set $\chi_2=1$ on $\Omega$ to obtain $\lambda_R$ is
  not in the spectrum of $ \LL_{R^3h}$ and that:
\begin{equation}
    \label{eq:94aa}
\| (\LL_{R^3h}-\lambda_R)^{-1} \| \leq C\, R^{-\alpha -2}\,.
  \end{equation}
  \end{remark}
\begin{proof}~\\
It can be readily verified that there
  exists $C_0>0$ and $R_1$  such that for $R\geq R_1$ and  $x\in\Omega\setminus B(z_0,R^{\alpha -1 })$, 
  \begin{equation}
\label{eq:95}
 |h\phi_n(x) -\gamma | + |B_n(x) | \geq C_0\, R^{\alpha-1} \,.
  \end{equation}
Note that in Case (\ref{eq:92}c), this inequality is satisfied for $\forall
x\in \Omega$ as $B(z_0,R^{\alpha-1})$ is not in $\Omega$ in this case.
We then define the following subdomains of $\Omega\setminus B(z_0,R^{\alpha -1  })$:
\begin{align*}
  &\D_1^+ = \Big\{x\in\Omega\setminus B(z_0,R^{\alpha -1  })\,|\, B_n \geq 
  \frac{C_0}{2}R^{\alpha -1  }\Big\} ;\\
  &\widehat \D_1^+ = \Big\{x\in\Omega\setminus B(z_0,R^{\alpha -1  })\,|\, B_n \geq 
  \frac{C_0}{4}R^{\alpha -1  }\Big\} ;\\
  & \D_1^- = \Big\{x\in\Omega\setminus B(z_0,R^{\alpha -1  })\,|\, B_n \leq - 
  \frac{C_0}{2}R^{\alpha -1  }\Big\}; \\
  &\widehat \D_1^- = \Big\{x\in\Omega\setminus B(z_0,R^{\alpha -1  })\,|\, B_n \leq - 
  \frac{C_0}{4}R^{\alpha -1  }\Big\}; \\
  &\D_2^+ = \Big\{x\in\Omega\setminus B(z_0,R^{\alpha -1  })\,|\, h\phi_n-\gamma   \geq
  \frac{C_0}{2}R^{\alpha -1  }\Big\} ; \\
    &\widehat \D_2^+ = \Big\{x\in\Omega\setminus B(z_0,R^{\alpha -1  })\,|\, h\phi_n-\gamma   \geq   \frac{C_0}{4}R^{\alpha -1  }\Big\} ; \\
 &\D_2^- = \Big\{x\in\Omega\setminus B(z_0,R^{\alpha -1  })\,|\, h\phi_n-\gamma   \leq 
  -\frac{C_0}{2}R^{\alpha -1  }\Big\}\,; \\
   &\widehat \D_2^- = \Big\{x\in\Omega\setminus B(z_0,R^{\alpha -1  })\,|\, h\phi_n-\gamma   \leq - \frac{C_0}{4}R^{\alpha -1  }\Big\}  \,.
  \end{align*}
It readily follows from \eqref{eq:95} that
\begin{equation}\label{cov}
\Omega\setminus B(z_0,R^{\alpha -1  })\subseteq\D_1^+\cup\D_1^-\cup\D_2^+\cup\D_2^-\,.
\end{equation}
Denote then by $\chi_3,\chi_4, \chi_5,\chi_6\in C^\infty(\Omega,[0,1])$ cutoff
$R$-dependent functions satisfying 
\begin{displaymath}
  \chi_3(x) = 1 \mbox{ for } x\in \D_1^+\,,\, \supp \chi_3 \subset \widehat \D_1^+\,,\mbox{ and }
   |\nabla\chi_3| \leq CR^{1-\alpha} \,,
\end{displaymath}
\begin{displaymath}
  \chi_4(x) = 1 \mbox{ for } x\in \D_1^-\,,\, \supp \chi_4 \subset \widehat \D_1^-\,,\mbox{ and }
    |\nabla\chi_4| \leq CR^{1-\alpha} \,,
\end{displaymath}
\begin{displaymath}
  \chi_5(x) =1 \mbox{ for } x\in \D_2^+\,,\, \supp \chi_5 \subset \widehat \D_2^+ \,,\mbox{ and }
   |\nabla\chi_5| \leq  CR^{1-\alpha}\,,
\end{displaymath}
and 
\begin{displaymath}
  \chi_6(x) = 1  \mbox{ for } x\in \D_2^-\,,\, \supp \chi_6 \subset \widehat \D_2^- \,,\mbox{ and }
   |\nabla\chi_6| \leq  CR^{1-\alpha} \,.
\end{displaymath}
Having introduced these cut-off functions, 
we consider 
$$
f\in L^2(\Omega) \mbox{ and } u=(\LL_{R^3h}-\lambda_R)^{-1}f\,.
$$
In order to control $\chi_2 u$ in $L^2$ we successively control
$\chi_j \chi_2 u$ for $j=3,\cdots,6$, observing that the support of these
$\chi_j$ ($j=3,\cdots,6$) cover by \eqref{cov} the support of $\chi_2$.
Let $\eta_1=\chi_2\chi_3$. An integration by parts readily yields
\begin{equation}\label{minoraaa}
  \Re \langle\eta_1^2u,(\LL_{R^3h}-\lambda_R)u\rangle = \|\nabla_{R^3hA_n}(\eta_1u)\|_2^2 - \|u\nabla\eta_1\|_2^2
  -\nu R^2\|\eta_1 u\|_2^2 = \Re\langle \eta_1u,\eta_1f\rangle \,.
\end{equation}
For sufficiently large $R$ we have, in view of Theorem  7.1 in
\cite{boda06}, our condition on the support of $\chi_3$ and  Assumption (R1):
\begin{equation}\label{minoraab}
  \|\nabla_{R^3hA_n}(\eta_1u)\|_2^2 \geq\frac{\Theta_{\frac \pi 2} C_0}{4}R^{2+\alpha}\|\eta_1 u\|_2^2 \,,
\end{equation}
where $\Theta_{\frac \pi 2}$ is the lowest eigenvalue of the magnetic Laplacian with constant magnetic field equal to 
$1$  in an infinite sector.\\
Hence, since by assumption $\nu(R)$ is bounded, we obtain from \eqref{minoraaa} and \eqref{minoraab}
\begin{equation}
  \label{eq:96}
 \|\chi_2u\|_{L^2(\mathcal D_1^+)}^2 \leq CR^{-2\alpha}(R^{-4}\|f\|_2^2+R^{-\alpha} \|u\|_2^2) \,.
\end{equation}
 
We next introduce $\eta_2=\chi_2\chi_5$. An integration by parts yields again
\begin{displaymath}
  \Im \langle\eta_2^2u,\LL_{R^3h}u\rangle  = - 2 \Im\langle \eta_2 u\nabla\eta_2,\nabla_{R^3hA_n}u\rangle +
     \langle R^3(h\phi_n-\gamma)\eta_2u,\eta_2 u\rangle = \Im \langle\eta_2u,\eta_2f\rangle  \,.
\end{displaymath}
Since
\begin{displaymath}
 \frac{C_0}{4} R^{2+\alpha}  \|\eta_2u\|_2^2 \leq |\langle R^3(h\phi_n-\gamma)\eta_2u,\eta_2u\rangle| \,,
\end{displaymath}
we obtain that
$$
   \frac{C_0}{4}R^{2+\alpha}  \|\eta_2u\|_2^2 \leq ||\eta_2 u||_2 \,||f||_2 + C R^{1-\alpha} ||\eta_2 u|| \,  \|\nabla_{hR^3A_n}u\|_2\,,
$$
which leads to
$$
   \frac{C_0}{8}R^{2+\alpha}  \|\eta_2u\|_2^2 \leq \frac{16}{C_0}  R^{-2-\alpha} \,||f||_2^2 + \hat C R^{-3 \alpha} \,  \|\nabla_{hR^3A_n}u\|_2^2\,.
$$
Observing that
\begin{displaymath}
  \|\nabla_{hR^3A_n}u\|_2^2 \leq  \nu R^2\|u\|_2^2+ \|u\|_2\|f\|_2 \,,
\end{displaymath}
we finally get
\begin{equation}
\label{eq:97}
  \|\chi_2u\|_{L^2(\mathcal D_2^+)}^2\leq C\, R^{-2\alpha}(R^{-4}\|f\|_2^2+ \|u\|_2^2) \,.
\end{equation}
In a similar manner, we obtain that
\begin{displaymath}
   \|\chi_2u\|_{L^2(\mathcal D_1^-)}^2 + \|\chi_2u\|_{L^2(\mathcal D_2^-)}^2\leq CR^{-2\alpha}(R^{-4}\|f\|_2^2+ \|u\|_2^2) \,,
\end{displaymath}
which, together with \eqref{eq:96} and  \eqref{eq:97}, proves 
\begin{equation}
\label{eq:97a}
  \|\chi_2u\|_{2}^2\leq CR^{-2 \alpha}(R^{-4}\|f\|_2^2+ \|u\|_2^2)  \,,
\end{equation}
and the claim of the lemma.
\end{proof}

\subsection{The entire plane limit case}
Consider first the case (\ref{eq:92}a).
We choose a frame of reference whose origin is located at $z_0$, and
with the $x$ and $y$ axes respectively directed tangentially at $z_0$
to the level curves of $\phi_n$ and $B_n$. We recall that $\mathfrak j$ was introduced in \eqref{deffrakj}. It  follows then from \eqref{eq:90} that
   \begin{equation}\label{defja}
   {\mathfrak j}=h|\nabla\phi_n(z_0)|/c\neq 0\,.
   \end{equation}
    
We next define a potential function $V:\R^2\to\R$ via
\begin{equation}
\label{eq:60}
  \nabla V({\bf x}) = \frac{1}{2}{\mathbf x}\cdot D^2A_n(z_0){\mathbf x} + {\mathbf
    x}\cdot\nabla A_n(z_0) + A_n(z_0)- \frac{1}{2}{\mathfrak j} x^2\hat{i}_y
  \quad \mbox{and} \quad V(0)=0\,, 
\end{equation}
where ${\mathbf x}=(x,y)$.  It can be readily verified that $V$ is
properly defined, since the curl  of the right-hand-side
identically vanishes in $\mathbb R^2$ (recall that $B_n\sim {\mathfrak j} x$ near $z_0$).  We
can thus define
\begin{displaymath}
  \tilde{\LL}_h=-\nabla_{h[A_n-\nabla V]}^2 +ih\phi_n \,.
\end{displaymath}
By the gauge transformation $u \mapsto e^{i V}u$,   it readily follows that
$$\|(\tilde{\LL}_{R^3h}-\lambda_R)^{-1}\| =\|(\LL_{R^3h}-\lambda_R)^{-1}\|\,.
$$
For convenience we drop the accent from $\LL_{R^3h}$ in the sequel and
refer to $A_n -\nabla V$ as $A_n$. \\
We have introduced $\A({\mathfrak j}, c)$ and its domain in
\eqref{eq:98} and \eqref{zzz2}. 
Define the dilation operator $\Tg_{\mathfrak j} u(x) =
u(\mathfrak j ^{1/3}x)$. It can be readily verified that
\begin{equation}
\label{eq:103}
  \Tg_{\mathfrak j}^{-1}
\A({\mathfrak j},c)\Tg_{\mathfrak j} = {\mathfrak j}^\frac 23 \mathcal
A (1, c) \,,
\end{equation}
and hence ${\mathfrak j}^\frac 23 \A (1, c)$ is unitarily equivalent
to $\A({\mathfrak j},c)$.  As in \eqref{zzz5} that we set $ \A(z_0) =
\A({\mathfrak j}(z_0),c)$ or simply $\mathcal A$.  Recall from
\cite{aletal10}, that $D(\A)$ is the closure of $C_0^\infty(\R^2)$ under
the graph norm $\|u\|_2+\|\A u\|_2$, that the spectrum of $\mathcal
A$ is empty, and
$$
\|(\A-\lambda)^{-1}\|=\|(\A-\Re\lambda)^{-1}\| 
$$ 
for all $\lambda\in\C$.  Moreover (see Lemma 4.4 in \cite{aletal10}), for
any $\lambda_0\in \mathbb R$, there exists $C(\lambda_0)$ such that for all $\lambda
\in \mathbb C$ such that $\Re \lambda \leq \lambda_0$, we have:
\begin{equation}\label{spectreA}
\|(\A-\lambda)^{-1}\| \leq C(\lambda_0)\,.
\end{equation}

We can now state the following:
\begin{lemma}
  \label{lem:5.2}
Let $\lambda_R=R^2\nu(R)+iR^3\gamma(R)$.  Suppose that there exist
positive $R_0$ and $\ell $ such that for all $R\geq R_0$ we have:
\begin{enumerate}
\item $\nu(R) \leq \ell $; 
\item (\ref{eq:92}a) holds true;
\item $\lambda_R\in\rho(\LL_{R^3h})$.
\end{enumerate}
Let
\begin{equation}
  \label{eq:99}
M(R) = R^2\|(\LL_{R^3h}-\lambda_R)^{-1}\| \,.
\end{equation}
Then, there exist $C>0$ and $R_1\geq R_0$ depending only on $\ell $, $\Omega$,
${\mathfrak j}$, and $\alpha$, such that, for all $R\geq R_1$ we have
\begin{equation}
  \label{eq:100}
M(R) \leq \|(\A-\nu(R))^{-1}\| + C\, (R^{-\alpha}+R^{-(1-3\alpha)}) \big[\|(\A-\nu(R))^{-1}\|+1+M(R)^2\big] \,.
\end{equation}
\end{lemma}
\begin{proof}~\\
Let $U_R:L^2(\R^2)\to L^2(\R^2)$ denote the unitary dilation operator
\begin{displaymath}
  u\longmapsto U_Ru(x) = R^{-1}u(x/R) \,,
\end{displaymath}
and then set 
\begin{displaymath}
  \RR_\infty = R^{-2}U_R^{-1}(\A-\nu(R))^{-1}U_R \,.
\end{displaymath}
Note that 
$$\RR_\infty = (\A_R-\nu(R) R^2)^{-1}\,.
$$
where $\A_R(z_0):D(\A_R)\to L^2(\R^2)$ is given by
\begin{displaymath}
  \A_R = \Big(\nabla-iR^3{\mathfrak j} \frac{x^2}{2}\hat{i}_y\Big)^2
+ icR^3{\mathfrak j} y \,.
\end{displaymath}

  We attempt to approximate $(\LL_{R^3h}-\lambda_R)^{-1}$ by the following
  operator
\begin{equation}\label{defRR}
  \RR =  \chi_1\RR_\infty \chi_1 +
  \chi_2(\LL_{R^3h}-\lambda_R)^{-1}\chi_2 \,.
\end{equation}
Clearly, 
\begin{multline}
\label{eq:101}
  (\LL_{R^3h}-\lambda_R)\RR = I + [\LL_{R^3h},\chi_1]\RR_\infty\chi_1 +
  \chi_1(\LL_{R^3h}-iR^3\gamma-\A_R)\RR_\infty\chi_1 \\ + [\LL_{R^3h},\chi_2]
  (\LL_{R^3h}-\lambda_R)^{-1}\chi_2 \,,
\end{multline}
We now  estimate the norms of the three operators appearing on the right-hand-side
of \eqref{eq:101} after the identity. 

{\em For the first operator}, we evaluate the commutator as follows:
\begin{multline}\label{term2}
  [\LL_{R^3h},\chi_1]=-\Delta\chi_1 + 2\nabla\chi_1\cdot\nabla_{R^3hA_n} = \\-\Delta\chi_1 +
  2\nabla\chi_1\cdot\Big(\nabla-iR^3{\mathfrak j} \frac{x^2}{2}\hat{i}_y\Big)+
  2i\nabla\chi_1\cdot R^3\Big(hA_n- {\mathfrak j} \frac{x^2}{2}\hat{i}_y\Big) \,,
\end{multline}
and then successively estimate in $\mathcal L (L^2)$ the three resulting terms on
the right-hand-side, i.e, $-\Delta\chi_1\RR_\infty\chi_1$, $2\nabla\chi_1\cdot\Big(\nabla-iR^3{\mathfrak j} \frac{x^2}{2}\hat{i}_y\Big)\RR_\infty\chi_1$, and 
$2i\nabla\chi_1\cdot R^3\Big(hA_n- {\mathfrak j}
\frac{x^2}{2}\hat{i}_y\Big)\RR_\infty\chi_1$.

For the first term, since by assumption $\nu(R)\leq \ell
\,,$ we observe,  by using  \eqref{spectreA}, that
\begin{equation}
\label{eq:175}
  \|\RR_\infty\|=\frac{\|(\A-\nu(R))^{-1}\|}{R^2} \leq \frac{C(\ell)}{R^2}  \,, 
\end{equation}
and that
\begin{equation*}
\|-\Delta\chi_1 \|\leq C R^{2-2 \alpha}\,.
\end{equation*}
We can, thus, readily conclude
\begin{equation}
\label{eq:175a}
\|-\Delta\chi_1 \RR_\infty\chi_1\| \leq C R^{-2 \alpha}\,.
\end{equation}

To estimate the second term, let $u=\mathcal R_\infty \chi_1f$ for some $f\in L^2(\Omega)$. 
It follows that 
\begin{displaymath}
  \Big\|\Big(\nabla-iR^3{\mathfrak j}\frac{x^2}{2}\hat{i}_y\Big)u\Big\|_2^2 = \nu(R) R^2\|u\|_2^2 + 
  \Re \langle u,f\rangle \leq \nu(R) R^2\|u\|_2^2 + \|u\|_2\|f\|_2\,.
\end{displaymath}
Thus, by \eqref{eq:175},
\begin{equation}\label{eq:104aa}
  \Big\|\Big(\nabla-iR^3{\mathfrak j}\frac{x^2}{2}\hat{i}_y\Big)\mathcal R_\infty \chi_1f\Big\|_2
  \leq \frac{C}{R}\|f\|_2 \,,
\end{equation}
and hence,
\begin{equation}\label{eq:104}
 \Big\|\nabla\chi_1\cdot\Big(\nabla-iR^3{\mathfrak j}\frac{x^2}{2}\hat{i}_y\Big) \mathcal R_\infty \chi_1\Big\| \leq C\, R^{-\alpha} \,.
\end{equation}
For the third sub-term, as  in view of \eqref{eq:60}, 
\begin{equation}
\label{eq:102}
  \| hA_n- {\mathfrak j} \frac{x^2}{2}\hat{i}_y\|_{L^\infty(B(z_0,2R^{\alpha-1}))} \leq
  CR^{-3(1-\alpha)} \,,
\end{equation}
we obtain 
\begin{displaymath}
 \|  2iR^3\nabla\chi_1\cdot\Big(hA_n- {\mathfrak j} \frac{x^2}{2}\hat{i}_y\Big) \|_\infty \leq CR^{1+2\alpha} \,,
 \end{displaymath}
 and then, using \eqref{eq:175} yields
 \begin{equation}\label{eq:102a}
 \| 2i\nabla\chi_1\cdot R^3\Big(hA_n- {\mathfrak j} \frac{x^2}{2}\hat{i}_y\Big)\RR_\infty\chi_1\| \leq
  C \,C(\ell) R^{1+2 \alpha}\, R^{-2} \leq 
 \widehat  CR^{-(1-2\alpha)} \,.
\end{equation}
Combining \eqref{eq:175a}, \eqref{eq:104} and \eqref{eq:102a}, we obtain
\begin{equation}
  \label{eq:105}
\|[\LL_{R^3h},\chi_1]\mathcal R_\infty\, \chi_1\| \leq C(R^{-\alpha} + R^{2 \alpha -1  }) \leq \tilde  C R^{-\alpha}\,,
\end{equation}
for any choice of $\alpha\in (0,\frac 13)$.\\

{\em Consider next the second operator} after the identity   on the right hand side of \eqref{eq:101}. We use the decomposition:
\begin{multline}
  (\LL_{R^3h}-iR^3\gamma(R) -\A_R)=\\ 2iR^3\Big(hA_n
  -{\mathfrak j}\frac{x^2}{2}\hat{i}_y\Big)\cdot\Big(\nabla
  -iR^3{\mathfrak j}\frac{x^2}{2}\hat{i}_y\Big) + R^6\Big|hA_n
  -{\mathfrak j}\frac{x^2}{2}\hat{i}_y\Big|^2 + iR^3(\phi_n-\gamma(R)-c{\mathfrak j}y)\,.
\end{multline}
By \eqref{eq:102a}  we have that
\begin{equation}
\label{eq:106}
 \Big\|  \chi_1R^6\Big|hA_n
  -{\mathfrak j}\frac{x^2}{2}\hat{i}_y\Big|^2\, \mathcal R_\infty\,\chi_1\Big\| \leq CR^{-2(1-3\alpha)} \,.
\end{equation}
Similarly, by \eqref{eq:102} and \eqref{eq:104aa} we obtain that
\begin{equation}
\label{eq:107}
  \Big\|\chi_1R^3\Big(hA_n
  -{\mathfrak j}\frac{x^2}{2}\hat{i}_y\Big)\cdot(\nabla
  -iR^3{\mathfrak j}\frac{x^2}{2}\hat{i}_y\Big)\,\mathcal R_\infty\, \chi_1\Big\|\leq CR^{-(1-3\alpha)}   \,.
\end{equation}
Finally, as
\begin{displaymath}
  \|\phi_n-\gamma(R) -c{\mathfrak j}y\|_{L^\infty(B(0,2R^{\alpha -1  }))}\leq CR^{-2(1-\alpha)} \,,
\end{displaymath}
we obtain, 
\begin{equation}\label{eq:107aa}
\| R^3\chi_1 \left(\phi_n-\gamma(R) -c{\mathfrak j}y\right)\mathcal R_\infty \chi_1\| \leq C R^{-1 + 2 \alpha}\,.
\end{equation}
Combining \eqref{eq:106}, \eqref{eq:107}, and \eqref{eq:107aa} we
obtain for the norm of the second operator:
\begin{equation}
\label{eq:108}
  \|\chi_1(\LL_{R^3h}-iR^3\gamma(R)-\A_R)\, \mathcal R_\infty\, \chi_1\| \leq CR^{-(1-3\alpha)}  \,.
\end{equation}

{\em For the last operator}  on the right hand side of \eqref{eq:101} we use the decomposition:
\begin{displaymath}
  [\LL_{R^3h},\chi_2] (\LL_{R^3h}-\lambda_R)^{-1}\chi_2 =  (-\Delta\chi_2 +
  2\nabla\chi_2\cdot\nabla_{R^3hA_n})(\LL_{R^3h}-\lambda_R)^{-1}\chi_2 \,.
\end{displaymath}
It can be easily verified, as in the derivation of \eqref{eq:104}, that
\begin{displaymath}
  \|\nabla_{R^3hA_n}(\LL_{R^3h}-\lambda_R)^{-1}\|\leq \frac{C}{R}(1+ M(R))\,.
\end{displaymath}
Hence,
\begin{equation}
\label{eq:109}
  \|[\LL_{R^3h},\chi_2] (\LL_{R^3h}-\lambda_R)^{-1}\chi_2\|\leq C\frac{M(R)+1}{R^\alpha} \,. 
\end{equation}

We next rewrite \eqref{eq:101} in the following form
\begin{multline*}
    (\LL_{R^3h}-\lambda_R)\left( \RR - (\LL_{R^3h}-\lambda_R)^{-1} [\LL_{R^3h},\chi_2]
  (\LL_{R^3h}-\lambda_R)^{-1}\chi_2\right) \\= I + [\LL_{R^3h},\chi_1](\A_R-\nu(R) R^2)^{-1}\chi_1 +
  \chi_1(\LL_{R^3h}-iR^3\gamma-\A_R)(\A_R-\nu(R) R^2)^{-1}\chi_1 \,.
\end{multline*}
In view of \eqref{eq:105} and \eqref{eq:108} we obtain
that for sufficiently large $R$ the right-hand-side of the above
identity becomes invertible. Hence,
\begin{multline}
\label{eq:110}
   (\LL_{R^3h}-\lambda_R)^{-1}= \Big(\RR- (\LL_{R^3h}-\lambda_R)^{-1} [\LL_{R^3h},\chi_2]
  (\LL_{R^3h}-\lambda_R)^{-1}\chi_2\Big) \times \\ \times \Big( I + [\LL_{R^3h},\chi_1]\RR_\infty \chi_1+
  \chi_1(\LL_{R^3h}-iR^3\gamma(R) -\A_R)\RR_\infty \chi_1 \Big)^{-1} \,.
\end{multline}
By \eqref{eq:94} there exists $C>0$ such that 
\begin{equation}\label{eq:110aa}
  \|\RR\|\leq \frac{\|(\A-\nu(R))^{-1}\|}{R^2} + CR^{-(2+\alpha)}\big(1+M (R) \big) \,.
\end{equation}
By \eqref{eq:105}, \eqref{eq:108}, \eqref{eq:109}, 
\eqref{eq:110}, and \eqref{eq:110aa},  we easily obtain \eqref{eq:100}.
\end{proof}

\subsection{The half-plane limit case}

We next consider together the cases (\ref{eq:92}b) and (\ref{eq:92}d).
Let $\tilde{z}_0(R)$ denote the projection of $z_0(R)$ on $\partial\Omega$,
($d(z_0,\partial\Omega)=d(z_0(R) ,\tilde{z}_0)$).  Note that by Assumption (B) (see
\eqref{(B)}) and (R2), since the curve $\Gamma = \{B_n^{-1}(0)\}$ intersects $\partial\Omega$
on the interior of $\partial\Omega_c$ it follows that $\partial\Omega$ is smooth near
$\tilde{z}_0(R)$. Hence, since $d(z_0,\partial\Omega)\leq 2R^{\alpha -1 }$,
$\tilde{z}_0(R)$ must exist.  In case (\ref{eq:92}d) where $z_0$
  lies outside $\Omega$ we have $\tilde z_0(R) = z_1$.  We define a curvilinear
coordinate system $(s,t)$ such that $t=d(x,\partial\Omega)$, and $s$ denotes the
arc length of $\partial\Omega$ from $\tilde{z}_0(R)$ in the positive
trigonometric direction to $z(s)$, which denotes the projection of
$x(s,t)$ on $\partial\Omega$. Thus (cf.  \cite{lupa99}),
\begin{displaymath}
  x=\F(s,t)=z(s)-t\nu(s) \,,
\end{displaymath}
where $\nu(s)$ denotes the outward normal on $\partial\Omega$ at $z(s)$. We
further set
\begin{displaymath}
  g = |\det \D \F|=1-t\kappa_r(s) \,,
\end{displaymath}
where $\kappa_r$ denotes the relative curvature of $\partial\Omega$ at
$z(s)$. Note that, outside a fixed neighborhood of the corners,  there exists $C>0$
depending only on $\Omega$ such that
\begin{equation}
\label{eq:111}
  |\kappa_r|+|\kappa_r^\prime| \leq C\,.
\end{equation}
In case (\ref{eq:92}b) we then set
\begin{equation}
\label{eq:176}
  \nabla V_+({\bf x}) = \frac{1}{2}{\mathbf x}\cdot D^2A_n(z_0){\mathbf x} + {\mathbf
    x}\cdot\nabla A_n(z_0) + A_n(z_0)- \frac{1}{2}{\mathfrak j} ({\mathbf x}\cdot\hat{i}_s(z_0))^2\hat{i}_t(z_0) \,,
\end{equation}
where $V_+(0)=0$.
In the case (\ref{eq:92}d) we set
\begin{displaymath}
  \tilde{A}_n ({\bf x})= \frac{1}{2}{\mathbf x}\cdot D^2A_n(z_1){\mathbf x} + {\mathbf
    x}\cdot\nabla A_n(z_1) + A_n(z_1)\,,
\end{displaymath}
and then define $V_+$ in the following manner
\begin{displaymath}
  \nabla V_+({\bf x}) = \frac{1}{2}{\mathbf x}\cdot D^2\tilde{A}_n(z_0){\mathbf x} + {\mathbf
    x}\cdot\nabla \tilde{A}_n(z_0) + \tilde{A}_n(z_0)+ \frac{1}{2}{\mathfrak j} ({\mathbf x}\cdot\hat{i}_s(z_0))^2\hat{i}_t(z_0) \,.
\end{displaymath}
Then we let
\begin{displaymath}
  \tilde{\LL}_h=-\nabla_{h[A_n-\nabla V_+]}^2 +ih\phi_n \,,
\end{displaymath}
and attempt to estimate $\|(\tilde{\LL}_{R^3h}-\lambda_R)^{-1}\|$.  As
before, we drop the accent from $\LL_{R^3h}$ in the sequel and refer
to $A_n -\nabla V_+$ as $A_n$.

We can now write $\LL_{R^3h}$ in terms of the curvilinear coordinates
(cf. \cite{lupa99})
\begin{displaymath}
  \LL_{R^3h} = -\Big(\frac{1}{g}\Big[\frac{\partial}{\partial s}-ihR^3a_s\Big]\Big)^2
  - 
  \frac{1}{g}\Big(\frac{\partial}{\partial t}-ihR^3a_t\Big)g\Big(\frac{\partial}{\partial t}-ihR^3a_t\Big) +ihR^3\phi_n(s,t) \,,
\end{displaymath}
where
\begin{displaymath}
  a_s= gA\cdot\hat{i}_s \quad ; \quad  a_t= A\cdot\hat{i}_t \,.
\end{displaymath}
\\~
We have defined the operator $\A_+({\mathfrak j},c)$ in \eqref{eq:112}.
Since \eqref{eq:103} is valid for $\A_+$ as well, it
readily follows that $\A_+({\mathfrak j}, c)$ is  unitarily equivalent
to ${\mathfrak j}^\frac 23 \mathcal A_+ (1, c)$. As before (see \eqref{zzz5}), we may set
$ \A_+ (z_0)=\A_+({\mathfrak j} (z_0),c)\,$ or more simply $\A_+$.
The domain of $\A_+$ is given by \eqref{zzz2} (cf. \cite{aletal13}).\\
 Note that, by  Assumption (C),  
$\nabla\phi_n$ (which is always orthogonal to $\nabla B_n$ by \eqref{eq:8}) is perpendicular to $\partial\Omega$ at the
intersection with the curve $\Gamma$, which explains why $\nabla\phi_n$ is  parallel to 
$\hat{i}_t$ in \eqref{eq:112}. 

We begin by stating a rather standard estimate.
\begin{lemma}
\label{lem:aux}
Let $\nu \in \mathbb R$, ${\mathfrak j} \neq 0$ and $c\neq 0$ . Then, there
exists $C>0$, such that, for any $f\in L^2(\R^2_+)$, with compact
support in $\overline{\R^2_+}$, and for all $\gamma\in \mathbb R$ such that
$\lambda=\nu+i\gamma\in\rho(\A_+(\mathfrak j,c))$, we have:
  \begin{equation}
    \label{eq:113}
\|u_{ss}\|_2 \leq C(\|u\|_2+\|(1+|s|)f\|_2) \,,
  \end{equation}
  where $u=(\A_+(\mathfrak j,c) -\lambda)^{-1}f$. 
\end{lemma}
\begin{proof}~\\
  We first establish local estimates before assembling them together by a covering argument.\\
  Let $s_0\in\R$. We set $v=e^{{\mathfrak j} s_0^2t/2}u$. We then have
  \begin{equation}
\label{eq:114}
    \begin{cases}
  -\Delta v=    - i{\mathfrak j}(s^2-s_0^2)\frac{\partial v}{\partial t} -
  \Big[{\mathfrak j}^2\frac{(s^2-s_0^2)^2}{4} + i\,c{\mathfrak j}t-\lambda\Big]v+f 
 & \text{in } \R^2_+ \\
v=0 & \text{on } \partial\R^2_+ \,.
    \end{cases}
   \end{equation}
Let $\tau_0\geq1$, and let $x_0=(s_0,\tau_0)$. It is easy to show that
\begin{displaymath}
   -\Delta v= -i{\mathfrak j}(s^2-s_0^2)\Big(\frac{\partial}{\partial t} -
   i{\mathfrak j}\frac{s^2-s_0^2}{2}\Big)v + \Big[{\mathfrak j} ^2\frac{(s^2-s_0^2)^2}{4} -
   ic{\mathfrak j}t+\lambda\Big]v+f \,.
\end{displaymath}
As $|s^2-s_0^2|\leq 1+2|s_0|$ in $B(x_0,1)$, standard elliptic estimates
show that
\begin{multline}
\label{eq:115}
  \|u_{ss}\|_{L^2(B(x_0,1/2))}= \|v_{ss}\|_{L^2(B(x_0,1/2))} \leq C(\nu,{\mathfrak j})\Big[(1+|s_0|)\Big\|\Big(\frac{\partial}{\partial t} -
   i{\mathfrak j}\frac{s^2-s_0^2}{2}\Big)v\Big\|_{L^2(B(x_0,1))} +\\
   (1+s_0^2)\|v\|_{L^2(B(x_0,1))}+ \|(t-\gamma)v\|_{L^2(B(x_0,1))} +
   \|f\|_{L^2(B(x_0,1))}\Big] \,.
\end{multline}

Recall the cutoff function $\chi$ defined by
\eqref{eq:93}. Multiplying \eqref{eq:114} by $\chi^2\bar{v}$ and integrating 
yield for the real part
\begin{displaymath}
  \Big\|\Big(\nabla-i{\mathfrak j}\frac{s^2-s_0^2}{2}\hat{i}_t\Big)(\chi v)\|_2^2
  = \|v\nabla\chi\|_2^2 + \nu\|\chi v\|_2^2 + \Re \langle\chi f,\chi v\rangle\,,
\end{displaymath}
which holds also for $\tau_0>2\,$.\\
 Consequently,
\begin{displaymath}
  \Big\|\Big(\nabla-i{\mathfrak j}\frac{s^2-s_0^2}{2}\hat{i}_t\Big)v\|_{L^2(B(x_0,1))}\leq  
C(\nu)[\|v\|_{L^2(B_+(x_0,2))} + \|f\|_{L^2(B_+(x_0,2))}]  \,,
\end{displaymath}
where $B_+(x_0,r)=B(x_0,r)\cap\R^2_+$.
Substituting the above into \eqref{eq:115} yields, using the fact that
$|s_0|\leq|s|+1$ in $B(x_0,1)$,
\begin{equation}
\label{eq:116}
\begin{array}{ll}
  \|u_{ss}\|_{L^2(B(x_0,1/2))} & \leq C(\nu,{\mathfrak j})\left[\|(1+|s|^2)v \|_{L^2(B(x_0,2)\cap\R^2_+)} \right. \\
 &\quad\quad +  \left.  \|(t-\gamma)v\|_{L^2(B(x_0,2)\cap\R^2_+)} +
   \|(1+|s|)f\|_{L^2(B(x_0,2)\cap\R^2_+)}\right] \,.
   \end{array}
\end{equation}
In a similar manner we obtain that, when $\tau_0=0$, we have 
\begin{equation}
\label{eq:117}
\begin{array}{ll}
  \|u_{ss}\|_{L^2(B_+(x_0,1))} &\leq C(\nu,{\mathfrak j})\left[ \|(1+|s|^2)v \|_{L^2(B_+(x_0,4))} \right.\\
 &\quad + \left.  \|(t-\gamma)v\|_{L^2(B_+(x_0,4))} +
   \|(1+|s|)f\|_{L^2(B_+(x_0,4))}\right] \,.
   \end{array}
\end{equation}
 We can now define a covering  of
$\R^2_+$ by discs of radius $1/2$ and semi-discs of radius $1$
centered on the boundary, such that each point in $\R^2_+$ is
covered a finite number of times. Summing up \eqref{eq:116} and
\eqref{eq:117} over these discs and semi-discs yields
\begin{equation}
  \label{eq:118}
 \|u_{ss}\|_2 \leq C(\nu,{\mathfrak j})\, [\|(1+|s|^2)u \|_2 +\\
  \|(t-\gamma)u\|_2 +
   \|(1+|s|)f\|_2] \,.
\end{equation}
Note that in \cite{aletal13}  it has been established that $u$ has
finite moments of any order, when $f$ is compactly supported in $\overline{\mathbb R^2_+}$. 

We next estimate $\|s^2u\|_2$. It has been proved in
\cite{aletal13} (cf. (5.20) therein) that for $k\geq0$ we have
\begin{displaymath}
  \||s|^{(k+1)/2}u\|_2^2 \leq C(\nu,{\mathfrak j})\,\left[ \|(1+|s|^{k/2})u\|_2^2 +
  |\langle|s|^{(k+1)/2}u,|s|^{(k-1)/2}f\rangle \right] \,.
\end{displaymath}
From this  it readily follows that for $k\geq1$
\begin{equation}
\label{eq:119}
   \||s|^{(k+1)/2}u\|_2 \leq C\,\left[ \|(1+|s|^{k/2})u\|_2 +
  \||s|^{(k-1)/2}f\|_2 \right] \,.
\end{equation}
For $k=0\,$, we have 
\begin{displaymath}
   \||s|^{1/2}u\|_2 \leq C\, [ \|u\|_2 +
  \|f\|_2 ] \,.
\end{displaymath}
Applying the above together with \eqref{eq:119} recursively for
$k=1,2,3$ yields
\begin{equation}
  \label{eq:120}
\|s^2u\|_2 \leq C(\nu,{\mathfrak j})\,\left[ \|u\|_2 +
  \|(1+|s|)f\|_2 \right] \,.
\end{equation}
Finally, we obtain an estimate for $\|(t-\gamma)u\|_2$. An integration by
parts yields
\begin{displaymath}
  \Im\langle(t-\gamma)u\,,\,(\A_+(\mathfrak j,c) -\lambda)u\rangle=\Im\Big\langle u\,,\, \Big(\frac{\partial}{\partial t} -
   i\,{\mathfrak j}\frac{s^2}{2}\Big)u\Big\rangle + {\mathfrak j}c\||t-\gamma|u\|_2^2 = \Im\langle(t-\gamma)u\,,\, f\rangle \,,
\end{displaymath}
from which it readily follows that
\begin{equation}
\label{eq:121}
  \||t-\gamma|\, u\|_2^2 \leq   C\, \Big[\||t-\gamma|\, u\|_2\|f\|_2 +  \|u\|_2\Big\| \Big(\frac{\partial}{\partial t} -
   i{\mathfrak j}\frac{s^2}{2}\Big)u\Big\|_2\Big] \,.
\end{equation}
As
\begin{displaymath}
  \Re\langle u\,,\,(\A_+(\mathfrak j,c) -\lambda)u\rangle =
  \Big\|\Big(\nabla-i{\mathfrak j}\frac{s^2}{2}\hat{i}_t\Big)u\Big\|_2^2 -\nu\, \|u\|_2^2
  = \Re\langle u\,,\,f\rangle \,,
\end{displaymath}
we immediately conclude that
\begin{displaymath}
  \Big\| \Big(\frac{\partial}{\partial t} -
   i\, {\mathfrak j}\frac{s^2}{2}\Big)u\Big\|_2 \leq C\, [ \|u\|_2 +
  \|f\|_2 ] \,.
\end{displaymath}
Substituting into \eqref{eq:121} then yields
\begin{displaymath}
 \||t-\gamma|\, u\|_2 \leq C\, [ \|u\|_2 +
  \|f\|_2 ] \,,
\end{displaymath}
which, in conjunction with \eqref{eq:120} and \eqref{eq:118} yields \eqref{eq:113}.
\end{proof}

We can now provide another estimate for $\|(\LL_{R^3h}-\lambda_R)^{-1}\|$.
\begin{lemma}
  \label{lem:5.3}
Let $\lambda_R=\nu(R) R^2+i\, R^3\gamma(R)$. Suppose that there exist
positive $R_0$ and $\ell $ such that for all $R\geq R_0$ we have
\begin{enumerate}
\item $\nu(R) \leq \ell $;
\item (\ref{eq:92}b) or (\ref{eq:92}d) hold true;
\item $\lambda_R\in\rho(\LL_{R^3h})$.
\end{enumerate}
Let $M(R)$ be defined by \eqref{eq:99}.  Then, there exist $C>0$ and
$R_1\geq R_0$ depending only on $\ell $, $\Omega$, ${\mathfrak j}$, and $\alpha$, such that for
all $R\geq R_1$ we have
\begin{equation}
\label{eq:122}
M(R)  \leq \|(\A_+-\nu -i\, c{\mathfrak j}Rt_0)^{-1}\|  + C\,(R^{-\alpha}+R^{-(1-3\alpha)})\, 
\big[\|(\A_+-\nu -i\, c{\mathfrak j}Rt_0)^{-1}\|  +1+M(R)^2\big] \,,
\end{equation}
where  $t_0=d_0$ in case (\ref{eq:92}b)  and $t_0=-d_0$  in case (\ref{eq:92}d).
\end{lemma}

\begin{proof}~\\
  We prove \eqref{eq:122} in a similar manner  to \eqref{eq:100}.
  Recall the definition of $\chi$ from \eqref{eq:93}. We let
  $\eta_1=\chi(4R^{\alpha -1  } |x-z_0|)\, {\mathbf 1}_{\Omega}$ and
  $\eta_2=\tilde  \chi(4R^{\alpha -1  } |x-z_0|)\,  {\mathbf 1}_{\Omega}$. Let further
  \begin{displaymath}
    \RR_+^\infty = U^{-1}_R(\A_+-\nu-i\, c{\mathfrak j}Rt_0)^{-1}U_R=(\A_+^R-R^2[\nu-i\, c{\mathfrak j}Rt_0])^{-1}\,,
  \end{displaymath}
where
\begin{displaymath}
  \A_+^R=\Big(\nabla-iR^3{\mathfrak j}\frac{s^2}{2}\hat{i}_t\Big)^2
+ icR^3{\mathfrak j}t \,.
\end{displaymath}
Then we
  set
\begin{displaymath}
  \RR_+ =  \eta_1\RR_+^\infty  \eta_1 + \eta_2\,(\LL_{R^3h}-\lambda_R)^{-1}\,\eta_2 \,.
\end{displaymath}
Clearly, 
\begin{multline}
\label{eq:123}
  (\LL_{R^3h}-\lambda_R)\RR_+ = I + [\LL_{R^3h},\eta_1]\, \RR_+^\infty\, \eta_1 +\\
  \eta_1\, \big(\LL_{R^3h}-i\,R^2(\gamma-c{\mathfrak j}Rt_0)-\A_+^R\big)\RR_+^\infty\, \eta_1 + [\LL_{R^3h},\eta_2]
  (\LL_{R^3h}-\lambda_R)^{-1}\, \eta_2 \,.
\end{multline}

The estimate for the various terms on the right-hand-side
of \eqref{eq:123} proceeds in exactly the same manner as in the proof
of \eqref{eq:100}. We note that as in the case (\ref{eq:92}a) we have
\begin{equation}
\label{eq:124}
  \| hA_n-{\mathfrak j}\frac{s^2}{2}\hat{i}_t\|_{L^\infty(B(z_0,8R^{\alpha -1  }))} \leq
  CR^{-3(1-\alpha)} \,,
\end{equation}
from which we obtain, as in the proof of \eqref{eq:105} that 
\begin{equation}
\label{eq:125}
\|[\LL_{R^3h},\eta_1]\RR_+^\infty\eta_1\| \leq C\, (R^{-\alpha} +
R^{\alpha -1  })\,  (1+\|(\A_+-\nu-ic{\mathfrak j}Rt_0)^{-1}\|)\,.
\end{equation}

Consider next the third term on the right hand side of
\eqref{eq:123}. We first note that for every $x\in\Omega\cap B(z_0,8R^{\alpha -1  })$
\begin{multline*}
  (\LL_{R^3h}-i\,R^3(\gamma -c{\mathfrak j}t_0)-\A_+)=
  -(g^{-2}-1)\frac{\partial^2}{\partial s^2}+g^{-3}(g_s+2i\,hR^3ga_s)\frac{\partial}{\partial s} +
  g^{-2}h^2R^6|a_s|^2 + \\
    g^{-1}\Big[g_t+2i\, gR^3\Big(ha_t-{\mathfrak j}\frac{s^2}{2}\Big)\Big]
     \Big(\frac{\partial}{\partial t}-i\, jR^3\frac{s^2}{2}\Big)  +
     R^6\Big(ha_t-{\mathfrak j}\frac{s^2}{2}\Big)^2 + i\, R^3[h\phi_n-\gamma-c{\mathfrak j}(t-t_0)]\,. 
\end{multline*}
By \eqref{eq:111} we have 
\begin{displaymath}
  \|\eta_1(g^{-2}-1)\| \leq \frac{C}{R^{1-\alpha}} \,.
\end{displaymath}
In view of \eqref{eq:113} we thus obtain 
\begin{equation}
\label{eq:126}
 \Big \|\eta_1(g^{-2}-1)\frac{\partial^2}{\partial s^2}
 \RR_+^\infty\eta_1\Big\|\leq \frac{C}{R^{1-2\alpha}}
 \big(1+\|(\A_+-\nu-i\, c{\mathfrak j}Rt_0)^{-1}\|\big)  \,. 
\end{equation}
As
\begin{displaymath}
  \|g_s\|\leq \frac{C}{R^{1-\alpha}} \,, 
\end{displaymath}
and since
\begin{displaymath}
  \Big\|\frac{\partial}{\partial s} \RR_+^\infty\Big\|  \leq
  \Big\|\Big(\nabla-iR^3{\mathfrak j}\frac{s^2}{2}\hat{i}_t\Big)
  (\A_+^R-R^2\nu-i\,c{\mathfrak j}R^3t_0)^{-1}\Big\|  \leq \frac{C}{R}(1+\|(\A_+^R-R^2\nu-i\,c{\mathfrak j}R^3t_0)^{-1}\|)\,,
\end{displaymath}
we readily obtain, in view of \eqref{eq:124}, that
\begin{multline}
  \label{eq:127}
  \Big \|\eta_1\Big[g^{-3}(g_s+2i\, hR^3ga_s)\frac{\partial}{\partial s} +
  g^{-2}h^2R^6|a_s|^2\Big] 
  (\A_+-\nu-i\,c{\mathfrak j}t_0)^{-1}\eta_1\Big\| \\ \leq \frac{C}{R^{1-3\alpha}}
  \big(1+\|(\A_+-\nu-i\,c{\mathfrak j}t_0)^{-1}\|\big)  \,.
\end{multline}
In a similar manner we show that
\begin{multline}
  \label{eq:128}
\Big \|\eta_1g^{-1}\Big\{\Big[g_t+2i g R^3\Big(ha_t-{\mathfrak j}\frac{s^2}{2}\Big)\Big]
     \Big(\frac{\partial}{\partial t}-iR^3{\mathfrak j}\frac{s^2}{2}\Big) \\ +
     R^6\Big(ha_t-{\mathfrak j}\frac{s^2}{2}\Big)^2\Big\} 
  (\A_+^R-R^2\nu-i\,c{\mathfrak j}R^3t_0)^{-1}\eta_1\Big\|\\ \leq \frac{C}{R^{1-3\alpha}}
  \big(1+\|(\A_+-\nu-i\,c{\mathfrak j}Rt_0)^{-1}\|\big)  \,.
\end{multline}
Finally, as 
\begin{displaymath}
  \|\eta_1[h\phi_n-\gamma-c{\mathfrak j}(t-t_0)]\| \leq CR^{-2(1-\alpha)} \,,
\end{displaymath}
we obtain, by combining the above together with \eqref{eq:126},
\eqref{eq:127}, and \eqref{eq:128} that
\begin{multline}
  \label{eq:129}
\|\eta_1\big(\LL_{R^3h}-iR^3(\gamma-c{\mathfrak j}t_0)-\A_+\big)(\A_+^R-R^2\nu-ic{\mathfrak j}Rt_0)^{-1}\eta_1\| \\ \leq C(R^{-\alpha}+R^{-(1-3\alpha)})\big(1+\|(\A_+-\nu-ic{\mathfrak j}Rt_0)^{-1}\|\big) \,.
\end{multline}

It is easy to show that \eqref{eq:109} holds true when $\chi_2$ is
replace by $\eta_2$, and hence in view of \eqref{eq:125} and
\eqref{eq:129} we obtain we can complete the proof of \eqref{eq:122}
in the same manner of \eqref{eq:100}.
\end{proof}

\begin{remark}
  In the previous proofs, we have substantially relied on Assumption
  (C) (formulated in \eqref{(C)}).  If we do not make this assumption,
  we could still attempt to approximate $\LL_{h,R} - \lambda_R$ by the
  Dirichlet realization in $\R^2_+$ of the following operator
  \begin{displaymath}
\A_+^\theta(\tilde{z}_0) =
-\Big(\nabla-i\Big[{\mathfrak j} \cos\theta\frac{s^2}{2}\hat{i}_t+{\mathfrak j} \sin\theta\frac{t^2}{2}\hat{i}_s\Big]\Big)^2
+ ic{\mathfrak j}[(t-t_0)\cos\theta+(s-s_0)\sin\theta] \,.
  \end{displaymath}
  In the above $\theta$ represents the angle between the curve $\Gamma $ and
  the boundary at the point of intersection.  However, since we are
  unaware of any study addressing the above operator, and even the
  definition of this operator in the half-plane appears to be
  non-trivial whenever $\theta \neq \frac \pi 2$\,, we defer the discussion of the
  more general case to a later stage.
\end{remark}
\subsection{Proof of Theorem \ref{thm1.3}}

  The theorem follows from Lemmas \ref{lem:5.1}, \ref{lem:5.2}, and
  \ref{lem:5.3}.  Let $\gamma\in\R$ and $\nu<\nu_m$ where $\nu_m$ is
  given by \eqref{eq:183c}. Then, for sufficiently
  large $R$, we must be in one of the cases listed in \eqref{eq:92}.
  Therefore, it is easy to show that a path $\gamma=\gamma(R)$, $\nu=\nu_0$, 
  exists in $[R_0,\infty)$ such that we remain in the same \eqref{eq:92}
  case with $\nu R^2 + i \gamma R^3\in\rho(\LL_{R^3h})$ for all $R\geq R_0$. In case
  (\ref{eq:92}a) we can thus apply \eqref{eq:100}, if either (\ref{eq:92}b)
  or (\ref{eq:92}d) is satisfied, we can apply \eqref{eq:122}.
  Finally, in cases where (\ref{eq:92}c) is satisfied, we can employ \eqref{eq:94aa}.
Let
\begin{displaymath}
   C_0(\nu)=\max
\Big(\sup_{z_0\in\Gamma}\|(\A(z_0)-\nu)^{-1}\|,\sup_{
  \begin{subarray}
\,    \gamma\in\R \\
    i=1,2
  \end{subarray}}\|(\A_+(z_i)-\nu-i\gamma)^{-1}\|\Big) \,.
\end{displaymath} 
Choosing
$\alpha=1/4$ in all of the  estimates of the two previous subsections yields that whenever
$\nu<\nu_m$ and $\lambda_R\in\rho(\LL_{R^3h})$ we have
\begin{displaymath}
  M(R) \leq C_0 + \frac{C}{R^{1/4}}\big[C_0 +1+M(R)^2\big] \,,
\end{displaymath}
where $M(R)$ is given by \eqref{eq:99}. It follows that there exist $C_1$
and $C_2$ such that either
\begin{equation}
  \label{eq:132}
 M(R) \leq C_0 \Big[1+ \frac{C_1}{R^{1/4}}\Big]\,,
\end{equation}
 or
\begin{equation}
  \label{eq:133}
M(R) \geq C_2R^{1/4}\,,
\end{equation}
for all sufficiently large $R$. Note that $[C_2R^{1/4},\infty)$ and
$(-\infty,C_0 \Big[1+ \frac{C_1}{R^{1/4}}\Big])$ are disjoint for
sufficiently large $R$ and $\nu<\nu_m$ .

Let $\nu\leq-1$. It is easy to show that $M(R) \leq1$ independently of $\gamma$ and
$R$ in that case. Consequently, for sufficiently large $R$, 
\eqref{eq:132} must be satisfied. We next increase $\nu$ gradually,
keeping $\gamma$ and $R$ fixed, thereby changing $M$ continuously. It
follows that as long as $\nu\leq\nu_m-\delta$ for any fixed positive $\delta$,
then for sufficiently large $R$, $M(R)$ must satisfy \eqref{eq:132}
since $\nu\mapsto M(\nu,R)$ maps $(-1,\nu_m-\delta)$ onto an interval in $\R$.

From the above discussion we can conclude that
\begin{displaymath}
  M(R,\nu,\gamma)\leq C_0(\nu) \Big[1+ \frac{2C_1(\nu)}{R^{1/4}}\Big]\,,
\end{displaymath}
for all $\nu<\nu_m$ and $R>R_0$.
Let
\begin{displaymath}
  M_0(R,\nu)=\sup_{\gamma\in\R}M(R,\gamma,\nu) \,.
\end{displaymath}
Since $C_1$ and $C_0$ are both independent of $\gamma$, it follows that
\begin{displaymath}
    M_0 (R) \leq C_0 \Big[1+ \frac{2C_1}{R^{1/4}}\Big]\,,
\end{displaymath}
for all $R>R_0$. The above inequality readily provides us with both
\eqref{eq:131} and a lower bound for $\mu_\infty$, i.e,
\begin{displaymath}
  \mu_\infty\geq\nu_m \,.
\end{displaymath}

To complete the proof of the theorem we need to show that
$$\limsup_{R\to\infty}\mu_R\leq\nu_m\,.
$$
If $\nu_m$ is infinite, then \eqref{eq:130} is readily proved. If
$\nu_m$ is finite, we now prove that for any sequence
$\{R_k\}_{k=1}^\infty$, such that $R_k\uparrow\infty$, there exists a corresponding
sequence of eigenvalues of $\LL_{R^3_kh}$, which we denote by
$\{\tilde{\lambda}_{R_k}\}_{k=1}^\infty$, satisfying
\begin{equation}\label{zzzgoal}
  R^2_k\,  \Re\tilde{\lambda}_{R_k}  \to \nu_m \,.
\end{equation}
For convenience of notation we denote $\A_+(z_i)$, at the point $z_i$ where
$\inf_{\lambda\in \sigma(\A_+(z_i))}\Re\lambda=\nu_m$, by $\A_+$ in the sequel.

We first claim that there exists $\lambda_{min}  \in\sigma(\A_+)$ which lies on the left
margin of $\sigma(\A_+)$, i.e., $\Re\lambda_{min}=\nu_m$. We prove this claim by
using the same techniques as in the proof of Lemma 7.2 in
\cite{aletal13}, to show that there exists $C>0$,  such that $\nu + i
\gamma\in\rho(\A_+)$ for $\nu\in [0, \nu_m +1]$ and $|\gamma| \geq C$. Hence the infimum
defining $\nu_m$ is attained by some $\lambda_{min}$ whose real part is
$\nu_m$. 

Let then $\lambda_{min}\in \sigma(\A_+)$ satisfy $\Re\lambda_{min}=\nu_m$.  We choose $\gamma$ such that $z_0\in\partial\Omega$, i.e., so that $F(z_1)-\gamma/h=0$,
and hence $d_0=0$ in \eqref{eq:185}.   Since \eqref{eq:123}
holds for any $\lambda_R\in\rho(\LL_{R^3h})$ for which $\nu<\nu_m$, we choose
$\nu=0$, and hence $\lambda_R=i\gamma R^3$.\\
Hence, we have
\begin{displaymath}
   (\LL_{R^3h}-\lambda_R)\RR_+ - I = [\LL_{R^3h},\eta_1](\A_+^R)^{-1}\eta_1 +\\
  \eta_1\big(\LL_{R^3h}-i\gamma R^3-\A_+^R\big)(\A_+^R)^{-1}\eta_1 + [\LL_{R^3h},\eta_2]
  (\LL_{R^3h}-\lambda_R)^{-1}\eta_2 \,.
\end{displaymath}
By \eqref{eq:125} (recalling that $\RR_+^\infty=(\A_+^R)^{-1}$ since
$t_0=\nu=0$), 
\eqref{eq:129}, \eqref{eq:94}, and \eqref{eq:131} we obtain that 
\begin{displaymath}
  R^2\|\RR_+ - (\LL_{R^3h}-\lambda_R)^{-1}\| \leq \frac{C}{R^{1/4}} \,,
\end{displaymath}
and hence, using \eqref{eq:94} once again yields
\begin{displaymath}
   R^2\|\eta_1(\A_+^R)^{-1}\eta_1  - R^2(\LL_{R^3h}-\lambda_R)^{-1}\| \leq \frac{C}{R^{1/4}} \,.
\end{displaymath}
Equivalently we may state that
\begin{displaymath}
  \|\eta_1^R\A_+^{-1}\eta_1^R  - R^2U_R(\LL_{R^3h}-\lambda_R)^{-1}U_R^{-1}\|=\|\eta_1^R\A_+^{-1}\eta_1^R  - (\LL_{h}^R-R^{-2}\lambda_R)^{-1}\| \leq
  \frac{C}{R^{1/4}} \,,
\end{displaymath}
where $\eta_1^R=U_R\eta_1$. \\
We can now define
$\RR^R_+:L^2(\R^2_+)\to H^2_{mag}(\R^2_+)$ by
\begin{displaymath}
  \RR^R_+=\eta_1^R\A_+^{-1}\eta_1^R \,.
\end{displaymath}
It can be easily verified that $\RR^R_+\to\A_+^{-1}$ in
$\LL\big(L^2(\R^2_+)\big)$. Here we use that the form domain of the
operator $\A_+$ is contained in a suitable weighted space $L^2(\mathbb
R^2_+,\rho)$ with $\rho \to +\infty$ as $|x|\to\infty$ (see the proof of Proposition
2.4 in \cite{aletal10}). Hence, by Section IV, \S3.5 in \cite{ka80}, it
follows that, for any sequence $R_k\uparrow\infty$, there exists $\Lambda_{R_k}\in\sigma
((\mathcal L_h^{R_k} - R_k^{-2} \lambda_{R_k})^{-1})$ such that
$\Lambda_{R_k}\to\lambda_{min}^{-1}$, thereby verifying (\ref{zzzgoal}) for
\begin{displaymath}
   \tilde{\lambda}_{R_k}=\frac{1}{R_k^2\Lambda_{R_k}} + \frac{1}{R_k^4\lambda_{R_k}}\,.
\end{displaymath}

\begin{remark}
  By the proof of Lemma 7.2 in \cite{aletal13} it follows that
  \begin{equation}
\label{eq:54}
   \| (\A(1,c)-\nu)^{-1}\| = \lim_{\gamma\to+\infty}\|(\A_+(1,c)-\nu-i\gamma)^{-1}\| \,.
  \end{equation}
  Let ${\mathfrak j}_m $ and ${\mathfrak j}_+$ be given by
  \eqref{eq:184}.  It follows from \eqref{eq:54} that, if $ {\mathfrak
    j}_m = {\mathfrak j}_+$ we can drop
  $\sup_{z_0\in\Gamma}\|(\A(z_0)-\nu)^{-1}\|$ from the right-hand-side of
  \eqref{eq:131}.
\end{remark}

\newpage

\appendix

\section{Scalar Regularity in non-smooth domains}\label{sec:a}

In this section we provide some elliptic estimates, valid for domains
with right-angled corners. For simplicity we assume property (R1) for
the domain. We note that one can extend these estimates to domains
composed of $C^2$ curves (or curved polygons).  For the case of
homogeneous boundary condition we prove, for the so-called variational
(or weak) solution obtained by the Lax-Milgram Theorem, regularity in
$W^{2,p}(\Omega)$ for any $p\geq2$. Inhomogeneous boundary conditions are
converted into homogeneous ones by subtracting an appropriate function
from the variational solution.

While most of the estimates we present here were first
obtained in the pioneering works of  Kondratiev
 \cite{kon67} and Mazya-Plamenevskii \cite{mapla}, we prefer to
 refer the reader to \cite{gr85}, \cite{gr92} or Chapter~1 in
 \cite{gira79}, where the results we use here can be obtained more
 easily.  We finally note that while many of the results we cite are
 stated for polygons, they are equally valid for domains with property
 (R1). The regularity away from the corners follows from standard
 elliptic estimates, and the regularity near the corners can be
 obtained using the arguments applied in the case of polygons.

\begin{proposition}\label{propoA1}~\\
 Let $ \Omega$ satisfy property (R1) and  $p\geq 2$, 
   Let $f\in L^p(\Omega)$,  $g_{i}\in W^{2-\frac 1p,p} (\pa\Omega_{i })$ $g_{c}\in W^{2-\frac 1p,p} (\pa\Omega_{c })$ such  that
 \begin{equation}\label{eq:161D}
    g_c (S_\ell) =  g_i (S_\ell) \quad l=1,\ldots,4
 \end{equation}
   at any  of the four corners $\{S_l\}_{l=1}^4$.\\
  Let  $u\in H^1(\Omega)$
  denote the  variational  solution for the following problem
  \begin{equation}
    \label{eq:162D}
    \begin{cases}
      -\Delta u = f & \text{in } \Omega \\
      u =g_i & \text{in } \partial\Omega_c \\
      u = g_c & \text{in } \partial\Omega_i\,.
    \end{cases}
  \end{equation}
  Then, $u \in W^{2,p}(\Omega)$ and   there exists $C(p,\Omega)$ such that
\begin{equation}
  \label{eq:173D}
\|u\|_{2,p} \leq C(p,\Omega) \left(  \|f\|_p + \|g_i\|_{W^{2- \frac 1 p}(\pa
    \Omega_i)} +  \|g_c\|_{W^{2- \frac 1 p} (\pa \Omega_i)}\right)  \,.
\end{equation}
Furthermore, if $f\in W^{1,q}(\Omega)$ for some $1<q<2$ and $g_i\equiv g_c\equiv0$, then $u \in W^{3,q}(\Omega)$
\begin{equation}
\label{eq:136}
\|u\|_{3,q} \leq C\, \|f\|_{1,q}   \,.
\end{equation}
\end{proposition}
\begin{proof}~\\
The proof of \eqref{eq:173D} follows immediately from Corollary
4.4.3.8 \cite{gr85}. The proof of \eqref{eq:136} follows from  Theorem
5.1.2.3 there.
\end{proof}
\begin{proposition}\label{propoA2}~\\
 Let $\Omega$ satisfy property (R1) and  $p\geq 2$.
   Let $f\in L^p(\Omega)$, $g_{i}\in W^{1-\frac 1p,p} (\pa\Omega_{i })$ $g_{c}\in W^{1-\frac 1p,p} (\pa\Omega_{c })$ such  that
   $$
   \int_\Omega f(x) dx + \int_{\pa \Omega_c} g_c \, ds + \int_{\pa \Omega_i} g_i \, ds=0\,.
   $$
  Let  $u\in H^1(\Omega)$
  denote the weak solution for the following problem
  \begin{equation}
    \label{eq:162N}
    \begin{cases}
      -\Delta u = f & \text{in } \Omega \\
      \int_\Omega   u\, dx =0 & \\
      \frac{\partial u}{\partial\nu} =g_i & \text{in } \partial\Omega_c \\
      \frac{\partial u}{\partial\nu}= g_c & \text{in } \partial\Omega_i\,.
    \end{cases}
  \end{equation}
  Then $u \in W^{2,p}(\Omega)$ and  there exists $C(p,\Omega)$ such that
\begin{equation}
  \label{eq:173N}
\|u\|_{2,p} \leq C \left(  \|f\|_p + \|g_i\|_{W^{1- \frac 1 p}(\pa
    \Omega_i)} +  \|g_c\|_{W^{1- \frac 1 p} (\pa \Omega_i)}\right) \,.
\end{equation}
Furthermore, if $f\in W^{1,q}$ for some $1<q<\infty$ and $g_i\equiv g_c\equiv0$, then $u\in W^{3,q}(\Omega)$ and 
\begin{equation}
\label{eq:137}
\|u\|_{3,q} \leq C\, \|f\|_{1,q}   \,.
\end{equation}

\end{proposition}
\begin{proof}~\\
The proof of \eqref{eq:173N} follows immediately from Corollary
4.4.3.8 in \cite{gr85}.  The proof of \eqref{eq:137} follows from  Theorem
5.1.2.3 there.
  \end{proof}

\begin{proposition}\label{propoA3}~\\
  Let $ \Omega$ satisfy property (R1)  and $ p\geq 2$. Let  $f\in L^p(\Omega)$ and  $u\in H^1(\Omega)$
  denote the variational  solution for the following problem
  \begin{equation}
    \label{eq:162}
    \begin{cases}
     - \Delta u = f & \text{in } \Omega \\
      u= 0& \text{in } \partial\Omega_c \\
      \frac{\partial u}{\partial\nu}=0 & \text{in } \partial\Omega_i\,.
    \end{cases}
  \end{equation}
 Then  $u\in W^{2,p}(\Omega)$  and  there exists $C(p,\Omega)$ such that, for any $f\in L^p(\Omega)$,
\begin{equation}
  \label{eq:173}
\|u\|_{2,p} \leq C \|f\|_p \,.
\end{equation}
\end{proposition}
\begin{proof}
This is an immediate consequence of Theorem 2.4.3  in \cite{gr92} (in
the case $p=2$) and Theorem 4.4.3.7 in \cite{gr85} for $p>2$.  
For $p=2$, the reader is referred to p.~55 in \cite{gr92} to verify
that indeed mixed Dirichlet-Neumann conditions near a right-angled
corners do not produce any singularities.  
\end{proof}
\begin{remark}\label{DN}~\\
  It is necessary to address, in addition, \eqref{eq:162} with
  non-homogeneous boundary conditions. In particular we consider:
$$
 \frac{\partial u}{\partial\nu}= B\, \mbox{ on } \pa \Omega_i\,.
$$
Assuming that $B\in H^\frac 12(\pa \Omega)$, an $H^2(\Omega)$ regularity of $u$
does not follow as in the homogeneous case unless $B$ vanishes at the
corners. The best regularity we can obtain in this case is:
$$
u\in \cup_{q<2}W^{2,q}(\Omega)\,.
$$
We refer the reader this to Corollary 4.4.3.8 in \cite{gr85} for the
proof of this result.
\end{remark}

\section{Vector Regularity in Non-Smooth domains}
\label{sec:b}
\subsection{Decomposition of vector fields}
We begin by presenting a well-known decomposition of vector fields in
$L^2(\Omega;\mathbb R^2)$ using their $\curl$ and their $\Div$.
\begin{proposition}\label{decprop}~\\
Any vector $U \in L^2(\Omega;\mathbb R^2)$ can be  uniquely written as the sum 
\begin{equation}
  \label{eq:148}
U = V + W\,,
\end{equation}
where $$V\in H^0_0(\curl, \Omega)=\{\widehat V \in L^2(\Omega,\mathbb R^2)\,,\, \curl \widehat V=0\}$$ and $$W\in \mathcal H^0_d:= \{\widehat W \in L^2(\Omega,\mathbb R^2)\,,\, \Div \widehat W=0\mbox{ and }
 \widehat W \cdot \nu =0\mbox{ on } \pa \Omega \}\,.
 $$
The maps $U\mapsto \pi_1 U=V$  and $U \mapsto \pi_2 U= W$ in \eqref{eq:148}  are mutually orthogonal projections 
in $L^2(\Omega;\mathbb R^2)$.

Moreover $\pi_1$ and $\pi_2$ are continuous from $H^1(\Omega,\mathbb R^2)$ into  $H^1(\Omega,\mathbb
R^2)$ 
 and there exist $q$ and $r$ such
that 
\begin{equation}\label{cd1}
V= \nabla q\,,\,  W=\nabla^{\perp} r\,,
\end{equation}
where $q\in H^2(\Omega)$  is the variational solution (orthogonal to the constant) of the Neumann problem
\begin{equation}\label{cd2}
\Delta q = \Div U\,,\, \pa_n q= U\cdot \nu \,,\, \int_\Omega q \, dx = 0\,,
\end{equation}
and $r \in H_0^1(\Omega)\cap H^2(\Omega)$ is the variational  solution of the Dirichlet problem
\begin{equation}
\label{eq:83}
  \begin{cases}
    \Delta r = \curl  U & \text{in } \Omega \\
    r=0 & \text{on } \partial\Omega
  \end{cases}
\end{equation}
\end{proposition}
\begin{proof}
  This is an immediate consequence of Remarks 3.3 and Theorem 3.4 in \cite{gira79}.
\end{proof}

\begin{remark}
Note that by applying Proposition \ref{propoA1} to \eqref{eq:83} and
Proposition \ref{propoA2} to \eqref{cd2} we obtain that, if $U\cdot \nu =0$ on $\pa \Omega$,
 for all $p\geq2$
there exist $C(p,\Omega)$ such that
\begin{equation}
  \|r\|_{2,p} \leq C\, \|\curl  U\|_p \quad ; \quad  \|q\|_{2,p} \leq C\,
  \|\Div  U\|_p \,.
\end{equation}
By \eqref{cd1} we then obtain that, if $U\cdot \nu =0$ on $\pa \Omega\,$,
\begin{equation}
\label{eq:135}
  \|U\|_{1,p} \leq C(p,\Omega)\left(   \| V\|_{1,p} + \|W\|_{1,p} \right) \leq C_1(p,\Omega) \left(\|\curl U\|_p + \|\Div U\|_p\right)  \,,
\end{equation}
holds for any $p\geq 2\,$.
\end{remark}

\subsection{A new Laplacian}
  Set  $$\mathcal L^{(1)}=  \nabla^{\perp} \curl  -  \grad \Div$$ which is
  defined as the linear operator associated with the form
  \begin{equation}
\label{eq:188}
  (V,W) \mapsto \int_\Omega \curl V\cdot \curl W\,dx  + \int_\Omega \Div V \cdot \Div W\, dx\,,
  \end{equation}
where $V\in H^1(\Omega;\mathbb R^2)\,,$ such that $V\cdot \nu =0$ on $\partial\Omega$. It
can be easily verified that the domain of $\LL^{(1)}$ is given by
  \begin{equation}\label{domaineL2}
 D(\mathcal L^{(1)})= \{ V\in H^2(\Omega;\mathbb R^2)\,,\, V\cdot \nu_{/\pa \Omega} =0\,,\,\curl V _{/\pa \Omega}= 0\,\}.
  \end{equation}
Note, that, for any smooth function $V=(u,v)$,  we have $\LL^{(1)} V=(-\Delta u, -\Delta v)$
at every point in $\Omega$.

\begin{proposition}\label{propoA4}~\\
  Let $ \Omega$ satisfy property (R1)  and $ p\geq 2$. Let
  $F\in L^p(\Omega,\mathbb R^2)$, $V\in H^1(\Omega,\mathbb R^2)$ denote the
  variational  solution for the above problem, i. e.
  \begin{equation}\label{dc6}
  \begin{array}{ll}
  \nabla^{\perp} \curl V  -  \grad \Div V= F\,& \mbox{ in } \Omega\,,\\
  V\cdot \nu =0& \mbox{on }\pa \Omega\,,\\
\curl V = 0 & \mbox{ on  } \pa \Omega\,.
  \end{array}
  \end{equation}
Then  $V \in W^{2,p}(\Omega)$  and  there exists $C(p,\Omega)$ such that, for
any $F \in L^p(\Omega)$, 
\begin{equation}
  \label{eq:173r}
\|V \|_{2,p} \leq C \| F\|_p  \,.
\end{equation}
\end{proposition}
\begin{proof}
Due to Assumption (R1), it is sufficient to estimate $\|V\|_{2,p}$ in
a neighborhood of one of the right angle corners. Set then a coordinate
system whose origin is located at the corner and its axes coincide
with the boundary in some neighborhoods of the corner. We set
$V=(u,v)$ and $F=(f,g)$. Without loss of generality we can represent the problem
satisfied by $u$ near the corner by
\begin{displaymath}
  \begin{cases}
     - \Delta u = f & \text{in } B(0,r)\cap\Omega \\
      u(x_1,0)= 0&\mbox{ for } 0<  x_1<r \\
      \frac{\partial u}{\partial x_1}(0,x_2)=0 &\mbox{ for } 0<  x_2<r\,,
  \end{cases}
\end{displaymath}
and
\begin{displaymath}
  \begin{cases}
     - \Delta v = g& \text{in } B(0,r)\cap\Omega \\
      \frac{\partial v}{\pa x_2} (x_1,0)= 0&\mbox{for } 0 < x_1<r \\
      v(0,x_2)=0 & \mbox{ for } 0< x_2<r\,.
  \end{cases}
\end{displaymath}

Hence, the system is decoupled near each corner into a
Dirichlet-Neumann problem, for $u$ (and a Neumann-Dirichlet problem
for $v$). Multiplying $u$ and $v$ by a cutoff function supported in
$B(0,r)$, we can transform the problem in $\Omega$ to a problem in a polygon.
We can then apply Proposition \ref{propoA3} to obtain that
\begin{displaymath}
  \|V \|_{W^{2,p}(\Omega\cap B(0,r/2)} \leq C (\| F\|_p +   \|V \|_{L^p(\Omega\cap B(0,r/2))})\,.
\end{displaymath}
With the aid of standard elliptic estimates away from the corners we
can glue together the estimates and establish that 
\begin{displaymath}
\|  V \|_{2,p} \leq C (\| F\|_p+ \|V\|_p)\,.  
\end{displaymath}
We complete the proof by substituting \eqref{eq:135} into the above, in
conjunction with the inequality
\begin{displaymath}
  \|\curl U\|_p + \|\Div U\|_p \leq \|F\|_p\,,
\end{displaymath}
which easily follows from \eqref{eq:188} and \eqref{dc6}.
\end{proof}
\begin{remark}\label{regnonhom}~\\
  In Appendix \ref{appendixD} we meet a non-homogeneous version of
  \eqref{dc6}. In particular we consider the boundary condition 
$$
 \curl V = B\,, \mbox{ on } \pa \Omega \,.
$$
Assuming that $B\in H^\frac 12(\pa \Omega)$, an $H^2(\Omega)$ regularity of $u$
does not follow as in the homogeneous case unless $B$ vanishes at the
corners. The best regularity we can obtain in this case is:
$$
u\in \cup_{q<2}W^{2,q}(\Omega)\,.
$$
Once we have observed the above decoupling near each corner, the proof
is a direct consequence of Remark \ref{DN}.
\end{remark}

\section{Time Dependent Regularity}
\label{sec:A}

We derive here some regularity results for solutions of parabolic
equations in domains satisfying property (R1), relying heavily on the
elliptic estimates recalled in the previous subsections. These
estimates are exploited throughout the work and in the next appendix,
where global existence and uniqueness of solutions for (TDGL2) is
established.  We derive all estimates for a general abstract setting.
When we refer to them throughout the work we cite the specific example
in which they are applied.

Let $(\mathcal V,\mathcal H,\mathcal V')$ denote a triplet of Hilbert
spaces such that $(\mathcal V\subset \mathcal H\subset \mathcal V')$ (where
$\Vg^\prime$ denotes the dual space of $\Vg$) with continuous injection such
that $\mathcal V$ is dense in $\mathcal H$ with compact injection. Let
$a:\Vg\times\Vg\to\R$ denote a coercive, continuous, and symmetric bilinear
map.  For some $\alpha >0$ we thus have
$$
\alpha\,  \|X\|^2_{\Vg} \leq a(X,X)\,,\, \forall X \in \mathcal V\,.
$$
By Lax-Milgram theorem we can associate two bijective operators
$A_{\mathcal V}$ and $A_{\mathcal V^\prime}$ from $\mathcal V$ onto $\mathcal V'$
satisfying
 \begin{displaymath}
   a(u,v)= \langle A_{\mathcal V}u\,,\,v\rangle_{\mathcal V}= \langle A_{\mathcal
   V^\prime}u\,,\,v\rangle_{\mathcal V^\prime}\,,\, \forall u \in \mathcal V\,,\, \forall v\in
 \mathcal V\,,
 \end{displaymath}
 and $\mathcal A$ which is a self-adjoint strictly positive unbounded
 operator in $\mathcal H$, with compact resolvent, satisfying 
 $$
 a(u,v) = \langle  \mathcal A u\,,\, v\rangle_{\mathcal H}\,,\, \forall u\in D(\mathcal A)\,,\, \forall v\in \mathcal V\,.
 $$
 
Use of the above general setting is being made throughout this work in
the following particular cases
\begin{enumerate}
\item The Dirichlet Laplacian $-\Delta^D$ with $\mathcal H=L^2(\Omega)$, $
  \mathcal V =H_0^1(\Omega)$,  and $D(\mathcal A)= H^2(\Omega)\cap H_0^1(\Omega)$; 
\item The Neumann Laplacian $-\Delta^N$ with $\mathcal H=\{ u\in
  L^2(\Omega)\,,\,\int_\Omega u\, dx =0\}$, $\mathcal V =H^1(\Omega)\cap \mathcal H$,
  and $D(\mathcal A)= \{u\in H^2(\Omega)\,,\int_\Omega u\, dx=0\,,\, 
 \pa_\nu u =0 \mbox{ on each regular piece of } \pa \Omega \}\,$;
 \item The Dirichlet-Neumann Laplacian $-\Delta^{DN}$ with $\mathcal
   H=L^2(\Omega)$,  $ \mathcal V =H_0^{1,\pa \Omega_c}(\Omega)$, and $D(\mathcal
   A)= \{u\in H^2(\Omega)\,,\, u_{/\pa \Omega_c}=0\,,\,\pa_\nu u_{/\pa \Omega_i}=0\}$
   \,; 
 \item The operator $\mathcal L^{(1)}$ with $\mathcal H=L^2(\Omega;\mathbb
   R^2)$, $\mathcal V =\{V \in H^1(\Omega;\mathbb R^2)\,,\, V\cdot \nu_{/\pa \Omega}
   =0 \}\,,$ and $D(\mathcal A) =\{ V \in H^2(\Omega;\mathbb R^2)\,,\, V \cdot
   \nu_{/\pa \Omega} =0\,,\, \curl V_{/\pa \Omega} =0\}\,$;
 \end{enumerate}
In each of the above cases the domain is deduced with the aid of the
relevant regularity result in the previous appendices. Furthermore, there
exists in each case $C(\Omega,\A)$ such that 
\begin{equation}
  \label{eq:84}
\|X\|_{D(\A)}\leq C\,  \|\A X\|_\Hg \,.
\end{equation}

We now state and prove the general statement.\\
 \begin{theorem}\label{abstractproposition}
   Let $T>0$, $F\in L^2(0,T;\Hg)$, and $X_0\in\Hg$. Then, there exists
   a unique $X(t,\cdot) \in L^2(0,T;\mathcal V)$, which serves as weak solution
   for
 \begin{equation}
\label{eq:146abstract}
    \begin{cases}
      X_t-\mathcal A X= F\,& \text{in } L^2(0,T;\mathcal V')\,, \\
    X(0,\cdot) = X_0\, & \mbox{ in }  \mathcal H\,.
    \end{cases}
  \end{equation}
Additionally, $X'\in L^2(0,T; \mathcal V')$ and $X\in L^\infty(0,T;\mathcal H)$.
Furthermore, for any $0<t_0<T$ we have:
    \begin{equation}\label{eq:148abstracta}
  \|X \|_{L^\infty(t_0,T; \mathcal V)} \leq C(t_0)
  \big[\| F \|_{L^2(0,T;\mathcal H)}+\|X _0\|_{\mathcal H}\big]\,.
\end{equation}
Moreover,  if in addition $X_0\in \mathcal V$ ,
then  $X\in L^2(0,T;D(\mathcal A))  \cap H^1(0,T;\mathcal H)$ and  there
exists $C>0$, independent 
   of $T$, such that
  \begin{equation}
\label{eq:147abstract}
    \|X \|_{L^2(0,T;D(\mathcal A))}+ \|X'\|_{L^2(0,T;\mathcal H)}  + \|X\|_{L^\infty(0,T; \mathcal V)} \leq C 
    \big[\| F\|_{L^2(0,T;\mathcal H)}+\|X_0\|_{\mathcal V} \big]\,,
    \end{equation}
    Finally, if $F \in L^2(0,T;\mathcal V )$ and $X_0\in \mathcal H$,
    then, for any $0<t_0<T$ we have
\begin{equation}\label{eq:148abstract}
  \|X \|_{L^\infty(t_0,T; D(\mathcal A))} \leq C(t_0)
  \big[\| F \|_{L^2(0,T;\mathcal V)}+\|X _0\|_{\mathcal H}\big]\,.
\end{equation}
\end{theorem}

\begin{proof}~\\
  We begin by concluding the existence and uniqueness of a weak
  solution for \eqref{eq:146abstract} from Theorem 3.2.2 in
  \cite{he81}. The remainder
  is devoted to the proof of the above estimates.\\
  Suppose first that $F\in L^2(0,T; \mathcal H)$. Denote by
  $\{\lambda_n\}_{n=1}^\infty$ the increasing sequence of the eigenvalues of
  $A$ counted with multiplicity and let $\{u_n\}_{n=1}^\infty$  denote the corresponding
  orthonormal basis of eigenfunctions. Let further $a_n=\langle
  X\,,\,u_n\rangle_\mathcal H $ and $b_n=\langle F \,,\,u_n\rangle_\Hg$. Taking the
  product of \eqref{eq:146abstract} by $u_n$ in $\Hg$ yields
  \begin{equation}
\label{eq:136abstract}
    \frac{da_n}{dt}+\lambda_na_n = b_n \,.
  \end{equation}
Multiplying \eqref{eq:136abstract} by $\frac{da_n}{dt}$, using
Cauchy-Schwarz,  and integrating between $0$ 
and $t_1$ (with $t_1\in (0,T)$) yields
\begin{equation}\label{eq:136abstracta}
  \Big\|\frac{da_n}{dt}\Big\|^2_{L^2(0,t_1)} +\lambda_n\big[|a_n(t_1)|^2
  - |a_n(0)|^2\big] \leq
   \|b_n\|_{L^2(0,t_1)}^2 \, .
\end{equation}
Summing up over $n$ gives
\begin{equation}
\label{eq:134}
  \|X^\prime\|_{L^2(0,T;\Hg)}^2 +  \|X\|_{L^\infty(0,T;\Vg)}^2 \leq \| F\|_{L^2(0,T;\mathcal H)}+\|X_0\|_{\mathcal V} \,.
\end{equation}
By \eqref{eq:146abstract} and \eqref{eq:84} we have 
\begin{displaymath}
  \|X\|_{L^2(0,T;D(\A))} \leq C\|AX\|_{L^2(0,T;\Hg)} \leq  C\big[\|
  F\|_{L^2(0,T;\mathcal H)}+  \|X^\prime\|_{L^2(0,T;\Hg)} \big] \,. 
\end{displaymath}
Combining the above with \eqref{eq:134} yields \eqref{eq:147abstract}.

Suppose now that $F \in L^2(0,T; \mathcal V )$.  Multiplying by
$\lambda_n+1$ in \eqref{eq:146abstract} and summing up over $n$ 
yields for another constant $C>0$:
\begin{equation}
\label{eq:137abstract}
  \|X_t\|_{L^2(0,T;\mathcal V)}^2 + \|AX\|_{L^\infty(0,T;\Hg)}^2  \leq  C \, \big[\|F\|_{L^2(0,T;\mathcal V)}^2 +
  \|X_0\|_{D(\mathcal A)}^2  \big] \,.
\end{equation}
With the aid of \eqref{eq:84} we then obtain \eqref{eq:148abstract}.

We next return to \eqref{eq:136abstracta} to obtain
\begin{equation}
\label{eq:139}
 a_n(t) = a_n(0)e^{-\lambda_nt} + \int_0^t e^{-\lambda_n(t-\tau)}b_n(\tau)\,d\tau \,,
\end{equation}
from which we readily obtain that
\begin{equation}
\label{eq:140}
  |a_n(t)|^2 \leq 2 |a_n(0)|^2e^{-2\lambda_nt} + \frac{2}{\lambda_n}\|b_n\|_2^2 \,.
\end{equation}
Let
\begin{displaymath}
  C_0(t,\Omega) =\frac{1}{et} \geq  \sup_{n\in\N} \lambda_ne^{-\lambda_nt}
\end{displaymath}
and 
\begin{displaymath}
  C_1(t,\Omega) = \frac{4}{t^2 e^2} \geq \sup_{n\in\N} \lambda_n^2e^{-\lambda_nt} \,.
\end{displaymath}
Obviously, $C_i(t,\Omega)$ are decreasing functions of $t$, and
hence
\begin{displaymath}
  C_i(t,\Omega)\leq C_i(t_0,\Omega)\,, \mbox{ for } i=0,1\,,
\end{displaymath}
for all $t \geq t_0$. \\
Upon multiplying \eqref{eq:140} by $\lambda_n$, we
sum over $n$ to obtain, for $t\in [t_0,T)$, 
\begin{displaymath}
  \|  u(t,\cdot)\|_{\mathcal V} ^2 \leq C \, 
  \|F (t,\cdot)\|_{L^2(0,T;\mathcal H)}^2+ 2C_0(t_0)^2 \, \|u_0\|_\mathcal H^2 \,,
\end{displaymath}
yielding \eqref{eq:148abstracta}.
Multiplying \eqref{eq:140} by $\lambda_n^2$ and
summing over $n$ yields, for $t\in [t_0,T)$, 
\begin{displaymath}
  \|\mathcal A  u(t,\cdot)\|_2^2 \leq C \, 
  \|F (t,\cdot)\|_{L^2(0,T;\mathcal V)}^2+ 2C_1(t_0)^2 \, \|u_0\|_\mathcal H^2 \,,
\end{displaymath}
which together with \eqref{eq:84} yields \eqref{eq:148abstract}.
\end{proof}

\section{Global existence and uniqueness}
\label{appendixD}
We adopt here the gauge 
\begin{equation}\label{gaugefl}
  \phi + c\, \Div A = 0 \quad ; \quad A\cdot\nu|_{\partial\Omega}=0 \,,
\end{equation}
and assume that $B\in H^\frac 12(\pa \Omega)$ and that condition (R1) is satisfied.
The normal field $A_{n,d}$ is defined by 
\begin{equation}
\label{eq:72}
    \begin{cases}
     \curl^2A_{n,d} - \nabla \Div A_{n,d} = 0 & \text{in } \Omega\,, \\
     \curl A_{n,d} = B   & \text{on } \partial\Omega\,,\\
     A_{n,d} \cdot \nu = 0 & \text{on } \partial\Omega\,.
  \end{cases}
\end{equation}
Using Remark \ref{regnonhom}, we note that unless $B$ vanishes at the
corners (an atypical case) $A_{n,d}\not \in H^2(\Omega;\mathbb R^2)$. One
can guarantee, however, that
\begin{equation}\label{regand}
A_{n,d}\in \cup_{q<2} W^{2,q}(\Omega,\mathbb R^2)\,.
\end{equation}
\begin{remark}~\\
From \eqref{eq:8} and \eqref{eq:8a}, we get that $(A_n-A_{n,d})$ is in the  domain of $\mathcal L^{(1)}$ and satisfies 
$$
\mathcal L^{(1)} (A_n-A_{n,d}) = -\frac 1c  \nabla \phi_n\,.
$$
But $\nabla \phi_n\in L^p(\Omega,\mathbb R^2)$ for all $p\geq 2$ by
\eqref{regphin}, and hence, by Proposition  \ref{propoA4}, we obtain 
\begin{equation}\label{eq:D4}
A_n-A_{n,d} \in W^{2,p}(\Omega)\,,\, \forall p \geq 2\,.
\end{equation}
In particular, we have that
\begin{equation}\label{regdiv}
\Div A_{n,d} \in W^{1,p} (\Omega)\,,\, \forall p\geq 2\,.
\end{equation}
\end{remark}

Next we set, 
\begin{equation}
\label{eq:186}
\widehat A_1=A-hA_{n,d}
\end{equation}
 to obtain, in the
above gauge, the system
\begin{subequations}  
\label{eq:1v3}
\begin{alignat}{2}
& \frac{\partial \psi}{\partial t} - \Delta\psi = F_1 & \quad \text{ in } \R_+\times\Omega\, ,\\
 & \frac{\partial \widehat A_1}{\partial t}  + c\, \LL^{(1)}\widehat A_1  =  F_2 & \quad \text{ in }  \R_+\times\Omega\,,\\
  &\psi=0 \,,&\quad \text{ on }  \R_+\times\partial\Omega_c\,, \\
 &\frac{\partial\psi}{\partial\nu}=0 \,,& \quad \text{ on }  \R_+\times\partial\Omega_i\,,\\
 &\curl \widehat A_1(t,x)  = 0 \,, & \text{ on } \R_+\,, \\
&\psi(0,x)=\psi_0(x) \,, & \quad \text{ in } \Omega\,, \\ 
&\widehat A_1(0,x)=A_0(x)-h A_{n,d}(x)\,, & \quad \text{ in } \Omega \,.
\end{alignat}
\end{subequations}
In the above,
\begin{align}
  &F_1 = i\, \Big(\frac 1 c -2\Big) \psi\, \Div A  - 2 i A\cdot\nabla_A\, \psi  +|A|^2\, \psi+ \psi 
  \left( 1 - |\psi|^{2} \right) \\
& F_2 = \frac{1}{\sigma}\Im(\bar\psi\, \nabla_A\psi)  \,.
\end{align}

We can now make the following statement:
\begin{theorem}
\label{thm:A.1}~\\
Suppose that for some $1/2<\alpha<1$, $(\widehat \psi_0,\widehat A_0)\in W^{1+\alpha,2}(\Omega,\C)\times
W^{1+\alpha,2}(\Omega,\R^2)$ satisfy\break  $\widehat A_0\cdot \nu =0$ on $\pa \Omega$
and $\|\psi_0\|_\infty \leq 1$. \\
Then, there exists a weak solution $(\psi_d,A_d)$ of
\eqref{eq:1v3} and \eqref{gaugefl}, such that: $$\psi_d \in C([0,+\infty);W^{1+\alpha,2}(\Omega,\C))\cap
H_{loc}^1([0,+\infty);L^2(\Omega,\C))$$ and
 $$ A_d  \in C([0,+\infty);W^{1+\alpha,2}(\Omega,\C))\cap
H_{loc}^1([0,+\infty) ;L^2(\Omega,\R^2))\,.$$
Moreover $\psi$ and $A$ satisfy:
 $$
  \|\psi_d(t,\cdot )\|_\infty\leq1\,,\, \forall t >0\,,
  $$
  \begin{displaymath}
    \psi_d \in L^2_{loc} ([0,+\infty), H^2(\Omega,\C)) \,,
  \end{displaymath}
  $$
  \widehat A_1 \in L^2_{loc} ([0,+\infty), H^2(\Omega,\mathbb R^2))\,,
  $$
and 
$$\Div A_d  \in L^2_{loc}([0,+\infty);H^1(\Omega))\,.$$
\end{theorem}
\begin{proof}
  We follow the same steps as in the proof of Theorem 1 in
  \cite{feetal98}. \\
{\em Local existence and uniqueness:} \\
Set
\begin{displaymath}
  \A =
  \begin{bmatrix}
    -\Delta_{DN} \\
    c\, \LL^{(1)}
  \end{bmatrix}\,,
\end{displaymath}
where $\Delta_{DN}$ denotes the Dirichlet-Neumann Laplacian whose domain
is given by
\begin{displaymath}
  D(\Delta_{DN}) = \{ v\in H^2(\Omega,\C) \,|\, v|_{\partial\Omega_c}=0 \; ;\;
  \partial v/\partial\nu|_{\partial\Omega_i}=0\,\} \,. 
\end{displaymath}
Set further $u=(\psi,A_1)$ and $
  \FF =    ( F_1,F_2)$. \\
The system \eqref{eq:1v3} can the be represented in the form
\begin{displaymath}
  \frac{\partial u}{\partial t} + \A u = \FF \,.
\end{displaymath}
Since both $-\Delta_{DN}$, by Proposition \ref{propoA3}, and $\LL^{(1)}$,
by Proposition \ref{propoA4} are positive self-adjoint operators, and
hence sectorial, it follows that $\A:D(\Delta_{DN})\times D(\LL^{(1)})$ is
sectorial as well. Furthermore it can be easily verified (cf.
\cite{feetal98}) that $\FF: W^{1+\alpha,2}(\Omega,\C)\times W^{1+\alpha,2}(\Omega,\R^2)\to
L^2(\Omega,\C)\times L^2(\Omega,\R^2)$ is Lipschitz continuous. We can therefore
apply Theorem 3.3.3 in \cite{he81} to establish the stated existence 
in $(0,T)$ for some $T>0$.

{\em Global existence:}\\
 We begin by establishing \eqref{eq:2.2a} for
all $t>0$. We use here the fact that
\begin{displaymath}
  \Delta(|\psi|^2-1) - \frac{\partial(|\psi|^2-1)}{\partial
    t}-2|\psi|^2(|\psi|^2-1)=2|\nabla_A\psi|^2 \geq 0 \,,
\end{displaymath}
By (\ref{eq:1}c,d), \eqref{eq:2.2}, and the maximum principle (cf.
Theorem 3.7 in \cite{prwe67}) a non-negative maximum cannot exist for
all $t>0$.

Next we take the inner product of (\ref{eq:1v3}a) with $\psi$ to obtain
that
\begin{displaymath}
      \frac{1}{2}  \frac{d\, \|\psi(t,\cdot)\|_2^2}{dt}  + \|\nabla_A\psi (t,\cdot)\|_2^2 \leq  \|\psi (t,\cdot)\|_2^2 \,.
\end{displaymath}
With the aid of the above and \eqref{eq:2.2a} we then obtain 
\begin{displaymath}
  \|F_2\|_{L^2(0,T,L^2(\Omega,\R^2)}^2 \leq C \Big[\frac{1}{\kappa^2}
  \|\psi\|_{L^2(0,T;L^2(\Omega,\C))}^2 +
      \|\psi(0,\cdot )\|_2^2\Big]\,.
\end{displaymath}
We can now use Proposition \ref{abstractproposition} with $\A=-c\LL^{(1)}$,
on the interval $(0,T)$, to obtain, 
that
\begin{multline}
\label{eq:28}
\Big\|\frac{\partial \widehat A_1}{\partial t}\Big\|_{L^2(0,T; L^2(\Omega,\mathbb R^2))} +
\|\widehat A_1\|_{L^2(0,T; H^2(\Omega,\mathbb R^2))} + \|\widehat A_1\|_{L^\infty(0,T; H^1(\Omega,\R^2))}^2
        \\  \leq C   \Big[\frac{1}{\kappa^2}  \|\psi\|_{L^2(0,T;L^2(\Omega,\C))}^2+\|\widehat A_1(0,\cdot )\|_{1,2}^2 +
      \frac{1}{\kappa^2 } \|\psi(0,\cdot )\|_2^2\Big]\,.
\end{multline}
With the aid of \eqref{eq:2.2a} we then obtain that
\begin{displaymath}
  \Big\|\frac{\partial \widehat A_1}{\partial t}\Big\|_{L^2(0,T; L^2(\Omega))} + \|\widehat A_1\|_{L^\infty(0,T; H^1(\Omega,\R^2))}^2 \leq C(1+T) \,.
\end{displaymath}

It can be easily verified that
\begin{multline*}
   \frac{d}{dt}\|\nabla_A\psi(t,\cdot)\|_2^2
   + \Big\|\frac{\partial\psi}{\partial t}\Big\|_2^2 =
   \frac{d}{dt}\Big(\frac{1}{2}\|\psi(t,\cdot)\|_2^2-
 \frac{1}{4}\|\psi(t,\cdot)\|_4^4\Big)\\ - \Big\langle i\frac{\partial \widehat A_1}{\partial t}\psi,
  \nabla_A\psi \Big\rangle +\Big\langle \frac{\partial\psi }{\partial t},
  ic\psi \Div A\Big\rangle \,.
\end{multline*}
From the above we deduce that
\begin{equation}
\label{eq:6}
  \|\nabla_A\psi(t,\cdot)\|_2^2 + \Big\|\frac{\partial \psi}{\partial t}\Big\|_{L^2(0,T;
    L^2(\Omega))} \leq C(1+T) \,.
\end{equation}
% Next, we apply  Proposition \ref{abstractproposition} with
% $\A=-\Delta_{DN}$ on the interval $[0,T]$ to obtain
% \begin{multline*}
%   \|\psi\|_{L^\infty(0,T; H^1(\Omega,\R^2))}^2 \leq C(1+T^2) \,.
% \end{multline*}

We have thus obtained uniform boundedness in $[0,T]$ of the $H^1$ norm
of $\psi$ and $\widehat A_1$, and an $L^2(0,T)$-bound for
$\|\psi_t(t,\cdot)\|_2$ and $\|A_t(t,\cdot)\|_2$. \\
We need yet to obtain boundedness in $C([0,T];W^{1+\alpha,2}(\Omega))$ of both
objects. To this end we use the smoothing property of the semigroup
$e^{-t\A}$. Since $\sigma(\A)$ is discrete and lies on the positive real
axis, there exists $\delta>0$ such that for any $\lambda\in\sigma(\A)$ we have $
\lambda>\delta$.  We may thus use Theorem 1.4.3 in \cite{he81}, or more
straightforwardly the spectral theorem for self-adjoint operators, to
establish that for any $\alpha>0$ there exists $C_\alpha\in\R_+$ such that
\begin{displaymath}
  \|\A^\alpha e^{-t\A}\|\leq C_\alpha\,  t^{-\alpha} e^{-\delta t} \,,\quad \forall t>0\,.
\end{displaymath}
It follows that there exists $C(\Omega,\alpha)$ s.t. for all $0<\alpha<1$,
\begin{multline*}
  \|\widehat A_1(t,\cdot)\|_{1+\alpha,2}\leq C\|\A^{(1+\alpha)/2}\widehat A_1\|_2 \\ \leq
  C\, t^{-\alpha}e^{-\delta t}\|\widehat A_0-A_{n,d}\|_{1+\alpha,2} + C\, \Big\|\int_0^t
  \A^{(1+\alpha)/2}e^{-(t-\tau)\A}F_1(\tau,\cdot)\,d\tau\Big\|_2\,.
\end{multline*}
Here we observe from \eqref{regand} and the Sobolev injection theorem
that
\begin{equation}
A_{n,d} \in W^{1+\alpha,2}(\Omega,\mathbb R^2)\,,\, \forall \alpha \in (\frac 12,1)\,.
\end{equation}
For the last term on the right-hand-side we have
\begin{multline*}
  \Big\|\int_0^t
  \A^{(1+\alpha)/2}e^{-(t-\tau)\A}F_2(\tau,\cdot)\,d\tau\Big\|_2 =\Big\|\int_0^t
  \A^{(1+\alpha)/2}e^{-\tau\A}F_2(t-\tau)\,d\tau\Big\|\\\leq C\int_0^t
  \tau^{-(1+\alpha)/2}e^{-\delta\tau}  \|F_2(t-\tau,\cdot)\|_2\,d\tau\leq  
C\int_0^t \tau^{-(1+\alpha)/2} \|F_2(t-\tau,\cdot)\|_2\,d\tau\\\leq C\, t^{(1-\alpha)/2} \|F_2\|_{L^\infty(0,T;L^2(\Omega,\R^2))} \,,
\end{multline*}
where we have used the fact that $\alpha <1$.

By \eqref{eq:6} and \eqref{eq:2.2a} we have that
\begin{displaymath}
   \|F_2\|_{L^\infty(0,T;L^2(\Omega, \R^2))} \leq C(1+T)\,.
\end{displaymath}
Hence,
\begin{displaymath}
   \|\widehat A_1(T,\cdot)\|_{1+\alpha,2} \leq C\, (1+T^2) \,.
\end{displaymath}
In a similar manner we establish first that
\begin{displaymath}
  \|\psi(t,\cdot)\|_{1+\alpha,2}\leq C\big[t^{-\alpha}e^{-\delta t}\|\psi_0\|_{1+\alpha,2} + t^{(1-\alpha)/2} \|F_1\|_{L^\infty(0,T;L^2(\Omega,\C))}\big] \,.
\end{displaymath}
By the continuous embedding of $W^{1+\alpha,2}$ into $L^\infty$ and
\eqref{eq:6} we have that
\begin{displaymath}
   \|F_1\|_{L^\infty(0,T;L^2(\Omega,\C))} \leq C\,(1+T^5)\,,
\end{displaymath}
and hence 
\begin{displaymath}
   \|\psi(T,\cdot)\|_{1+\alpha,2} \leq C\, (1+T^6) \,.
\end{displaymath}
This completes the proof of global existence. Note that by
\eqref{eq:6} and \eqref{eq:28} we obtain that both $\psi$ and $\widehat
A_1$ are
in $H^1(0,T;L^2(\Omega))$ for any $T>0$.\\
Finally, as $\widehat A_1 \in L^2_{loc} ([0,+\infty), H^2(\Omega,\mathbb R^2))$
then $\Div \widehat{A}_1 \in L^2_{loc}([0,+\infty);H^1(\Omega))$. Furthermore,
by \eqref{regdiv} $\Div A_{n,d}$ is in $H^1(\Omega)$. Consequently, we
obtain that $\Div A\in L^2_{loc}([0,+\infty);H^1(\Omega))$. Furthermore, since
by \eqref{eq:2.2a}, we have for any $T>0$, 
\begin{multline*}
  \|F_1\|_{L^2((0,T)\times\Omega)} = C(\|\Div A\|_{L^2((0,T)\times\Omega)} + \|\nabla_A\,
  \psi\|_{L^\infty(0,T;L^2(\Omega))} \|A\|_{L^2(0,T,L^\infty(\Omega))}  \\+
  \|A\|_{L^2(0,T,L^\infty(\Omega))}^2 +1)  \,,
\end{multline*}
it follows by \eqref{eq:147abstract} that $\psi\in L^2_{loc} ([0,+\infty), H^2(\Omega,\C))$.
\end{proof}

\paragraph{\bf Acknowledgements}
Y. Almog was supported by NSF grant DMS-1109030 and by US-Israel BSF grant
no.~2010194. He also wishes to thank the Technion's department of
Mathematics and his host Professor Itai Shafrir for supporting his
Sabbatical stay there during the work on this manuscript. B. Helffer was supported by the ANR programme NOSEVOL and thanks
the Technion's department of Mathematics and Professor Itai Shafrir for
hosting his visit there. Both authors gratefully acknowledge Monique
Dauge for introducing to them some important aspects of the theory of
elliptic regularity in domains with corners. Professor X. B. Pan is
acknowledged for some preliminary discussions.

%\bibliography{surfk}

\begin{thebibliography}{10}
    \bibitem{al08}
    {\sc Y.~Almog}, {\em The stability of the normal state of superconductors in
      the presence of electric currents}, SIAM Journal on Mathematical Analysis, 40
      (2008), pp.~824--850.

    \bibitem{aletal10}
    {\sc Y.~Almog, B.~Helffer, and X.-B. Pan}, {\em Superconductivity near the
      normal state under the action of electric currents and induced magnetic
      fields in {$\mathbb R^2$}}, Comm. Math. Phys., 300 (2010), pp.~147--184.

 \bibitem{aletal12}
\leavevmode\vrule height 2pt depth -1.6pt width 23pt, {\em Superconductivity
  near the normal state in a half-plane under the action of a perpendicular
  electric current and an induced magnetic field, {P}art {II}: {T}he large
  conductivity limit}, SIAM J. Math. Anal., 44 (2012), pp.~3671--3733.

\bibitem{aletal13}
\leavevmode\vrule height 2pt depth -1.6pt width 23pt, {\em Superconductivity
  near the normal state in a half-plane under the action of a perpendicular
  electric current and an induced magnetic field}, Trans. Amer. Math. Soc., 365
  (2013), pp.~1183--1217.

\bibitem{bejo12}
{\sc C.~M. Bender and H.~F. Jones}, {\em {WKB analysis of PT-symmetric
  Sturm-Liouville problems}}, {Journal of Physics A-mathematical and theoretical}, {45} ({2012}).

   \bibitem{boda06}
  {\sc V.~Bonnaillie-No{\"e}l and M.~Dauge}, {\em Asymptotics for the low-lying
    eigenstates of the {S}chr\"odinger operator with magnetic field near
    corners}, Ann. Henri Poincar\'e, 7 (2006), pp.~899--931.

  \bibitem{chenetal93}
    {\sc Z.~M. Chen, K.-H. Hoffmann, and J.~Liang}, {\em On a nonstationary
      {G}inzburg-{L}andau superconductivity model}, Math. Methods Appl. Sci., 16
      (1993), pp.~855--875.

    \bibitem{da07}
    {\sc E.~B. Davies}, {\em Linear operators and their spectra}, vol.~106 of
      Cambridge Studies in Advanced Mathematics, Cambridge University Press,
      Cambridge, 2007.
  
      \bibitem{DaLi} {\sc R. Dautray and J.L. Lions}
      \newblock {\em Mathematical Analysis and Numerical Methods for Science and Technology.}
      \newblock Vol. 5, Evolution problems I, Springer-Verlag, 1988.

    \bibitem{du94}
    {\sc Q.~Du}, {\em Global existence and uniqueness of solutions of the
      time-dependent {G}inzburg-{L}andau model for superconductivity}, Appl. Anal.,
      53 (1994), pp.~1--17.

    \bibitem{dugr96}
    {\sc Q.~Du and P.~Gray}, {\em High-kappa limits of the time-dependent
      {G}inzburg-{L}andau model}, SIAM J. Appl. Math., 56 (1996), pp.~1060--1093.
  
      \bibitem{EnNa} {\sc K.J.~Engel and R.~Nagel}
      {\em One-parameter semigroups for linear evolution equations.}
      Graduate texts in Mathematics 194, Springer.

    \bibitem{ev98}
    {\sc L.~C. Evans}, {\em Partial Differential Equations}, AMS, 1st.~ed., 1998.

    \bibitem{feta01}
    {\sc E.~Feireisl and P.~Tak\'{a}\v{c}}, {\em Long-time stabilization of
      solutions to the {G}inzburg-{L}andau equations of superconductivity},
      Monatsh. Math., 133 (2001), pp.~197--221.

    \bibitem{feetal98}
    {\sc J.~Fleckinger-Pell\'e, H.~G. Kaper, and P.~Tak\'a{\v{c}}}, {\em Dynamics of
      the {G}inzburg-{L}andau equations of superconductivity}, Nonlinear Anal., 32
      (1998), pp.~647--665.

    \bibitem{fohe09}
    {\sc S.~Fournais and B.~Helffer}, {\em Spectral Methods in Surface
      Superconductivity}, Birkh\"auser, 2009.

    \bibitem{giph99}
    {\sc T.~Giorgi and D.~Philips}, {\em The breakdown of superconductivity due to
      strong fields for the {G}inzburg-{L}andau model}, SIAM J. Math. Anal., 30
      (1999), pp.~341--359.

    \bibitem{gitr01}
    {\sc D.~Gilbarg and N.~S. Trudinger}, {\em Elliptic partial differential
      equations of second order}, Classics in Mathematics, Springer-Verlag, Berlin,
      2001.
    \newblock Reprint of the 1998 edition.


    \bibitem{gira79}
    {\sc V.~Girault and P.-A. Raviart}, {\em Finite element approximation of the
      {N}avier-{S}tokes equations}, vol.~749 of Lecture Notes in Mathematics,
      Springer-Verlag, Berlin, Extended version 1986.



    \bibitem{gr85}
    {\sc P.~Grisvard}, {\em Elliptic problems in nonsmooth domains}, vol.~24 of
      Monographs and Studies in Mathematics, Pitman (Advanced Publishing Program),
      Boston, MA, 1985.
  
      \bibitem{gr92} {\sc P.~Grisvard}, {\em Singularities in boundary value problems}. Springer (1992).

    \bibitem{ha02}
    {\sc P.~Hartman}, {\em Ordinary differential equations}, SIAM, 2002.

    \bibitem{HeSj} {\sc B.~Helffer and J.~Sj\"ostrand}, 
     {\em From resolvent bounds to semigroup bounds}.
     Preprint : arXiv:1001.4171v1, (2010).

   \bibitem{he81}
   {\sc D.~Henry}, {\em Geometric theory of semilinear parabolic equations},
     vol.~840 of Lecture Notes in Mathematics, Springer-Verlag, Berlin, 1981.

    \bibitem{ivko84}
    {\sc B.~I. Ivlev and N.~B. Kopnin}, {\em Electric currents and resistive states
      in thin superconductors}, Advances in Physics, 33 (1984), pp.~47--114.
  
 \bibitem{ka80}
 {\sc T.~Kato}, {\em Perturbation theory for linear operators}, Springer,
   3rd.~ed., 1980.

        \bibitem{kon67} {\sc  V.A. Kondratiev,}
     {\em  Boundary Value Problems for elliptic equations in domain with conical or angular points.}
      Trans. Moscow Math Soc (1967), pp.~227-313 .



    \bibitem{lupa99}
    {\sc K.~Lu and X.-B. Pan}, {\em Estimates of the upper critical field for the
      {G}inzburg-{L}andau equations of superconductivity}, Phys. D, 127 (1999),
      pp.~73--104.
   
      \bibitem{mapla} {\sc V.G.~Mazya and V.A. Plamenevskii,}
      {\em $L^p$ estimates of solutions of elliptic boundary problems in domains with edges.}
      Transactions of the Moscow Mathematical Society (1980) Issue 1, pp.~49-97.


    \bibitem{mo95}
    {\sc R.~Montgomery}, {\em Hearing the zero locus of a magnetic field}, Comm.
      Math. Phys., 168 (1995), pp.~651--675.

    \bibitem{pakw02}
    {\sc X.-B. Pan and K.-H. Kwek}, {\em Schr\"odinger operators with
      non-degenerately vanishing magnetic fields in bounded domains}, Trans. Amer.
      Math. Soc., 354 (2002), pp.~4201--4227 (electronic).

\bibitem{peetal13}
{\sc L.~Peres-Hari, J.~Rubinstein, and P.~Sternberg}, {\em Kinematic and dynamic
  vortices in a thin film driven by an applied current and magnetic field}.
\newblock Accepted for publication in Physica D.

    \bibitem{prwe67} {\sc M.~H. Protter and H.~F. Weinberger}, {\em
        Maximum principles in differential equations}, Prentice-Hall,
      1967.

\bibitem{ruetal07}
{\sc J.~Rubinstein, P.~Sternberg, and Q.~Ma}, {\em {Bifurcation diagram and
  pattern formation of phase slip centers in superconducting wires driven with
  electric currents}}, Physical Review Letters, {99} ({2007}).

    \bibitem{ruetal10b}
    {\sc J.~Rubinstein, P.~Sternberg, and J.~Kim}, {\em On the behavior of a
      superconducting wire subjected to a constant voltage difference}, SIAM
      Journal on Applied Mathematics, 70 (2010), pp.~1739--1760.

    \bibitem{ruetal10}
    {\sc J.~Rubinstein, P.~Sternberg, and K.~Zumbrun}, {\em {The resistive state in
      a superconducting wire: bifurcation from the normal state}}, {Arch. 
      Ration. Mech. Anal.}, {195} ({2010}), pp.~{117--158}.

    \bibitem{seti11}
{\sc S.~Serfaty and I.~Tice}, {\em Ginzburg-{L}andau vortex dynamics with
  pinning and strong applied currents}, Arch. Ration. Mech. Anal., 201 (2011),
  pp.~413--464.

\bibitem{ti10}
    {\sc I.~Tice}, {\em Ginzburg-{L}andau vortex dynamics driven by an applied
      boundary current}, Comm. Pure Appl. Math., 63 (2010), pp.~1622--1676.
    \end{thebibliography}

\end{document}